\documentclass{article}
\usepackage[utf8]{inputenc}
\usepackage{xcolor}

\definecolor{Gred}{RGB}{219, 50, 54}
\definecolor{Ggreen}{RGB}{60, 186, 84}
\definecolor{Gblue}{RGB}{72, 133, 237}
\definecolor{Gyellow}{RGB}{247, 178, 16}
\definecolor{ToCgreen}{RGB}{0, 128, 0}
\definecolor{myGold}{RGB}{231,141,20}
\definecolor{myBlue}{rgb}{0.19,0.41,.65}
\definecolor{myPurple}{RGB}{175,0,124}
\definecolor{Sepia}{RGB}{44,26,8}

\usepackage{subfig}
\usepackage{hyperref}
\hypersetup{
			colorlinks=true,
	citecolor=ToCgreen,
	linkcolor=Sepia,
	filecolor=Gred,
	urlcolor=Gred
	}

\title{Settling the Robust Learnability of Mixtures of Gaussians}
\author{Allen Liu \thanks{Email: \texttt{cliu568@mit.edu} This work was supported in part by Ankur Moitra's ONR Young Investigator Award.}\and Ankur Moitra \thanks{Email: \texttt{moitra@mit.edu} This work was
supported in part by a Microsoft Trustworthy AI Grant, NSF CAREER Award CCF-1453261, NSF Large CCF1565235, a David and Lucile Packard Fellowship and an ONR Young Investigator
Award.}}
\input{preamble.sty}

\begin{document}

\maketitle

\begin{abstract}
This work represents a natural coalescence of two important lines of work \--- learning mixtures of Gaussians and algorithmic robust statistics. In particular, we give the first provably robust algorithm for learning mixtures of any constant number of Gaussians. We require only mild assumptions on the mixing weights and that the total variation distance between components is bounded away from zero. At the heart of our algorithm is a new method for proving a type of dimension-independent polynomial identifiability \--- which we call robust identifiability \--- through applying a carefully chosen sequence of differential operations to certain generating functions that not only encode the parameters we would like to learn but also the system of polynomial equations we would like to solve. We show how the symbolic identities we derive can be directly used to analyze a natural sum-of-squares relaxation.
\end{abstract}

\thispagestyle{empty}

\newpage

\setcounter{page}{1}

\section{Introduction}


This work represents a natural coalescence of two important lines of work \--- learning mixtures of Gaussians and algorithmic robust statistics \--- that we describe next: In 1894 Karl Pearson \cite{Pea94} introduced mixture models and asked: 

\begin{quote}
{\em Is there a statistically efficient method for learning the parameters of a mixture of Gaussians from samples?}
\end{quote}

\noindent Mixtures of Gaussians are natural models for when data is believed to come from two or more heterogenous sources. Since then they have found a wide variety of applications spanning statistics, biology, physics and computer science. The textbook approach for fitting the parameters is to use the maximum likelihood estimator. However it is not clear how many samples it requires to estimate the parameters up to some desired accuracy. Even worse it is hard to compute in high-dimensions \cite{sanjeev2001learning}. 

In a seminal work, Sanjoy Dasgupta \cite{Das99} introduced the problem to theoretical computer science and asked: 
\begin{quote}
    {\em Is there an efficient algorithm for learning the parameters?}
\end{quote}

\noindent Many early works were based on clustering the samples into which component generated them \cite{dasgupta2013two, sanjeev2001learning, VW04, AM05, brubaker2008isotropic, KK10}. However when the components overlap non-trivially this is no longer possible. Nevertheless Kalai, Moitra and Valiant \cite{KMV10} gave an algorithm for learning the parameters of a mixture of two Gaussians that works even if the components are almost entirely overlapping. Their approach was based on reducing the high-dimensional problem to a series of one-dimensional problems and exploiting the structure of the moments. In particular they proved that every mixture of two Gaussians is uniquely determined by its first six moments. Subsequently Moitra and Valiant \cite{MV10} and Belkin and Sinha \cite{BS10} were able to give an algorithm for learning the parameters of a mixture of any constant number of Gaussians. These algorithms crucially made use of even higher moments along with several new ingredients like methods for separating out submixtures that are flat along some directions and difficult to directly learn. There are also approaches based on tensor decompositions \cite{HK13, bhaskara2014smoothed, ge2015learning} that get polynomial dependence on the number of components assuming that the parameters are non-degenerate and subject to some kind of smoothing. However all of these algorithms break down when the data does not exactly come from a mixture of Gaussians. In fact in Karl Pearson's original application \cite{Pea94}, and in many others, mixtures of Gaussians are only intended as an approximation to the true data generating process. 

The field of robust statistics was launched by the seminal works of John Tukey \cite{tukey1960survey, tukey1975mathematics} and Peter Huber \cite{huber1964robust} and seeks to address this kind of shortcoming by designing estimators that are provably robust to some fraction of their data being adversarially corrupted. The field had many successes and explicated some of the general principles behind what makes estimators robust \cite{huber2004robust, hampel2011robust}. Provably robust estimators were discovered for fundamental tasks such as estimating the mean and covariance of a distribution and for linear regression. There are a variety of types of robustness guarantees but the crucial point is that these estimators can all tolerate a constant fraction of corruptions that is {\em independent of the dimension}. However all of these estimators turn out to be hard to compute in high-dimensions \cite{johnson1978densest, bernholt2006robust, hardt2013algorithms}. 

Recently Diakonikolas et al. \cite{diakonikolas2019robust} and Lai et al. \cite{lai2016agnostic} designed the first provably robust and computationally efficient estimators for the mean and covariance. They operate under some kind of assumption on the uncorrupted data \--- either that they come from a simple generative model like a single Gaussian or that they have bounded moments. To put this in perspective, without corruptions this is a trivial learning task because if you want to learn the mean and covariance for any distribution with bounded moments you can just use the empirical mean and empirical covariance respectively. Algorithmic robust statistics has transformed into a highly active area \cite{diakonikolas2017being, charikar2017learning, balakrishnan2017computationally, klivans2018efficient, diakonikolas2019sever, hopkins2018mixture, kothari2018robust, li2018principled, steinhardt2018robust, diakonikolas2019recent, bakshi2020robust, chen2020online} with many successes. Since then, a much sought-after goal has been to answer the following challenge:
\begin{quote}
{\em Is there a provably robust and computationally efficient algorithm for learning mixtures of Gaussians? Can we robustify the existing learning results?}
\end{quote}

There has been steady progress on this problem. Diakonikolas et al. \cite{diakonikolas2019robust} gave a robust algorithm for learning mixtures of spherical Gaussians. In recent breakthroughs Bakshi and Kothari \cite{bakshi2020outlier} and Diakonikolas et al. \cite{diakonikolas2020robustly} gave a robust algorithm for learning clusterable mixtures of Gaussians and building on this Kane \cite{kane2020robust} gave a robust algorithm for learning mixtures of two Gaussians. We note that these works do place some mild restrictions on the mixing weights and the variances. In particular they need the mixing weights to have bounded fractionality and the variances of all components to be nonzero in all directions.

The algorithms of Bakshi and Kothari \cite{bakshi2020outlier} and Diakonikolas et al. \cite{diakonikolas2020robustly} rely on the powerful sum-of-squares hierarchy \cite{parrilo2000structured}. One view is that it finds an operator, called the pseudo-expectation, that maps low degree polynomials to real values. Moreover a large number of consistency constraints are imposed that force it to in some vague sense behave like taking the expectation over a distribution on assignments to the variables. It gives a natural way to incorporate systems of polynomial constraints into a relaxation which can model complex primitives like selecting a large subset of the samples and enforcing that they have approximately the same types of moment bounds that hold for a single Gaussian. Of course the real challenge is that you need some way to reason about the pseudo-expectation operator that only uses certain types of allowable steps that can be derived through the constraints that you enforced in the relaxation. 

\subsection{Key Challenges}

It is believed that the sum-of-squares hierarchy might actually be the key to solving the more general problem of robustly learning a mixture of any constant number of Gaussians. However there are some obstacles that need to be overcome:

\paragraph{Robust Identifiability}

Behind every polynomial time algorithm for learning the parameters of a mixture model is an argument for why there cannot be two mixtures with noticeably different parameters that produce almost the same distribution. In fact we need quantitative bounds that say any two mixtures that are $\epsilon$-different must be at least $\mbox{poly}(\epsilon, 1/d)$ different according to some natural family of test functions, usually the set of low degree polynomials. Here $d$ is the dimension. This is called {\em polynomial identifiability} \cite{Tei61, moitra2018algorithmic}. Because we allow a polynomial dependence on $1/d$, it often does not matter too much how we measure the differences between two mixtures, either in terms of some natural parameter distance between their components or in terms of the total variation distance, again between their components. 

However we need much stronger bounds when it comes to robust learning problems where we want to be able to tolerate a constant fraction of corruptions that is dimension independent. In particular we need a richer family of test functions with the property that whenever we have two mixtures whose components are $\epsilon$-different in total variation distance there is some function in the family that has at least $\mbox{poly}(\epsilon)$ discrepancy when we take the difference between its expectations with respect to the two distributions (and its variance must also be bounded). In particular, this relationship cannot involve dependence on the dimension $d$.  We will call this {\em robust identifiability}. Recall that the non-robust learning algorithms for mixtures of Gaussians reduce to a series of one-dimensional problems. Unfortunately this strategy inherently introduces polynomial factors in $d$ and it cannot give what we are after. For the special case of clusterable mixtures of Gaussians, Bakshi and Kothari \cite{bakshi2020outlier} and Diakonikolas et al. \cite{diakonikolas2020robustly} proved robust identifiability and their approach was based on classifying the ways in which two single Gaussians can have total variation distance close to one. When it comes to the more general problem of handling mixtures where the components can overlap non-trivially it seems difficult to follow the same route because we can no longer match components from the two mixtures to each other and almost cancel them both out. 

\paragraph{Reasoning About the Sum-of-Squares Relaxation}

The sum-of-squares hierarchy is a general framework for coming up with large and powerful semidefinite programming relaxations that can be applied to many sorts of problems \cite{barak2014rounding, barak2015dictionary, hopkins2015tensor, barak2016noisy}. However it is often quite challenging to understand whether or not it works and/or to identify, out of all the constraints that are enforced on the pseudo-expectation, which ones are actually useful in the analysis \cite{hopkins2016fast, hopkins2017power}. What makes matters especially challenging in our setting is that it is clear the structure of the higher moments of a mixture of Gaussians must play a major role. But how exactly do we leverage them in our analysis? 

\subsection{Our Techniques and Main Result}

Actually, we overcome both obstacles using the same approach. We store the relevant moments and variables we would like to solve for in certain generating functions. Then by manipulating the generating functions using differential operators, we are able to reason about an SOS relaxation of a natural polynomial system that allows us to solve for the parameters that we want to learn. 

Let us describe the setup.  For simplicity, assume that the mixture is in isotropic position.  First, we have the unknown parameters of the mixture.  Let 
\[
\mcl{M} = w_1 N(\mu_1, I + \Sigma_1) + \dots + w_k N(\mu_k, I + \Sigma_k)
\]
where $w_i$ are the mixing weights and $\mu_i$ and $I + \Sigma_i$ denote the mean and covariance of the $i$th component. Second, we have the indeterminates we would like to solve for.  These will be denoted $\widetilde{\mu_i}$ and $\widetilde{\Sigma_i}$ and our intention is for these to be the means and covariances of a hypothetical mixture of Gaussians \footnote{Technically our setup is slightly different in that we solve for vectors $u_1, \dots , u_k$ and $v_1, \dots , v_k$ that are supposed to form an orthonormal basis for the span of the $\{ \wt{\mu_i} \}$ and $\{ \wt{\Sigma_i} \}$ respectively.  Regardless, the argument is conceptually the same.}.  We will also guess the mixing weights of the hypothetical mixture $\wt{w_i}$.  Finally, we have a $d$-dimensional vector $X = (X_1, \dots , X_d)$ of formal variables and one auxiliary formal variable $y$.  These will mostly be used to help us organize everything in a convenient way.  Roughly, we would like to solve for  $\widetilde{\mu_i}$ and $\widetilde{\Sigma_i}$ so that the hypothetical mixture 
\[
\wt{\mcl{M}} = \wt{w_1} N(\wt{\mu_1}, I + \wt{\Sigma_1}) + \dots + \wt{w_k} N(\wt{\mu_k}, \wt{I + \Sigma_k})
\]
is close to $\mcl{M}$ on the family of test functions (which will be low-degree multivariate Hermite polynomials).  It turns out that this amounts to solving a polynomial system for the indeterminates.   
\\\\
Now we explain in more detail how to actually reason about and solve the polynomial system.  It will be useful to work with the following generating functions. First let
\[
F(y) = \sum_{i = 1}^k w_i e^{\mu_i(X)y + \frac{1}{2}\Sigma_i(X)y^2}
\]
Here we have used the notation that $\mu_i(X)$ denotes the inner product of $\mu_i$ with the $d$-dimensional vector $X$ and that $\Sigma_i(X)$ denotes the quadratic form of $X$ on $\Sigma_i$. Second let
\[
\widetilde{F}(y) = \sum_{i = 1}^k \widetilde{w_i} e^{\widetilde{\mu_i}(X)y + \frac{1}{2}\widetilde{\Sigma_i}(X)y^2}
\]
As is familiar from elementary combinatorics we can tease out important properties of the generating function by applying carefully chosen operators that involve differentiation. This requires a lot more bookkeeping than usual because there are unknown parameters of the mixture, indeterminates and formal variables. But it turns out that there are simple differential operators we can apply which can isolate components. 
To gain some intuition, consider the operator
\[
\mcl{D}_i = \partial_y - (\mu_i(X) + \Sigma_i(X)y) \,.
\]
Note that 
\[
\mcl{D}_i\left( e^{\mu_i(X)y + \frac{1}{2}\Sigma_i(X)y^2} \right) = 0 \,,
\]
in other words, this operator annihilates the $i$\ts{th} component.  Thus, by composing such operators, we can annihilate all but one of the components in $F$ \footnote{There are some additional details because applying $\mcl{D}_i$ to a different component creates an extra polynomial factor in front.  Section \ref{sec:alg-identities} shows how to deal with this complication.}.  On the other hand, note that applying differential operators is just a rearrangement of the polynomials that show up in the infinite sum representation of the generating function but using differential operators and generating functions in exponential form gives a particularly convenient way to derive useful expressions that would otherwise be extremely complex to write down. 

Ultimately we derive a symbolic identity
\begin{equation}\label{eq:intro}
\widetilde{w_k}\prod_{i=1}^k (\widetilde{\Sigma_k}(X) - \Sigma_i(X) )^{2^{i-1}}\prod_{i=1}^{k-1}(\widetilde{\Sigma_k}(X) - \widetilde{\Sigma_i}(X) )^{2^{k+i-1}} = \sum_{i = 1}^m P_i(X)(\widetilde{h_i}(X) - h_i(X)) 
\end{equation}
where $m$ is a function of $k$.  In the above, the $h_i(X)$'s are the expectations of the multivariate Hermite polynomials for the true mixture $\mcl{M}$  and the $\widetilde{h_i}$'s are the expectations of the multivariate Hermite polynomials for the hypothetical mixture $\wt{\mcl{M}}$.  A more detailed explanation of Hermite polynomials is given in Section \ref{sec:gen-functions} but for now we may simply think of them as modified moments.  The reason that we use Hermite polynomials instead of standard moments is that they can be robustly estimated without losing dimension-dependent factors (see e.g. \cite{kane2020robust}).

The above identity (when combined with a few others of a similar form) allows us to deduce robust identifiability.  Roughly, this is because if we have a mixture $\wt{\mcl{M}}$ with means $\widetilde{\mu_i}$ and covariances  $I + \widetilde{\Sigma_i}$ that don't match those of $\mcl{M}$, then the LHS of (\ref{eq:intro}) will be bounded away from $0$, implying that some term on the RHS must also be bounded away from $0$.  This means that there must be some $i \leq m = O_k(1)$ such that $h_i(X)$ and $\wt{h_i(X)}$ are substantially different \--- i.e. there is some test function that is a low-degree multivariate Hermite polynomial that distinguishes the two mixtures.

Robust identifiability alone does not give us a polynomial time learning algorithm.  However, it turns out that we can use SOS to obtain a polynomial time learning algorithm from the argument for robust identifiability i.e. we essentially get the learning algorithm ``for free". The key point is that the $h_i(X)$ can be estimated using our samples and the coefficients of the $\widetilde{h_i}$ are explicit polynomials in the indeterminates that we can write down.  Also the $P_i$'s are polynomials in {\em everything}: the unknown parameters, the indeterminates and the formal variables (except for $y$).   To set up an SOS system, we obtain robust estimates $\overline{h_i}$ for the expectations of the Hermite polynomials $h_i(X)$ for the true mixture that we can compute from existing techniques in the literature.  We then enforce that the expectations of the  Hermite polynomials for the hypothetical mixture $\widetilde{h_i}(X)$  are close to these robust estimates where closeness is defined in terms of the distance between their coefficient vectors.  


It is not immediately clear why the expression in (\ref{eq:intro}) ought to be useful for solving the SOS system that we set up. After all, we cannot explicitly compute it because it depends on things we do not have, like the true parameters, the $\mu_i$'s and $\Sigma_i$'s. However the sum-of-squares relaxation enforces that the pseudo-expectation operator assigns values to polynomials in the indeterminates in a way that behaves like an actual distribution on solutions when we are evaluating certain types of low degree expressions that contain the one above. So even though we do not know the actual polynomials in the identity, they exist and the fact that they are enforced is enough to ensure that we can estimate the covariance of the $k$th component 
We stress that this is just the high-level picture and many more details are needed to fill it in. 

Using these techniques, we come to the main result of our paper, which is a polynomial-time algorithm for robustly learning the parameters of a high-dimensional mixture of a constant number of Gaussians.  Our main theorem is (informally) stated below.  A formal statement can be found in Theorem \ref{thm:parameter-learning}.

\begin{theorem}\label{thm:full-informal}
Let $k$ be a constant.  Let $\mcl{M} = w_1G_1 + \dots + w_kG_k$ be a mixture of Gaussians in $\R^d$ whose components are non-degenerate and such that the mixing weights have bounded fractionality and TV distances between different components are lower bounded. (Both of these bounds can be any function of $k$).  Given $n = \poly(d/\eps)$ samples from $\mcl{M}$ that are $\eps$-corrupted, there is an algorithm that runs in time $\poly(n)$ and with high probability outputs a set of mixing weights $\widetilde{w_1}, \dots ,\widetilde{w_k}$ and components $\widetilde{G_1}, \dots , \widetilde{G_k}$ that are $\poly(\eps)$-close to the true components (up to some permutation).
\end{theorem}

\subsubsection{Discussion of Assumptions and Later Improvements}


\paragraph{Bounded Fractionality of Mixing Weights: }  The assumption of bounded fractionality stems from an issue in previous work \cite{diakonikolas2020robustly} about learning clusterable mixtures of Gaussians i.e. when the components have essentially no overlap.  We use some of their subroutines in our algorithm for clustering the mixture into submixtures where the components are not too far apart from each other (see Section \ref{sec:clustering}).  While \cite{diakonikolas2020robustly} claims to handle general mixing weights, the analysis of the algorithm in \cite{diakonikolas2020robustly} is only done in detail for the case of uniform mixing weights and the argument in Appendix C for reducing from general mixing weights to uniform mixing weights does not work.  The authors of \cite{diakonikolas2020robustly} along with the authors of \cite{bakshi2020robust} were able to fix these arguments to handle general mixing weights in \cite{bakshi2020robustlya}.  The modified proof uses essentially the same techniques.  Plugging this improved clustering result into our analysis, instead of using the weaker guarantees for uniform mixing weights, we are able to straightforwardly remove the bounded fractionality assumption.  See Section \ref{sec:improved-clustering} for a formal statement and explanation.

\paragraph{Separation Assumption: } While our original proof required constant separation between components, we show in a follow-up paper,  \cite{liu2021learning}, that this assumption can be replaced with a separation of $\eps^{\Omega_k(1)}$ (which is a necessary assumption for parameter learning) using standard tricks (see Theorem 9.2 in \cite{liu2021learning}).  This argument works independently of the improvement for the mixing weight assumption.  See Section \ref{sec:improved-separation} for a more formal statement and explanation.

\paragraph{Non-degeneracy of Components:} This assumption was included so there are no bit complexity issues.  In fact, dealing with potentially degenerate covariance matrices requires a formal specification of the model of computation.
\\\\
\indent To make the timeline of events clear, we stated our original result above.  However, plugging in these improvements (for removing the bounded fractionality and separation assumption), we obtain an improved result, stated below.  The formal statement can be found in Theorem \ref{thm:main-improved}.  We emphasize that these modifications are \textit{completely independent} of our main  contributions, but are rather tools that we employ to reduce to the case where the components are not too far from each other. We reduce to this case by running a preprocessing step where we cluster the mixture into such submixtures.  All of the assumptions stem from the clustering step, which is done via modifications to the techniques for learning fully clusterable mixtures in \cite{diakonikolas2020robustly, bakshi2020outlier}.

\begin{theorem}\label{thm:full-new}
Let $k$ be a constant.  Let $\mcl{M} = w_1G_1 + \dots + w_kG_k$ be a mixture of Gaussians in $\R^d$ whose components are non-degenerate and such that the mixing weights are lower bounded by some function of $k$.  Also assume that the TV distances between components are at least $\eps^{\Omega_k(1)}$.  Then given $n = \poly(d/\eps)$ samples from $\mcl{M}$ that are $\eps$-corrupted, there is an algorithm that runs in time $\poly(n)$ and with high probability outputs a set of mixing weights $\widetilde{w_1}, \dots ,\widetilde{w_k}$ and components $\widetilde{G_1}, \dots , \widetilde{G_k}$ that are $\poly(\eps)$-close to the true components (up to some permutation).
\end{theorem}

\begin{remark}
The improved clustering arguments of \cite{bakshi2020robustlya} are able to get polynomial dependence on the minimum mixing weight instead of exponential dependence so they are actually able to deal with mixing weights that are $\eps^{\Omega_k(1)}$.  This improvement also plugs into our result as well.
\end{remark}

\subsection{Proof Overview}
The proof of our main theorem can be broken down into several steps.  We first present our main contribution, an algorithm for learning mixtures of Gaussians when no pair of components are too far apart.  We introduce the necessary generating function machinery in Section \ref{sec:gen-functions} and then present our algorithm in Sections \ref{sec:close-case} and \ref{sec:moment-est}.  Specifically, in Section \ref{sec:close-case} we show how to learn the parameters once we have estimates for the Hermite polynomials of the true mixture.  And in Section \ref{sec:moment-est}, we show how to robustly estimate the Hermite polynomials, using similar techniques to \cite{kane2020robust}.

To complete our full algorithm for learning general mixtures of Gaussians, we combine our aforementioned results with a clustering algorithm similar to \cite{diakonikolas2020robustly}.  Combining these algorithms, we prove that our algorithm outputs a mixture that is close to the true mixture in TV distance.  This is done in Sections \ref{sec:clustering} and \ref{sec:full-alg}.  We then prove identifiability in Section \ref{sec:identifiability}, implying that our algorithm actually learns the true parameters.

\subsection{Concurrent and Subsequent Work}

There are three main pieces of work that we discuss.  The first by Bakshi et. al \cite{bakshi2020robustlya}, was independent and concurrent to this one.  The next is a subsequent work by Bakshi et. al \cite{bakshi2020robustlyb} that improves their earlier result but also borrows techniques from our paper.  We also discuss our follow-up work \cite{liu2021learning} that was after \cite{bakshi2020robustlya} but before \cite{bakshi2020robustlyb} (and the contributions are essentially disjoint).

Bakshi et. al in \cite{bakshi2020robustlya} obtain a result that is similar to our main result Theorem \ref{thm:full-informal}, but using rather different techniques.  There are a few key differences, which we discuss below.  We learn the mixture to accuracy $\eps^{\Omega_k(1)}$ while their result only achieves accuracy $(1/ \log (1/\eps))^{\Omega_k(1)}$, an exponentially worse guarantee.  Also, our result solves parameter learning \--- i.e. we estimate the parameters of the true mixture \--- while their algorithm solves proper density estimation \--- i.e. it outputs a mixture of $k$ Gaussians that is close to the true density. Their algorithm does not need any lower bound on the mixing weights or on the pairwise separation of the components. In fact lower bounds on these quantities are necessary for parameter learning. However, our original result makes an even stronger assumption about the bounded fractionality of the mixing weights and pairwise separation.

Subsequently in \cite{bakshi2020robustlyb}, Bakshi et. al improve their earlier result to achieve parameter learning and accuracy $\eps^{\Omega_k(1)}$.   In fact the analysis of their new parameter learning algorithm relies crucially on the robust identifiability result from our paper (see Section 9 in \cite{bakshi2020robustlyb}). Thus, all known algorithms for robust parameter learning go through our machinery and robust identifiability.  Compared to our result, the main improvement in \cite{bakshi2020robustlyb} is an improved, polynomial dependence on the minimum mixing weight.

Our follow-up work \cite{liu2021learning} improves the results here in a different direction.  We are able to obtain an algorithm that solves density estimation to accuracy $\wt{O}(\eps)$ (instead of $\eps^{\Omega_k(1)}$).  Note that parameter learning to this accuracy is information-theoretically impossible.  The work in \cite{liu2021learning} does need a stronger assumption that the components have variances in all directions lower and upper bounded by a constant and the learning algorithm is improper, outputting a mixture of polynomials times Gaussians, instead of just a mixture of Gaussians.  The main relevance of \cite{liu2021learning} to this paper is that as a first step, \cite{liu2021learning} shows how to improve the separation assumption in Theorem \ref{thm:full-informal} from some constant to $\eps^{\Omega_k(1)}$.



\section{Preliminaries}

\subsection{Problem Setup}
We use $N(\mu, \Sigma)$ to denote a Gaussian with mean $\mu$ and covariance $\Sigma$.  We use $d_{\TV}(\mcl{D}, \mcl{D'})$ to denote the total variation distance between two distributions $\mcl{D}, \mcl{D'}$.  We begin by formally defining the problem that we will study.  First we define the contamination model.  This is a standard definition from robust learning (see e.g. \cite{diakonikolas2020robustly}).
\begin{definition}[Strong Contamination Model]\label{def:corruption}
We say that a set of vectors $Y_1, \dots , Y_n$ is an $\eps$-corrupted sample from a distribution $\mcl{D}$ over $\R^d$ if it is generated as follows.  First $X_1, \dots , X_n$ are sampled i.i.d. from $\mcl{D}$.  Then a (malicious, computationally unbounded) adversary observes $X_1, \dots , X_n$ and replaces up to $\eps n$ of them with any vectors it chooses.  The adversary may then reorder the vectors arbitrarily and output them as $Y_1, \dots , Y_n$
\end{definition}

\noindent In this paper, we study the following problem.  There is an unknown mixture of Gaussians 
\[
\mcl{M} = w_1G_1 + \dots + w_kG_k
\]
where $G_i = N(\mu_i, \Sigma_i)$.  We receive an $\eps$-corrupted sample $Y_1, \dots , Y_n $ from $\mcl{M}$ where $n = \poly(d/\eps)$.  The goal is to output a set of parameters $\wt{w_1}, \dots , \wt{w_k}$ and $(\wt{\mu_1}, \wt{\Sigma_1}), \dots , (\wt{\mu_k}, \wt{\Sigma_k})$ that are $\poly(\eps)$ close to the true parameters in the sense that there exists a permutation $\pi$ on $[k]$ such that for all $i$
\[
|w_i - \wt{w_{\pi(i)}}|, d_{\TV}\left(N(\mu_i, \Sigma_i) , N(\wt{\mu_{\pi(i)}}, \wt{\Sigma_{\pi(i)}})\right) \leq \poly(\eps).
\]

Throughout our paper, we will assume that all of the Gaussians that we consider have variance at least $\poly(\eps/d)$ and at most $\poly(d/\eps)$ in all directions i.e. they are not too flat.  This implies that their covariance matrices are invertible so we  may write expressions such as $\Sigma_i^{-1}$.  We will also make the following assumptions about the mixture:
\begin{itemize}
    \item The $w_i$ are rational with denominator at most $A$
    \item For all $i \neq j$, $d_{\TV}(G_i, G_j) > b$
\end{itemize}
for some positive constants $A,b$.  Note that a lower bound on the minimum mixing weight and a lower bound on the TV distance between components is necessary for parameter learning.  Throughout this paper, we treat $k,A,b$ as constants \--- i.e. $A$ and $b$ could be any function of $k$ \--- and when we say polynomial, the exponent may depend on these parameters.  We are primarily interested in dependence on $\eps$ and $d$ (the dimension of the space). 

\subsection{Sum of Squares Proofs}
We will make repeated use of the Sum of Squares (SOS) proof system.  We review a few basic facts here (see \cite{barak2016proofs} for a more extensive treatment).  Our exposition here closely mirrors \cite{diakonikolas2020robustly}.

\begin{definition}[Symbolic Polynomials]
A degree-$t$ symbolic polynomial $P$ is a collection of indeterminates $\wh{P}(\alpha)$, one for each multiset $\alpha \subseteq [n]$ of size at most $t$.  We think of it as representing a polynomial $P: \R^n \rightarrow \R$ in the sense that
\[
P(x) = \sum_{\alpha \subseteq [n], |\alpha| \leq t} \wh{P}(\alpha)x^{\alpha}
\]
\end{definition}

\begin{definition}[SOS proof]
Let $x_1, \dots , x_n$ be indeterminates and let $\mcl{A}$ be a set of polynomial inequalities
\[
\{ p_1(x) \geq 0, \dots , p_m(x) \geq 0 \}
\]
An SOS proof of an inequality $r(x) \geq 0$ from constraints $\mcl{A}$ is a set of polynomials $\{ r_S(x) \}_{S \subseteq [m]}$ such that each $r_S$ is a sum of squares of polynomials and 
\[
r(x) = \sum_{S \subseteq [m]} r_S(x) \prod_{i \in S} p_i(x)
\]
The degree of this proof is the maximum of the degrees of $r_S(x) \prod_{i \in S} p_i(x)$ over all $S$.  We write
\[
\mcl{A} \vdash_k r(x) \geq 0
\]
to denote that the constraints $\mcl{A}$ give an SOS proof of degree $k$ for the inequality $r(x) \geq 0$.  Note that we can represent equality constraints in $\mcl{A}$ by including $p(x) \geq 0$ and $-p(x) \geq 0$.
\end{definition}

\noindent The dual objects to SOS proofs are pseudoexpectations.  We will repeatedly make use of pseudoexpectations later on.
\begin{definition}\label{def:pseudoexpectation}
Let $x_1, \dots , x_n$ be indeterminates.  A degree-$k$ pseudoexpectation $\wt{\E}$ is a linear map 
\[
\wt{\E}: \R[x_1, \dots , x_n]_{\leq k} \rightarrow \R
\]
from degree-$k$ polynomials to $\R$ such that $\wt{\E}[p(x)^2] \geq 0$ for any $p$ of degree at most $k/2$ and $\wt{\E}[1] = 1$.  For a set of polynomial constraints $\mcl{A} = \{ p_1(x) \geq 0, \dots , p_m(x) \geq 0 \}$, we say that $\wt{\E}$ satisfies $\mcl{A}$ if 
\[
\wt{\E}[s^2(x)p_i(x)] \geq 0
\]
for all polynomials $s(x)$ and $i \in [m]$ such that $s(x)^2p_i(x)$ has degree at most $k$.
\end{definition}

The key fact is that given a set of polynomial constraints, we can solve for a constant-degree pseudoexpectation that satisfies those constraints (or determine that none exist) in polynomial time as it reduces to solving a polynomially sized SDP.

\begin{theorem}[SOS Algorithm \cite{barak2016proofs}]
There is an algorithm that takes a natural number $k$ and a satisfiable system of polynomial inequalities $\mcl{A}$ in varibles $x_1, \dots , x_n$ with coefficients at most $2^n$ containing an inequality of the form $\norm{x}^2 \leq M$ for some real number $M$ and returns in time $n^{O(k)}$ a degree-$k$ pseudoexpectation $\wt{\E}$ which satisfies $\mcl{A}$ up to error $2^{-n}$.
\end{theorem}
Note that there are a few technical details with regards to only being able to compute a pseudoexpectation that nearly satisfies the constraints.  These technicalities do not affect our proof (as $2^{-n}$ errors will be negligible) so we will simply assume that we can compute a pseudoexpectation that exactly satisfies the constraints.  See \cite{barak2016proofs} for more details about these technicalities.

Finally, we state a few simple inequalities for pseudoexpectations that will be used repeatedly later on.
\begin{claim}[Cauchy Schwarz for Pseudo-distributions]
Let $f,g$ be polynomials of degree at most $k$ in indeterminates $x = (x_1, \dots , x_n)$.  Then for any degree $k$ pseudoexpectation,
\[
\wt{\E}[fg] \leq \sqrt{\wt{\E}[f^2]}\sqrt{\wt{\E}[g^2]}.
\]
\end{claim}

\begin{corollary}
Let $f_1,g_1, \dots , f_m,g_m$ be polynomials of degree at most $k$ in indeterminates $x = (x_1, \dots , x_n)$.  Then for any degree $k$ pseudoexpectation,
\[
\wt{\E}[f_1g_1 + \dots + f_mg_m] \leq \sqrt{\wt{\E}[f_1^2 + \dots + f_m^2]}\sqrt{\wt{\E}[g_1^2+ \dots + g_m^2]}.
\]
\end{corollary}
\begin{proof}
Note 
\[
\wt{\E}[f_1g_1 + \dots + f_mg_m] \leq \sqrt{\wt{\E}[f_1^2]}\sqrt{\wt{\E}[g_1^2]} + \dots + \sqrt{\wt{\E}[f_m^2]}\sqrt{\wt{\E}[g_m^2]} \leq \sqrt{\wt{\E}[f_1^2 + \dots + f_m^2]}\sqrt{\wt{\E}[g_1^2+ \dots + g_m^2]}
\]
where the first inequality follows from Cauchy Schwarz for pseudoexpectations and the second follows from standard Cauchy Schwarz.
\end{proof}

\section{Fun with Generating Functions}\label{sec:gen-functions}

We now introduce the generating function machinery that we will use in our learning algorithm.  We begin with a standard definition.
\begin{definition}
Let $\mcl{H}_m(x)$ be the univariate Hermite polynomials $\mcl{H}_0 = 1, \mcl{H}_1 = x, \mcl{H}_2 = x^2 - 1 \cdots $ defined by the recurrence
\[
\mcl{H}_m(x) = x\mcl{H}_{m-1}(x) - (m-1)\mcl{H}_{m-2}(x)
\]
\end{definition}

Note that in $\mcl{H}_m(x)$, the degree of each nonzero monomials has the same parity as $m$.  In light of this, we can write the following:
\begin{definition}
Let $\mcl{H}_m(x,y^2)$ be the homogenized Hermite polynomials e.g. $\mcl{H}_2(x,y^2) = x^2 - y^2, \mcl{H}_3(x,y^2) = x^3 - 3xy^2$.
\end{definition}

It will be important to note the following fact:
\begin{claim}\label{claim:hermite-identity}
We have
\[
e^{xz - \frac{1}{2}y^2z^2} = \sum_{m = 0}^{\infty} \frac{1}{m!}\mcl{H}_m(x, y^2)z^m
\]
where the RHS is viewed as a formal power series in $z$ whose coefficients are polynomials in $x,y$.
\end{claim}

Now we define a multivariate version of the Hermite polynomials.
\begin{definition}\label{def:Hermite-two-var}
Let $H_m(X,z)$ be a formal polynomial in variables $X = X_1, \dots , X_d$ whose coefficients are polynomials in $d$ variables $z_1, \dots , z_d$ that is given by
\[
H_m(X,z) = \mcl{H}_m(z_1X_1 + \dots + z_dX_d , X_1^2 + \dots + X_d^2)
\]
Note that $H_m$ is homogeneous of degree $m$ as a polynomial in $X_1, \dots , X_d$
\end{definition}

\begin{definition}\label{def:Hermite-final}
For a distribution $D$ on $\R^d$, we let
\[
h_{m,D}(X) = \E_{(z_1, \dots , z_d) \sim D}[H_m(X,z)]
\]
where we take the expectation of $H_m$ over $(z_1, \dots , z_d)$ drawn from $D$.  Note that $h_{m,D}(X)$ is a polynomial in $(X_1, \dots , X_d)$. We will omit the $D$ in the subscript when it is clear from context. Moreover for a mixture of Gaussians $$\mcl{M} = w_1 N(\mu_1, \Sigma_1) + \dots w_k N(\mu_k, \Sigma_k)$$ we will refer to the Hermite polynomials $h_{m, \mcl{M}}$ as the Hermite polynomials of the mixture.  
\end{definition}

We remark that if there is a mixture $\mcl{M} = w_1 N(\mu_1, \Sigma_1) + \dots w_k N(\mu_k, \Sigma_k)$ where instead of real numbers, the $w_i, \mu_i, \Sigma_i$ are given in terms of indeterminates, the Hermite polynomials will be polynomials in those indeterminates.  We will repeatedly make use of this abstraction later on.

The first important observation is that the Hermite polynomials for Gaussians can be written in a simple closed form via generating functions.
\begin{claim}
Let $D = N(\mu, I + \Sigma)$.  Let $a(X) = \mu \cdot X$ and $b(X) = X^T \Sigma X$ .  Then
\[
e^{a(X)y + \frac{1}{2}b(X)y^2} = \sum_{m=0}^{\infty} \frac{1}{m!} \cdot h_{m,D}(X) y^m
\]
as formal power series in $y$.  
\end{claim}
\begin{proof}
By Claim \ref{claim:hermite-identity}, we have
\[
e^{a(X)y + \frac{1}{2}b(X)y^2} = \sum_{m=0}^{\infty}\frac{1}{m!}\mcl{H}_m(a(X), -b(X))y^m
\]
It now suffices to verify that 
\[
\E_{(z_1, \dots z_d) \sim D}\left[\mcl{H}_m(z_1X_1 + \dots + z_dX_d , X_1^2 + \dots + X_d^2)\right] = \mcl{H}_m(a(X), -b(X))
\]
This can be verified through straight-forward computations using the moment tensors of a Gaussian (see Lemma 2.7 in \cite{kane2020robust}).
\end{proof}

We now have two simple corollaries to the above.

\begin{corollary}
Let $\mcl{M} = w_1 N(\mu_1, I +  \Sigma_1) + \dots w_k N(\mu_k, I + \Sigma_k)$.  Let $a_i(X) = \mu_i \cdot X$ and $b_i(X) = X^T\Sigma_i X$.  Then
\[
\sum_{m=0}^{\infty} \frac{1}{m!} \cdot h_{m,\mcl{M}}(X) y^m = w_1e^{a_1(X)y + \frac{1}{2}b_1(X)y^2} + \dots + w_ke^{a_k(X)y + \frac{1}{2}b_k(X)y^2}
\]
\end{corollary}

\begin{corollary}
Let $\mcl{M} = w_1 N(\mu_1, I +  \Sigma_1) + \dots w_k N(\mu_k, I + \Sigma_k)$.  Let $a_i(X) = \mu_i \cdot X$ and $b_i(X) = X^T\Sigma_i X$.  Then the Hermite polynomials $h_{m,\mcl{M}}(X)$ can be written as a linear combination of products of the $a_i(X), b_i(X)$ such that the number of terms in the sum, the number of terms in each product, and the coefficients in the linear combination are all bounded as functions of $m,k$.
\end{corollary}

The next important insight is that the generating functions for the Hermite polynomials behave nicely under certain differential operators.  We can use these differential operators to derive identities that the Hermite polynomials must satisfy and these identities will be a crucial ingredient in our learning algorithm.
\\\\
The proceeding claims all follow from direct computation.
\begin{claim}\label{claim:degree-reduction}
Let $\partial $ denote the differential operator with respect to $y$.  If 
\[
f(y) = P(y,X)e^{a(X)y + \frac{1}{2}b(X)y^2}
\]
where $P$ is a polynomial in $y$ of degree $k$ (whose coefficients are polynomials in $X$) then
\[
(\partial - (a(X) + yb(X)))f(y) = Q(y,X)e^{a(X)y + \frac{1}{2}b(X)y^2}
\]
where $Q$ is a polynomial in $y$ with degree exactly $k-1$ whose leading coefficient is $k$ times the leading coefficient of $P$.
\end{claim}

\begin{corollary}\label{corollary:null-operator}
Let $\partial $ denote the differential operator with respect to $y$.  If 
\[
f(y) = P(y,X)e^{a(X)y + \frac{1}{2}b(X)y^2}
\]
where $P$ is a polynomial in $y$ of degree $k$ then
\[
(\partial - (a(X) + yb(X)))^{k+1}f(y) = 0.
\]
\end{corollary}


\begin{claim}\label{claim:leading-coeff}
Let $\partial $ denote the differential operator with respect to $y$.  Let
\[
f(y) = P(y,X)e^{a(X)y + \frac{1}{2}b(X)y^2}
\]
where $P$ is a polynomial in $y$ of degree $k$.  Let the leading coefficient of $P$ (viewed as a polynomial in $y$) be $L(X)$.  Let $c(X), d(X)$ be a linear and quadratic polynomial in the $X$ variables respectively such that $\{a(X), b(X) \} \neq \{c(X) , d(X) \}$.  If $b(X) \neq d(X)$ then
\[
(\partial - (c(X) + yd(X)))^{k'}f(y) = Q(y,X)e^{a(X)y + \frac{1}{2}b(X)y^2}
\]
where $Q$ is a polynomial of degree $k + k'$ in $y$ with leading coefficient
\[
L(x)(b(X) - d(X))^{k'} 
\]
and if $b(X) = d(X)$ then 
\[
(\partial - (c(X) + yd(X)))^{k'}f(y) = Q(y,X)e^{a(X)y + \frac{1}{2}b(X)y^2}
\]
where $Q$ is a polynomial of degree $k$ in $y$ with leading coefficient
\[
L(X)(a(X) - c(X))^{k'}
\]
\end{claim}

\subsection{Polynomial Factorizations}

The analysis of our SOS-based learning algorithm will rely on manipulations of Hermite polynomials.  An important piece of our analysis is understanding how the coefficients of polynomials behave under addition and (polynomial) multiplication.  Specifically, if we have two polynomials $f(X), g(X)$ and we have bounds on the coefficients of $f$ and $g$, we now want to give bounds on the coefficients of the polynomials $f(X) + g(X)$ and $f(X)g(X)$.  Most of these bounds are easy to obtain.  The one that is somewhat nontrivial is lower bounding the coefficients of $f(X)g(X)$ i.e. if the coefficients of $f$ and $g$ are not all small, then the product $f(X)g(X)$ cannot have all of its coefficients be too small.

\begin{definition}
For a polynomial $f(X)$ in the $d$ variables $X_1, \dots , X_d$ with real coefficients define $v(f)$ to be the vectorization of the coefficients. (We will assume this is done in a consistent manner so that the same coordinate of vectorizations of two polynomials corresponds to the coefficient of the same monomial.) We will frequently consider expressions of the form $\norm{v(f)}$ i.e. the $L^2$ norm of the coefficient vector.
\end{definition}

\begin{definition}
For a polynomial $A(X)$ of degree $k$ in $d$ variables $X_1, \dots , X_d$ and a vector $v \in \R^d$ with nonnegative integer entries summing to at most $k$, we use $A_v$ to denote the corresponding coefficient of $A$.
\end{definition}

First, we prove a simple result about the norm of the vectorization of a sum of polynomials.
\begin{claim}\label{claim:sum-bound}
Let $f_1, \dots , f_m$ be polynomials in $X_1, \dots , X_d$ whose coefficients are polynomials in formal variables $u_1, \dots , u_n$ of degree $O_k(1)$.  Then
\[
\norm{v(f_1 + \dots + f_m)}^2 \leq m(\norm{v(f_1)}^2 + \dots + \norm{v(f_m)}^2)
\]
Furthermore, the difference can be written as a sum of squares of polynomials of degree $O_k(1)$ in $u_1, \dots, u_n$. 
\end{claim}
\begin{proof}
Note
\[
(a_1 + \dots + a_m)^2 \leq m(a_1^2 + \dots + a_m^2)
\]
and the difference between the two sides can be written as a sum of squares 
\[
\sum_{i \neq j} (a_i - a_j)^2
\]
The desired inequality can now be obtained by summing expressions of the above form over all coefficients.
\end{proof}

Next, we upper bound the norm of the vectorization of a product of polynomials.  
\begin{claim}\label{claim:factor-upper-bound}
Let $f,g,h_1, \dots , h_k$ be polynomials in $X_1, \dots , X_d$ of degree at most $k$ with coefficients that are polynomials in formal variables $u_1, \dots , u_n$ of degree $O_k(1)$   Then for any pseudoexpectation $\widetilde{\E}$ of degree $C_k$ for some sufficiently large constant $C_k$ depending only on $k$,
\[
\widetilde{\E}[\norm{v(h_1)}^2 \dots \norm{v(h_k)}^2 \norm{v(fg)}^2] \leq O_k(1) \widetilde{\E}[\norm{v(h_1)}^2 \dots \norm{v(h_k)}^2\norm{v(f)}^2 \norm{v(g)}^2 ]
\]
where the pseudoexpectation operates on polynomials in $u_1, \dots , u_n$.
\end{claim}
\begin{proof}
Note that each monomial in the product $fg$ has degree at most $2k$ and thus can only be split in $O_k(1)$ ways.  Specifically, each entry of $v(fg)$ can be written as a sum of $O_k(1)$ entries of $v(f) \otimes v(g)$ so 
\[
\norm{v(fg)}^2 \leq O_k(1)\norm{v(f)}^2\norm{v(g)}^2
\]
where the difference between the two sides can be written as a sum of squares.  This implies the desired inequality.  
\end{proof}

Before we prove the final result in this section, we introduce a few definitions.
\begin{definition}
For a vector $v \in \R^d$ with integer coordinates, we define $\tau(v)$ to be the multiset formed by the coordinates of $v$.  We call $\tau$ the type of $v$.
\end{definition}

\begin{definition}
For a monomial say $X_1^{a_1}\dots X_d^{a_d}$, we call $(a_1, \dots , a_d) \in \R^d$ its degree vector.
\end{definition}

Now we can prove a lower bound on the norm of the vectorization of the product of polynomials.

\begin{claim}\label{claim:factor-lower-bound}
Let $f,g,h_1, \dots , h_k$ be polynomials in $X_1, \dots , X_d$ of degree at most $k$ with coefficients that are polynomials in formal variables $u_1, \dots , u_n$ of degree $O_k(1)$.   Then for any pseudoexpectation $\widetilde{\E}$ of degree $C_k$ for some sufficiently large constant $C_k$ depending only on $k$,
\[
\widetilde{\E}[\norm{v(h_1)}^2 \dots \norm{v(h_k)}^2 \norm{v(fg)}^2] \geq \Omega_k(1) \widetilde{\E}[\norm{v(h_1)}^2 \dots \norm{v(h_k)}^2\norm{v(f)}^2 \norm{v(g)}^2 ]
\]
where the pseudoexpectation operates on polynomials in $u_1, \dots , u_n$.
\end{claim}
\begin{proof}
We will first prove the statement for $h_1 = \dots = h_k = 1$.
\\\\
Let $S$ be the set of all types that can be obtained by taking the sum of two degree vectors for monomials of degree at most $k$ and let $T$ be the set of all types that can be obtained by taking the difference of two degree vectors for monomials of degree at most $k$.  Note that $|S|,|T| = O_k(1)$.  Now
\begin{align*}
\widetilde{\E}[\norm{v(fg)}^2] = \widetilde{\E}\left[\sum_{a} \left( \sum_{u+v = a}f_ug_v\right)^2\right] = \widetilde{\E}\left[ \sum_{\substack{u_1 + v_1 = u_2 +v_2}} f_{u_1}g_{v_1}f_{u_2}g_{v_2}\right]
\\ = \widetilde{\E}\left[ \sum_{\substack{u_1 - v_2 = u_2  - v_1}} f_{u_1}g_{v_1}f_{u_2}g_{v_2}\right] = \widetilde{\E}\left[\sum_{b} \left(\sum_{u-v = b}f_ug_v \right)^2\right]
\end{align*}

where the sums in the above expression are over all $a$ and all $b$ that are vectors in $\Z^d$ for which the inner summands are nonempty.  Let $T = \{t_1, \dots , t_n \}$ where the types $t_1, \dots , t_n$ are sorted in non-increasing order of their $L^2$ norm.  Recall that $T$ consists of all types that can be obtained by taking the difference of two degree vectors corresponding to monomials of degree at most $k$.  Now first note
\[
\widetilde{\E}[\norm{v(fg)}^2] \geq \widetilde{\E}\left[\sum_{b,\tau(b) = t_1}\left(\sum_{u-v = b}f_ug_v \right)^2 \right] = \widetilde{\E}\left[\sum_{b,\tau(b) = t_1}\left(\sum_{u-v = b}(f_ug_v)^2 \right)\right]
\]
since $t_1$ corresponds to the type $(k,-k)$ and each of the inner summands only contains one term.  Now consider $t_i$ for $i > 1$.
\begin{align*}
 \widetilde{\E}\left[\sum_{b,\tau(b) = t_i}\left(\sum_{u-v = b}f_ug_v \right)^2 \right] = \widetilde{\E}\left[\sum_{b,\tau(b) = t_i}\left(\sum_{u-v = b}(f_ug_v)^2 \right) \right] + 2\widetilde{\E}\left[\sum_{b,\tau(b) = t_i}\left(\sum_{\substack{\{u_1,v_1\} \neq \{u_1,v_2 \} \\ u_1 - v_1 = u_2 - v_2 = b}}f_{u_1}f_{u_2}g_{v_1}g_{v_2}\right) \right]
\end{align*}
Note that in the second sum, either $u_1 - v_2 \in t_j$ for $j < i$ or $u_2 - v_1 \in t_j$ for $j < i$.  To see this, let $a = v_1 - v_2$.  Then $u_1 - v_2 = b + a$ and $u_2 - v_1 = b - a$.  Now 
\[
\norm{b - a}_2^2 + \norm{b + a}_2^2 > \norm{b}_2^2
\]
since $a \neq 0$ so one of the differences must be of an earlier type.
\\\\
Next, note that for a fixed $u_1, v_2$, there are at most $O_k(1)$ possible values for $u_2',v_1'$ such that the term $f_{u_1}f_{u_2'}g_{v_1'}g_{v_2}$ appears.  This is because we must have $u_1 + v_2 = u_2' + v_1'$ and there are only $O_k(1)$ ways to achieve this.  Thus, by Cauchy Schwarz
\begin{align*}
 \widetilde{\E}\left[\sum_{b,\tau(b) = t_i}\left(\sum_{u-v = b}f_ug_v \right)^2 \right] \geq \widetilde{\E}\left[\sum_{b,\tau(b) = t_i}\left(\sum_{u-v = b}(f_ug_v)^2 \right) \right] \\ - O_k(1) \sqrt{\widetilde{\E}\left[\sum_{j < i}\sum_{b,\tau(b) = t_j}\left(\sum_{u-v = b}(f_ug_v)^2 \right) \right]} \cdot \sqrt{\widetilde{\E}[\norm{v(f)}^2 \norm{v(g)}^2 ]}
\end{align*}
Now combining the above with the fact that $|T| = n =  O_k(1)$ and that
\[
\widetilde{\E}[\norm{v(f)}^2 \norm{v(g)}^2 ] = \widetilde{\E}\left[ \sum_{i = 1}^n \sum_{b, \tau(b) = t_i} \left(\sum_{u - v = b} (f_u g_v)^2 \right)\right]
\]
we can complete the proof.  To see this, for each $i$, let
\begin{align*}
Q_i = \widetilde{\E}\left[\sum_{b,\tau(b) = t_i}\left(\sum_{u-v = b}(f_ug_v)^2 \right)\right] \\
R_i =  \widetilde{\E}\left[\sum_{b,\tau(b) = t_i}\left(\sum_{u-v = b}f_ug_v \right)^2\right]
\end{align*}
Also normalize so that 
\[
\widetilde{\E}[\norm{v(f)}^2 \norm{v(g)}^2 ]  = 1.
\]
Let $\delta$ be some suitably chosen constant depending only on $k$.  If $Q_1 \geq \delta$ then we are done.  Otherwise, we have an upper bound on the square root terms that are subtracted in the expression for $R_2$.  If $Q_2 \geq \Omega_1(k)\sqrt{\delta} $ then we are again done (since we have now reduced to the case where $Q_1 \leq \delta$).  Iteratively repeating this procedure, we are done whenever one of the $Q_i$ is sufficiently large compared to $Q_1, \dots , Q_{i-1}$.  However, not all of $Q_1, \dots , Q_n$ can be small since their sum is $1$.  Choosing $\delta$ to be a sufficiently small constant but depending only on $k$ we conclude that
\[
\widetilde{\E}[\norm{v(fg)}^2] \geq \Omega_k(1) \widetilde{\E}[\norm{v(f)}^2 \norm{v(g)}^2 ]
\]
as desired.
\\\\
For the general case when not all of the $h_i$ are $1$, we can multiply the insides of all of the pseudoexpectations above by $\norm{v(h_1)}^2 \dots \norm{v(h_k)}^2$ and the same argument will work.
\end{proof}

\section{Components Are Not Far Apart}\label{sec:close-case}

Now we are ready to present our main contribution: an algorithm that learns the parameters of a mixture of Gaussians $\mcl{M} = w_1G_1 + \dots + w_kG_k$ from an $\eps$-corrupted sample when the components are not too far apart.  In this section, we will assume that the mixture is in nearly isotropic position and that we have estimates for the Hermite polynomials.  We will show how to learn the parameters from these estimates.  In the next section, Section \ref{sec:moment-est}, we show how to actually place the mixture in isotropic position and obtain estimates for the Hermite polynomials.  
\\\\
We use the following conventions:
\begin{itemize}
    \item The true means and covariances are given by $(\mu_1, I + \Sigma_1), \dots , (\mu_k, I + \Sigma_k)$
    \item The true mixing weights are $w_1, \dots , w_k$ and are all bounded below by some value $w_{\min}$
    \item $\Delta$ is an upper bound that we have on $\norm{\mu_i}$ and $\norm{\Sigma_i}$ i.e. the components are not too far separated.  
    \item  $\norm{\mu_i - \mu_j}_2 + \norm{\Sigma_i - \Sigma_j}_2 \geq c$ for all $i \neq j$ i.e. no pair of components is too close
    \item We should think of $w_{\min}, c$ as being at least $\eps^{r}$ and $\Delta$ being at most $\eps^{-r}$ for some sufficiently small value of $r > 0$.
    \item Let the Hermite polynomials for the true mixture be given by $h_1 = h_{1, \mcl{M}}, h_2 = h_{2, \mcl{M}}, \dots $ where 
    \[
    \mcl{M} = w_1 N(\mu_1 , I + \Sigma_1) + \dots + w_k N(\mu_k , I + \Sigma_k)
    \]
\end{itemize}
In this section we assume that we have the following:
\begin{itemize}
    \item Estimates $\overline{h}_i(X)$ for the Hermite polynomials such that $\norm{v(\overline{h_i}(X) - h_i(X))}^2 \leq \epsilon'  = \poly(\eps)$
\end{itemize}
and our only interaction with the actual samples is through these estimates.  We will show how to obtain these estimates in Section \ref{sec:moment-est} (closely mirroring the method in \cite{kane2020robust}).
\\\\
The main theorem that we prove in this section is as follows.
\begin{theorem}\label{thm:main-close-case}
Let $\eps'$ be a parameter that is sufficiently small in terms of $k$.  There is a sufficiently small function $f(k)$ and a sufficiently large function $F(k)$ such that if 
\[
\mcl{M} = w_1N(\mu_1, I + \Sigma_1) + \dots + w_kN(\mu_k, I + \Sigma_k)
\]
is a mixture of Gaussians with
\begin{itemize}
    \item $\norm{\mu_i}_2 , \norm{\Sigma_i}_2 \leq \Delta$ for all $i$
    \item $\norm{\mu_i - \mu_j}_2 + \norm{\Sigma_i - \Sigma_j}_2 \geq c$ for all $i \neq j$
    \item $w_1, \dots , w_k \geq w_{\min}$
\end{itemize}
for parameters $w_{\min}, c \geq \eps'^{f(k)}$ and $\Delta \leq \eps'^{-f(k)}$ and we are given estimates $\overline{h}_i(X)$ for the Hermite polynomials for all $i \leq F(k)$ such that 
\[
\norm{v(\overline{h_i}(X) - h_i(X))}^2 \leq \epsilon' 
\]
where $h_i$ are the Hermite polynomials for the true mixture $\mcl{M}$, then there is an algorithm that returns $\poly(1/\eps')^{O_1(k)}$ candidate mixtures, at least one of which satisfies 
\[
\norm{w_i - \widetilde{w_i}} + \norm{\mu_i - \widetilde{\mu_i}}_2 + \norm{\Sigma_i - \widetilde{\Sigma_i}}_2 \leq \eps'^{f(k)}
\]
for all $i$.
\end{theorem}

Informally, assuming that the parameters of the components of the mixture are bounded by $\poly(1/\eps')$ and that their separation is at least $\poly(\eps')$, given $\eps'$-accurate estimates for the Hermite polynomials, we can learn the parameters of the mixture to within Frobenius error $\poly(\eps')$.

\subsection{Reducing to all pairs of parameters equal or separated}\label{sec:reduction}

We claim that it suffices to work under the following assumption.  All pairs of parameters are either separated of equal.  More specifically, for each pair of parameters $\mu_i, \mu_j$ (and same for $\Sigma_i, \Sigma_j$), either $\mu_i = \mu_j$ or 
\[
\norm{\mu_i - \mu_j}_2 \geq c
\]

We now prove that it suffices to work with the above simplification.  For any function $0 < f(k) < 1$ depending only on $k$, there is some $C \geq (f(k))^{k^2}$ such that there is no pair of parameters $\mu_i, \mu_j$ or $\Sigma_i, \Sigma_j$ whose distance is in the interval $[ \eps'^{C}, \eps'^{ f(k)C} ]$.  Now consider the graph on the $k$ nodes where $i,j$ are connected if and only if 
\[
\norm{\mu_i - \mu_j} \leq \eps'^{f(k)C}
\]
We now construct a new mixture $N(\mu_i', \Sigma_i')$.  For each connected component say $\{i_1, \dots , i_j \}$ , pick a representative and set $\mu_{i_1}' = \mu_{i_2}' = \dots = \mu_{i_j}'  = \mu_{i_1}$.  Do this for all connected components and similar in the graph on covariance matrices.  For all $i$, we have
\[
\norm{\mu_i' - \mu_i} , \norm{\Sigma_i' - \Sigma_i} \leq O_k(1)\eps'^C
\]
because there is a path of length at most $k$ connecting $i$ to the representative in its component that it is rounded to, and all edges correspond to pairs within distance of $\eps'^C$.

The Hermite polynomials of this new mixture satisfy
\[
\norm{v(h_m' - h_m)}^2 \leq O_k(1) \Delta^{O_k(1)} \epsilon'^{\Omega_k(1)C}
\]
as long as $m$ is bounded as a function of $k$.  If we pretend that the new mixture is the true mixture, we have estimates $\overline{h}_i(X)$ such that 
\[
\norm{v(h_i' - \overline{h_i})}^2 \leq O_k(1) \Delta^{O_k(1)} \epsilon'^{\Omega_k(1)C}
\]
and all pairs of parameters in the new mixture are either equal or $\eps'^{f(k)C}$ separated. If we prove Theorem \ref{thm:main-close-case} with the assumption that the pairs of parameters are separated or equal, then we can choose $f(k)$ accordingly and then we deduce that the theorem holds in the general case (with worse, but still polynomial, bounds on $\Delta ,c,  w_{\min}$ and the accuracy of our output as a function of $\eps'$).  
\\\\
From now on we will work with the assumption that each pair of parameters is either equal or separated by $c$.

\subsection{SOS Program Setup}\label{sec:sos-setup}
Our algorithm for learning the parameters when given estimates of the Hermite polynomials involves solving an SOS program.  Here we set up the SOS program that we will solve.  
\\\\
We will let $D = \binom{d}{2} + d$.  We think of mapping between symmetric $d \times d$ matrices and $\R^D$ as 
\[
\begin{bmatrix}
a_{11} & \dots & a_{1d} \\ \vdots & \ddots & \vdots \\ a_{d1} & \dots & a_{dd}
\end{bmatrix}
\leftrightarrow (a_{11}, 2a_{12} , 2a_{13}, \dots , a_{dd})
\]

\begin{definition}[Parameter Solving Program $\mcl{S}$]\label{def:parameter-SOS}
We will have the following variables
\begin{itemize}
    \item $u_1 = (u_{11}, \dots , u_{1d}) , \dots, u_k = (u_{k1}, \dots, u_{kd})$
    \item $v_1 = (v_{1, (1,1)}, v_{1, (1,2)},  \dots , v_{1, (d,d)}) , \dots , v_k = (v_{k, (1,1)}, v_{k, (1,2)},  \dots , v_{k, (d,d)})$
\end{itemize}
In the above $u_1, \dots, u_k \in \R^d$ and $v_1, \dots v_k \in \R^D$.  Our goal will be to solve for these variables in a way so that the solutions form orthonormal bases for the span of the $\mu_i$ and the span of the $\Sigma_i$.  Note $v_1, \dots , v_k$ live in $\R^D$ because the $\Sigma_i$ must be symmetric.  

We guess coefficients $a_{ij}, b_{ij}$ where $i,j \in [k]$ expressing the means and covariances in this orthonormal basis.  We ensure that the guesses satisfy the property that for every pair of vectors $A_i = (a_{i1}, \dots , a_{ik}), A_j = (a_{j1}, \dots , a_{jk})$ either $A_i = A_j$ or 
\[
\norm{A_i - A_j}_2 \geq \frac{c}{2} 
\]
and similarly for $B_i, B_j$.  We ensure that
\[
\norm{A_i}_2 \leq 2\Delta
\]
Ensure similar conditions for the $\{B_i \}$.  We also guess the mixing weights $\widetilde{w_1}, \dots , \widetilde{w_k}$ and ensure that our guesses are all at least $w_{\min}/2$.

Now we set up the constraints.  Let $C$ be a sufficiently large integer depending only on $k$.  Define $\widetilde{\mu_i} = a_{i1}u_1 + \dots + a_{ik}u_k$ and define $\widetilde{\Sigma_i}$ similarly.  These are linear expressions in the variables that we are solving for.  Now consider the hypothetical mixture with mixing weights $\widetilde{w_i}$, means $\widetilde{\mu_i}$, and covariances $I + \widetilde{ \Sigma_i}$.  The Hermite polynomials for this hypothetical mixture $\widetilde{h_i}(X)$ can be written as formal polynomials in $X = (X_1, \dots , X_d)$ with coefficients that are polynomials in $u,v$.  Note that we can explicitly write down these Hermite polynomials.  The set of constraints for our SOS system is as follows:
\begin{itemize}
    \item $\norm{u_i}_2^{2} = 1$ for all $1 \leq i \leq k$
    \item $\norm{v_i}_2^{2} = 1$ for all $1 \leq i \leq k$
    \item $u_i \cdot u_j = 0$ for all $i \neq j$
    \item $v_i \cdot v_j = 0$ for all $i \neq j$
    \item For all $p = 1,2, \dots , C$
    \[
    \norm{v(\widetilde{h_p}(X) - \overline{h_p}(X))}^2 \leq 100\epsilon'
    \]
\end{itemize}

\end{definition}

 Note that we can explicitly write down the last set of constraints because we have estimates $\overline{h_i}$.  
 \\\\
 It is important to note that the $\widetilde{w_i}, A_i, B_i$ are real numbers.  We will attempt to solve the system for each of our guesses and show that for some set of guesses, we obtain a solution from which we can recover the parameters.  We can brute-force search over an $\eps'$-net because there are only $O_k(1)$ parameters to guess.  We call the SOS program that we set up $\mathcal{S}$.

\subsection{Analysis}
We now prove a set of properties that must be satisfied by any pseudoexpectation of degree $C_k$ satisfying $\mathcal{S}$ where $C_k$ is a sufficiently large constant depending only on $k$.  What we would ideally want to show is that
\begin{itemize}
    \item The span of the $\widetilde{\Sigma_i}$ is close to the span of the $\Sigma_i$
    \item The span of the $\widetilde{\mu_i}$ is close to the span of the $\mu_i$ 
\end{itemize}
However, it appears to be difficult to prove a statement of the above form within an SOS framework.  Instead, we will look at the pseudoexpectations of the matrices
\[
M_i = \widetilde{\E}[\widetilde{\Sigma_i} \widetilde{\Sigma_i}^T]
\]
(where $\widetilde{\Sigma_i}$ is viewed as a length-$D$ vector so $\widetilde{\Sigma_i} \widetilde{\Sigma_i}^T$ is a $D \times D$ matrix.)  The two key properties that we will prove about these matrices are in Lemmas \ref{lemma:trace-bound} and \ref{lemma:span}. \\\\ 
Roughly Lemma \ref{lemma:trace-bound} says that any singular vector that corresponds to a large singular value of $M_i$ must be close to the span of the $\{ \Sigma_i \}$.  Lemma \ref{lemma:span} says that any vector $v$ that has large projection onto the subspace spanned by the $\{ \Sigma_i \}$ must have the property that $v^TM_iv$ is large for some $i$.  Putting these together, we can take the the top-$k$ principal components of each of $M_1, \dots , M_k$ and show that the span of these essentially contains the span of the $\{ \Sigma_i \}$ (this last step is done outside the SOS framework).  We can now brute-force over an $\eps'$-net and guess the $\Sigma_i$ (since we have narrowed them down to an $O_k(1)$-dimensional subspace).  We can then plug in real values for the covariances and solve for the means using a similar method.

\subsubsection{Algebraic Identities}\label{sec:alg-identities}

First we will prove several purely algebraic identities.  We will slightly abuse notation and for $\mu \in \R^d$, we use $\mu(X)$ to denote the inner product of $\mu$ with the formal variables $(X_1, \dots , X_d)$ and for $\Sigma \in \R^D$, we will use $\Sigma(X)$ to denote the quadratic form in formal variables $(X_1, \dots , X_d)$ given by $X^T\Sigma X$ (when $\Sigma$ is converted to a symmetric $d \times d$ matrix). It will be useful to consider the following two formal power series (in $y$)
\begin{align*}
F(y) = \sum_{i = 1}^k w_i e^{\mu_i(X)y + \frac{1}{2}\Sigma_i(X)y^2} \\
\widetilde{F}(y) = \sum_{i = 1}^k \widetilde{w_i} e^{\widetilde{\mu_i}(X)y + \frac{1}{2}\widetilde{\Sigma_i}(X)y^2}
\end{align*}

We view these objects in the following way: the coefficients of $1, y, y^2, \cdots $ are formal polynomials in $(X_1, \dots , X_d)$.  In the first expression, the coefficients of these polynomials are (unknown) constants.  In the second, the coefficients are polynomials in the variables $u_1, \dots , u_k, v_1, \dots , v_k$.  In fact, the coefficients in the first power series are precisely $h_1, h_2, \dots $ while the coefficients in the second power series are precisely $\widetilde{h_1}, \widetilde{h_2}, \dots $.
The key insight is the following:

\begin{quote}
    {\em After taking derivatives and polynomial combinations of either of the above formal power series, the coefficients can still be expressed as polynomial combinations of their respective Hermite polynomials.}
\end{quote}

\begin{definition}
Let $\mcl{D}_i$ denote the differential operator $(\partial - (\mu_i(X) + \Sigma_i(X)y))$ and $\widetilde{\mcl{D}_i}$ denote the differential operator $(\partial - (\widetilde{\mu_i}(X) + \widetilde{\Sigma_i}(X)y))$.  As usual, the partial derivatives are taken with respect to $y$.
\end{definition}

To simplify the exposition, we make the following definition:
\begin{definition}\label{def:k-simple}
Consider a polynomial $P(X)$ that is a formal polynomial in $X_1, \dots , X_d$ whose coefficients are polynomials in the indeterminates $u_1, \dots , u_k, v_1, \dots , v_k$.  We say $P$ is $m$-simple if $P$ can be written as a linear combination of a constant number of terms that are a product of some of $\{\mu_i(X) \}, \{\Sigma_i(X) \}, \{ \widetilde{\mu_i}(X) \}, \{ \widetilde{\Sigma_i}(X) \}$  where 
    \begin{enumerate}
        \item The coefficients in the linear combination are bounded by a constant depending only on $m,k$
        \item The number of terms in the sum depends only on $m$ and $k$
        \item The number of terms in each product depends only on $m$ and $k$ 
    \end{enumerate}
\end{definition}

\begin{claim}\label{claim:identity-1}
Consider the power series
\begin{align*}
\widetilde{\mcl{D}_{k-1}}^{2^{2k-2}} \dots \widetilde{\mcl{D}_{1}}^{2^{k}} \mcl{D}_k^{2^{k-1}} \dots \mcl{D}_1^{1}(\widetilde{F})
\end{align*}
For any $m$, the coefficient of $y^{m}$ when the above is written as a formal power series can be written in the form 
\[
P_0(X) + P_1(X)\widetilde{h_1}(X) + \dots + P_{m'}(X)\widetilde{h_{m'}}(X)
\]
where 
\begin{itemize}
    \item $m'$ depends only on $m$ and $k$ 
    \item Each of the $P_i$ is $m$-simple
    \item We have 
    \[
    P_0(X) + P_1(X)h_1(X) + \dots + P_{m'}(X)h_{m'}(X) = 0
    \]
    as an algebraic identity over formal variables $X_1, \dots , X_d, \{u_i \}, \{v_i \}$.
\end{itemize}
\end{claim}

\begin{proof}
Note the coefficients of $\widetilde{F}$ (as a formal power series in $y$) are exactly given by the $\widetilde{h_i}$.  Now the number of differential operators we apply is $O_k(1)$.  The first two statements can be verified through straightforward computations since when applying each of the differential operators, we are simply multiplying the coefficients by some of $\{\mu_i(X) \}, \{\Sigma_i(X) \}, \{ \widetilde{\mu_i}(X) \}, \{ \widetilde{\Sigma_i}(X) \}$ and taking a linear combination.  Next, note that by Corollary \ref{corollary:null-operator}
\[
\mcl{D}_k^{2^{k-1}} \dots \mcl{D}_1^{1}(F) = 0.
\]
To see this, we prove by induction that the differential operator
\[
\mcl{D}_j^{2^{j-1}} \dots \mcl{D}_1^{1}(F)
\]
annihilates the components of $F$ corresponding to Gaussians $N(\mu_1, I + \Sigma_1), \dots , N(\mu_j, I + \Sigma_j)$.  The base case is clear.  To complete the induction step, note that by Corollary \ref{corollary:null-operator}, the above operator puts polynomials of degree at most $1 + 2 + \dots + 2^{j-1} = 2^j - 1$ in front of the other components.  Thus the operator  
\[
\mcl{D}_{j + 1}^{2^{j}} \dots \mcl{D}_1^{1}(F)
\]
annhiliates the first $j + 1$ components, completing the induction.  We now conclude that 
\[
\widetilde{\mcl{D}_{k-1}}^{2^{2k-2}} \dots \widetilde{\mcl{D}_{1}}^{2^k} \mcl{D}_k^{2^{k-1}} \dots \mcl{D}_1^{1}(F) = 0
\]
implying that if the coefficients of $\widetilde{F}$ were $h_1, \dots , h_m$, then the result would be identically zero.
\end{proof}

\begin{claim}\label{claim:identity-2}
Consider the power series
\begin{align*}
\mcl{D}_{k-1}^{2^{2k-2}} \dots\mcl{D}_{1}^{2^k} \widetilde{\mcl{D}_k}^{2^{k-1}} \dots \widetilde{\mcl{D}_1}^{1} (F)
\end{align*}
For any $m$, the coefficient of $y^{m}$ when the above is written as a formal power series can be written in the form
\[
P_0(X) + P_1(X)h_1(X) + \dots + P_{m'}(X)h_{m'}(X)
\]
where 
\begin{itemize}
    \item $m'$ depends only on $m$ and $k$
    \item Each of the $P_i$ is $m$-simple
    \item We have 
    \[
    P_0(X) + P_1(X)\widetilde{h_1}(X) + \dots + P_{m'}(X)\widetilde{h_{m'}}(X) = 0
    \]
    as an algebraic identity over formal variables $X_1, \dots , X_d, \{u_i \}, \{v_i \}$.
\end{itemize}
\end{claim}
\begin{proof}
The proof is identicial to the proof of Claim \ref{claim:identity-1}.
\end{proof}

Note that the polynomials $P_i$ in Claim \ref{claim:identity-1} and Claim \ref{claim:identity-2} are \emph{not} necessarily the same.

\subsubsection{Warm-up: All Pairs of Parameters are Separated}\label{sec:solve-covariances}

As a warm-up, we first analyze the case where all pairs of true parameters $\mu_i, \mu_j$ and $\Sigma_i, \Sigma_j$ satisfy $\norm{\mu_i - \mu_j}_2 \geq c$ and $\norm{\Sigma_i - \Sigma_j}_2 \geq c$.  We will show how to deal with the general case where parameters may be separated or equal in Section \ref{sec:equal-or-sep}.
\\\\
We can assume that our guesses satisfy $\norm{A_i - A_j}_2 \geq c/2$ and $\norm{B_i - B_j}_2 \geq c/2$ for all $i,j$.  The key expressions to consider are applying the following differential operators
\begin{align*}
\mcl{D} = \widetilde{\mcl{D}_{k}}^{2^{2k-1} - 1}\widetilde{\mcl{D}_{k-1}}^{2^{2k-2}} \dots \widetilde{\mcl{D}_{1}}^{2^k} \mcl{D}_k^{2^{k-1}} \dots \mcl{D}_1^{1} \\ 
\widetilde{\mcl{D}} = \mcl{D}_{k}^{2^{2k-1}- 1}\mcl{D}_{k-1}^{2^{2k-2}} \dots\mcl{D}_{1}^{2^k} \widetilde{\mcl{D}_k}^{2^{k-1}} \dots \widetilde{\mcl{D}_1}^{1} 
\end{align*}
to $F$ and $\widetilde{F}$ respectively. The reason these differential operators are so useful is that $\mcl{D}$ zeros out the generating function for the true mixture and also zeros out all but one component of the generating function for the hypothetical mixture with parameters $\widetilde{w_i} , \widetilde{\mu_i}, I + \widetilde{\Sigma_i}$.  For the one component that is not zeroed out, only the leading coefficient remains and we can use Claim \ref{claim:leading-coeff} to explicitly compute the leading coefficient.  Thus, we can compare the results of applying these operators on the generating functions for the true and hypothetical mixtures and, using the fact that the Hermite polynomials for these mixtures must be close, we obtain algebraic relations that allow us to extract information about individual components.

We begin by explicitly computing the relevant leading coefficients.

\begin{claim}\label{claim:elimination1}
Write
\[
\widetilde{\mcl{D}_{k}}^{2^{2k-1} - 1}\widetilde{\mcl{D}_{k-1}}^{2^{2k-2}} \dots \widetilde{\mcl{D}_{1}}^{2^k} \mcl{D}_k^{2^{k-1}} \dots \mcl{D}_1^{1}(\widetilde{F})
\]
as a formal power series in $y$.  Its evaluation at $y = 0$ is 
\[
C_k\widetilde{w_k}\prod_{i=1}^k (\widetilde{\Sigma_k}(X) - \Sigma_i(X) )^{2^{i-1}}\prod_{i=1}^{k-1}(\widetilde{\Sigma_k}(X) - \widetilde{\Sigma_i}(X) )^{2^{k+i-1}}
\]
where $C_k$ is a constant depending only on $k$.

\end{claim}
\begin{proof}
Write 
\[
\widetilde{F}(y) = \sum_{i = 1}^k \widetilde{w_i} e^{\widetilde{\mu_i}(X)y + \frac{1}{2}\widetilde{\Sigma_i}(X)y^2}
\]
When applying the differential operator, by Corollary \ref{corollary:null-operator}, all of the terms become $0$ except for 
\[
\widetilde{w_k} e^{\widetilde{\mu_k}(X)y + \frac{1}{2}\widetilde{\Sigma_i}(X)y^2}.
\]
We now use Claim \ref{claim:leading-coeff} and Claim \ref{claim:degree-reduction} to analyze what happens when applying the differential operator to this term.  We know that 
\[
\widetilde{\mcl{D}_{k-1}}^{2^{2k-2}} \dots \widetilde{\mcl{D}_{1}}^{2^k} \mcl{D}_k^{2^{k-1}} \dots \mcl{D}_1^{1}(\widetilde{F}) = P(y)e^{\widetilde{\mu_k}(X)y + \frac{1}{2}\widetilde{\Sigma_i}(X)y^2}
\]
where $P$ has leading coefficient 
\[
\widetilde{w_k}\prod_{i=1}^k (\widetilde{\Sigma_k}(X) - \Sigma_i(X) )^{2^{i-1}}\prod_{i=1}^{k-1}(\widetilde{\Sigma_k}(X) - \widetilde{\Sigma_i}(X) )^{2^{k+i-1}}
\]
and degree $2^{2k-1} - 1$.  Thus,
\begin{align*}
&\widetilde{\mcl{D}_{k}}^{2^{2k-1} - 1}\widetilde{\mcl{D}_{k-1}}^{2^{2k-2}} \dots \widetilde{\mcl{D}_{1}}^{2^k} \mcl{D}_k^{2^{k-1}} \dots \mcl{D}_1^{1}(\widetilde{F}) \\ &= (2^{2k-1} - 1)!\widetilde{w_k}\prod_{i=1}^k (\widetilde{\Sigma_k}(X) - \Sigma_i(X) )^{2^{i-1}}\prod_{i=1}^{k-1}(\widetilde{\Sigma_k}(X) - \widetilde{\Sigma_i}(X) )^{2^{k+i-1}}e^{\widetilde{\mu_k}(X)y + \frac{1}{2}\widetilde{\Sigma_i}(X)y^2}
\end{align*}
and plugging in $y = 0$, we are done.
\end{proof}

\begin{claim}\label{claim:elimination2}
Write
\[
\mcl{D}_{k}^{2^{2k-1}- 1}\mcl{D}_{k-1}^{2^{2k-2}} \dots\mcl{D}_{1}^{2^k} \widetilde{\mcl{D}_k}^{2^{k-1}} \dots \widetilde{\mcl{D}_1}^{1} (F)
\]
as a formal power series in $y$.  Its evaluation at $y = 0$ is 
\[
C_k w_k\prod_{i=1}^{k}(\Sigma_k(X) - \widetilde{\Sigma_i}(X) )^{2^{i-1}} \prod_{i=1}^{k-1} (\Sigma_k(X) - \Sigma_i(X) )^{2^{k + i-1}}
\]
where $C_k$ is a constant depending only on $k$.

\end{claim}
\begin{proof}
This can be proved using the same method as Claim \ref{claim:elimination1}.
\end{proof}

Combining the previous two claims with Claim \ref{claim:identity-1} and Claim \ref{claim:identity-2}, we can write the expressions for the leading coefficients as polynomial combinations of the Hermite polynomials.

\begin{lemma}\label{lemma:key-bound1}
Consider the polynomial
\[
\widetilde{w_k}\prod_{i=1}^k (\widetilde{\Sigma_k}(X) - \Sigma_i(X) )^{2^{i-1}}\prod_{i=1}^{k-1}(\widetilde{\Sigma_k}(X) - \widetilde{\Sigma_i}(X) )^{2^{k+i-1}}
\]
It can be written in the form
\[
P_0(X) + P_1(X)\widetilde{h_1}(X) + \dots + P_{m}(X)\widetilde{h_{m}}(X)
\]
where 
\begin{itemize}
    \item $m$ is a function of $k$
    \item Each of the $P_i$ is $m$-simple
    \item We have 
    \[
    P_0(X) + P_1(X)h_1(X) + \dots + P_{m}(X)h_{m}(X) = 0
    \]
    as an algebraic identity over formal variables $X_1, \dots , X_d, \{u_i \}, \{v_i \}$.
\end{itemize}
\end{lemma}
\begin{proof}
Consider the power series
\begin{align*}
\widetilde{\mcl{D}_{k}}^{2^{2k-1} - 1}\widetilde{\mcl{D}_{k-1}}^{2^{2k-2}} \dots \widetilde{\mcl{D}_{1}}^{2^k} \mcl{D}_k^{2^{k-1}} \dots \mcl{D}_1^{1}(\widetilde{F})
\end{align*}
Now using Claim \ref{claim:elimination1} and repeating the proof of Claim \ref{claim:identity-1}, we get the desired.
\end{proof}

Similarly, we have:

\begin{lemma}\label{lemma:key-bound2}
Consider the polynomial
\[
w_k\prod_{i=1}^{k}(\Sigma_k(X) - \widetilde{\Sigma_i}(X) )^{2^{i-1}} \prod_{i=1}^{k-1} (\Sigma_k(X) - \Sigma_i(X) )^{2^{k + i-1}}
\]
It can be written in the form
\[
P_0(X) + P_1(X)h_1(X) + \dots + P_{m}(X)h_{m}(X)
\]
where 
\begin{itemize}
    \item $m$ is a function of $k$
    \item Each of the $P_i$ is $m$-simple
    \item We have 
    \[
     P_0(X) + P_1(X)\widetilde{h_1}(X) + \dots + P_{m}(X) \widetilde{h_{m}}(X) = 0
    \]
    as an algebraic identity over formal variables $X_1, \dots , X_d, \{u_i \}, \{v_i \}$.
\end{itemize}
\end{lemma}

Everything we've done so far has been symbolic manipulations and the claims in this section are all true as algebraic identities.  We are now ready to analyze the SOS program.  Note the polynomials $P_0, \dots, P_m$ in Lemma \ref{lemma:key-bound1} are unknown because they depend on the true parameters.  This is fine because we will simply use their existence to deduce properties of pseudoexpectations that solve the SOS-system $\mcl{S}$.

Let $U$ be the subspace spanned by the true $\mu_1, \dots , \mu_k$ and let $V$ denote the subspace spanned by the true (flattened) $\Sigma_1, \dots , \Sigma_k$.  We will use $\Gamma_{V}, \Gamma_{V^{\perp}}$ to denote projections onto $V$ and the orthogonal complement of $V$ (and similar for $U, U^{\perp}$).  Note that these are linear maps.  

Our goal now will be to show that $V$ is essentially contained within the span of the union of the top $k$ principal components of the matrices 
\[
\widetilde{\E}[\widetilde{\Sigma_1} \widetilde{\Sigma_1}^T], \dots , \widetilde{\E}[\widetilde{\Sigma_k} \widetilde{\Sigma_k}^T]
\]
This gives us a $k^2$-dimensional space that essentially contains $V$ and then we can guess the true covariance matrices via brute force search.  In the first key lemma, we prove that the matrix $\widetilde{\E}[\widetilde{\Sigma_i} \widetilde{\Sigma_i}^T]$ lives almost entirely within the subspace $V$.

\begin{lemma}\label{lemma:trace-bound}
Let $\widetilde{\E}$ be a pseudoexpectation of degree $C_k$ for some sufficiently large constant $C_k$ depending only on $k$ that solves $\mcl{S}$.  Consider the matrix
\[
M = \widetilde{\E}[\widetilde{\Sigma_k} \widetilde{\Sigma_k}^T]
\]
where by this we mean we construct the $D \times D$ matrix $\widetilde{\Sigma_k} \widetilde{\Sigma_k}^T$ whose entries are quadratic in the variables $\{ u \}, \{v \}$ and then take the entry-wise pseudoexpectation.  Then
\[
\Tr_{V^{\perp}}(M) \leq \eps'^{2^{-k}}O_k(1)  \left(\frac{\Delta}{w_{\min}c} \right)^{O_k(1)}
\]
where $\Tr_{V^{\perp}}(M)$ denotes the trace of $M$ on the subspace $V^{\perp}$.
\end{lemma}
\begin{proof} 
Using Lemma \ref{lemma:key-bound1}, we may write
\[
\widetilde{w_k}\prod_{i=1}^k (\widetilde{\Sigma_k}(X) - \Sigma_i(X) )^{2^{i-1}}\prod_{i=1}^{k-1}(\widetilde{\Sigma_k}(X) - \widetilde{\Sigma_i}(X) )^{2^{k+i-1}} = P_1(X)(\widetilde{h_1}(X) - h_1(X)) + \dots + P_{m}(X)(\widetilde{h_{m}}(X) - h_m(X))
\]
where $m = O_k(1)$

Now we bound
\[
\widetilde{\E}\left[\norm{v\left(P_1(X)(\widetilde{h_1}(X) - h_1(X)) + \dots + P_{m}(X)(\widetilde{h_{m}}(X) - h_m(X))\right)}^2\right] 
\]
Using Claim \ref{claim:factor-upper-bound} and Claim \ref{claim:sum-bound},
\begin{align*}
&\widetilde{\E}\left[\norm{v\left(P_1(X)(\widetilde{h_1}(X) - h_1(X)) + \dots + P_{m}(X)(\widetilde{h_{m}}(X) - h_m(X))\right)}^2\right] \\
&\leq O_k(1) \sum_{i=1}^m \widetilde{\E}\left[ \norm{v(P_i(X))}^2 \cdot \norm{v(\widetilde{h_i}(X) - h_i(X))}^2\right] \\
&\leq O_k(1) \sum_{i=1}^m \widetilde{\E}\left[ \norm{v(P_i(X))}^2 \cdot 2\left(\norm{v(\widetilde{h_i}(X) - \overline{h_i}(X))}^2 + \norm{v(\overline{h_i}(X) - h_i(X))}^2 \right)\right]
\end{align*}
Where the last step is true because Claim \ref{claim:sum-bound} allows us to write the difference between the two sides as a sum of squares.
\\\\
Now $\norm{v(\overline{h_i}(X) - h_i(X))}^2$ is just a real number and is bounded above by $\eps'$ by assumption.  We also have the constraint that 
\[
\norm{v(\widetilde{h_i}(X) - \overline{h_i}(X))}^2 \leq 100\eps'
\]
so
\[
\widetilde{\E}\left[\norm{v\left(P_1(X)(\widetilde{h_1}(X) - h_1(X)) + \dots + P_{m}(X)(\widetilde{h_{m}}(X) - h_m(X))\right)}^2\right]  \leq O_k(1)\eps'\sum_{i=1}^m \widetilde{\E}\left[ \norm{v(P_i(X))}^2\right]
\]
Now we use the properties from Lemma \ref{lemma:key-bound1} that each of the $P_i$ can be written as a linear combination of a constant number of terms that are a product of some of $\{\mu_i(X) \}, \{\Sigma_i(X) \}, \{ \widetilde{\mu_i}(X) \}, \{ \widetilde{\Sigma_i}(X) \}$ where 
    \begin{itemize}
        \item The coefficients in the linear combination are bounded by a constant depending only on $k$
        \item The number of terms in the sum depends only on $k$
        \item The number of terms in each product depends only on $k$
    \end{itemize}
Note for each $\widetilde{\mu_i}(X)$, since we ensured that our guesses for the coefficients that go with the orthonormal basis $u_1, \dots , u_k$ are at most $\Delta$ and we have the constraints $\norm{u_i}_2^2 = 1, u_i \cdot u_j = 0$, we have
\[
\norm{v(\widetilde{\mu_i}(X))}^2 \preceq_{SOS} O_k(1)\Delta^2
\]
where $\preceq_{SOS}$ means the difference can be written as a sum of squares. 
We can make similar arguments for $\widetilde{\Sigma_i}(X), \mu_i(X), \Sigma_i(X)$.  Now using Claim \ref{claim:sum-bound} and Claim \ref{claim:factor-upper-bound} we can deduce
\[
\widetilde{\E}\left[ \norm{v(P_i(X))}^2\right] \leq O_k(1) \Delta^{O_k(1)}
\]
Overall, we have shown
\[
\widetilde{\E}\left[\norm{v\left(P_1(X)(\widetilde{h_1}(X) - h_1(X)) + \dots + P_{m}(X)(\widetilde{h_{m}}(X) - h_m(X))\right)}^2\right]  \leq O_k(1)\eps' \Delta^{O_k(1)}
\]
Now we examine the expression
\[
\widetilde{\E}\left[ \norm{v\left(\widetilde{w_k}\prod_{i=1}^k (\widetilde{\Sigma_k}(X) - \Sigma_i(X) )^{2^{i-1}}\prod_{i=1}^{k-1}(\widetilde{\Sigma_k}(X) - \widetilde{\Sigma_i}(X) )^{2^{k+i-1}}\right)}^2\right]
\]
By Claim \ref{claim:factor-lower-bound} (recall $\widetilde{w_k}$ is a constant that we guess),
\begin{align*}
&\widetilde{\E}\left[ \norm{v\left(\widetilde{w_k}\prod_{i=1}^k (\widetilde{\Sigma_k}(X) - \Sigma_i(X) )^{2^{i-1}}\prod_{i=1}^{k-1}(\widetilde{\Sigma_k}(X) - \widetilde{\Sigma_i}(X) )^{2^{k+i-1}}\right)}^2 \right] \\ &\geq \widetilde{w_{k}} \Omega_k(1)\widetilde{\E}\left[ \prod_{i=1}^k \left(\norm{v(\widetilde{\Sigma_k}(X) - \Sigma_i(X) )}^2\right)^{2^{i-1}}\prod_{i=1}^{k-1}\left(\norm{v(\widetilde{\Sigma_k}(X) - \widetilde{\Sigma_i}(X) )}^2\right)^{2^{k+i-1}}\right]
\end{align*}
Note that
\[
\norm{v(\widetilde{\Sigma_k}(X) - \Sigma_i(X) )}^2 \succeq_{SOS} \norm{\Gamma_{V^{\perp}}(\widetilde{\Sigma_k})}^2
\]
(recall that $\Gamma_{V^{\perp}}$ is a projection map with unknown but constant coefficients).  Next, since we ensure that the coefficients $B_i$ that we guess for the orthonormal basis satisfy $\norm{B_i - B_j}_2 \geq \frac{c}{2}$, we have
\[
\norm{v(\widetilde{\Sigma_k}(X) - \widetilde{\Sigma_i}(X) )}^2 \succeq_{SOS} \frac{c^2}{4}
\]
where we use the constraints in $\mcl{S}$ that $\norm{v_i}_2^2 = 1, v_i \cdot v_j = 0$.  Overall, we conclude
\[
\widetilde{\E}\left[ \norm{v\left(\widetilde{w_k}\prod_{i=1}^k (\widetilde{\Sigma_k}(X) - \Sigma_i(X) )^{2^{i-1}}\prod_{i=1}^{k-1}(\widetilde{\Sigma_k}(X) - \widetilde{\Sigma_i}(X) )^{2^{k+i-1}}\right)}^2 \right] \geq \Omega_k(1)\E\left[ \norm{\Gamma_{V^{\perp}}(\widetilde{\Sigma_k})}^{2^{k+1} - 2} \right] (\widetilde{w_k}c)^{O_k(1)}
\]
Note
\begin{align*}
\widetilde{\E}\left[ \norm{v\left(\widetilde{w_k}\prod_{i=1}^k (\widetilde{\Sigma_k}(X) - \Sigma_i(X) )^{2^{i-1}}\prod_{i=1}^{k-1}(\widetilde{\Sigma_k}(X) - \widetilde{\Sigma_i}(X) )^{2^{k+i-1}}\right)}^2 \right] = \\ \widetilde{\E}\left[\norm{v\left(P_1(X)(\widetilde{h_1}(X) - h_1(X)) + \dots + P_{m}(X)(\widetilde{h_{m}}(X) - h_m(X))\right)}^2\right]  
\end{align*}
because the inner expressions are equal symbolically.  Thus
\[
\widetilde{\E}\left[ \norm{\Gamma_{V^{\perp}}(\widetilde{\Sigma_k})}^{2^{k+1} - 2} \right] \leq O_k(1) \eps' \left(\frac{\Delta}{w_{\min}c} \right)^{O_k(1)}
\]
Thus
\[
\widetilde{\E}\left[ \norm{\Gamma_{V^{\perp}}(\widetilde{\Sigma_k})}^{2} \right] \leq \eps'^{2^{-k}}O_k(1)  \left(\frac{\Delta}{w_{\min}c} \right)^{O_k(1)}
\]
It remains to note that 
\[
\Tr_{V^{\perp}}(M) = \widetilde{\E}\left[ \norm{\Gamma_{V^{\perp}}(\widetilde{\Sigma_k})}^{2} \right]
\]
and we are done.
\end{proof}

In the next key lemma, we prove that any vector that has nontrivial projection onto $V$ must also have nontrivial projection onto $\widetilde{\E}[\widetilde{\Sigma_i} \widetilde{\Sigma_i}^T]$ for some $i$.

\begin{lemma}\label{lemma:span}
Let $\widetilde{\E}$ be a pseudoexpectation of degree $C_k$ for some sufficiently large constant $C_k$ depending only on $k$ that solves $\mcl{S}$.  Consider the matrix
\[
N = \sum_{i=1}^k\widetilde{\E}[\widetilde{\Sigma_i} \widetilde{\Sigma_i}^T]
\]
where by this we mean we construct the $D \times D$ matrix whose entries are quadratic in the variables $\{ u \}, \{v \}$ and then take the entry-wise pseudoexpectation.  Then for any unit vector $z \in \R^D$,
\[
z^TNz \geq  \left( \frac{w_{\min}(z \cdot \Sigma_k)^{O_k(1)} -  O_k(1)\eps' \Delta^{O_k(1)}}{O_k(1)\Delta^{O_k(1)} }\right)^2
\]
as long as 
\[
w_{\min}(z \cdot \Sigma_k)^{O_k(1)} > O_k(1)\eps' \Delta^{O_k(1)}
\]
\end{lemma}
\begin{proof}
Using Lemma \ref{lemma:key-bound2}, we may write 
\[
w_k\prod_{i=1}^{k}(\Sigma_k(X) - \widetilde{\Sigma_i}(X) )^{2^{i-1}} \prod_{i=1}^{k-1} (\Sigma_k(X) - \Sigma_i(X) )^{2^{k + i-1}} = P_1(X)(\widetilde{h_1}(X) - h_1(X)) + \dots + P_{m}(X)(\widetilde{h_{m}}(X) - h_m(X))
\]
where $m = O_k(1)$.

Using the same method as the proof in Lemma \ref{lemma:trace-bound}, we have
\[
\widetilde{\E}\left[\norm{v\left(P_1(X)(\widetilde{h_1}(X) - h_1(X)) + \dots + P_{m}(X)(\widetilde{h_{m}}(X) - h_m(X))\right)}^2\right]  \leq O_k(1)\eps' \Delta^{O_k(1)}
\]
Now by Claim \ref{claim:factor-lower-bound},
\begin{align*}
&\widetilde{\E}\left[ \norm{v\left(w_k\prod_{i=1}^{k}(\Sigma_k(X) - \widetilde{\Sigma_i}(X) )^{2^{i-1}} \prod_{i=1}^{k-1} (\Sigma_k(X) - \Sigma_i(X) )^{2^{k + i-1}}\right)}^2 \right] \\ & \geq w_k\widetilde{\E}\left[ \prod_{i=1}^{k}\left(\norm{v\left(\Sigma_k(X) - \widetilde{\Sigma_i}(X)\right)}^2\right )^{2^{i-1}} \prod_{i=1}^{k-1} \left(\norm{\left(\Sigma_k(X) - \Sigma_i(X)\right)}^2 \right)^{2^{k + i-1}} \right] \\ & \geq 
w_k\widetilde{\E}\left[ \prod_{i=1}^k\left((z \cdot \Sigma_k - z \cdot \widetilde{\Sigma_i} )^2 \right)^{2^{i-1}} c^{O_k(1)}\right]
\end{align*}
where the second inequality is true because 
\[
\norm{v\left(\Sigma_k(X) - \widetilde{\Sigma_i}(X)\right)}^2 \succeq_{SOS} (z \cdot \Sigma_k - z \cdot \widetilde{\Sigma_i} )^2
\]
Now we claim
\begin{align*}
\widetilde{\E}\left[ \prod_{i=1}^k\left((z \cdot \Sigma_k - z \cdot \widetilde{\Sigma_i} )^2 \right)^{2^{i-1}} \right] \geq  \left( z \cdot \Sigma_k \right)^{O_k(1)} - O_k(1)\Delta^{O_k(1)} \sqrt{\widetilde{\E}\left[ \sum_{i} (z \cdot \widetilde{\Sigma_i})^2 \right]}
\end{align*}
To see this, first recall that $z \cdot \Sigma_k$ is just a constant.  Next, we can expand the LHS into a sum of monomials in the $z \cdot \widetilde{\Sigma_i}$.  In particular, we can write the expansion in the form
\[
\left( z \cdot \Sigma_k \right)^{O_k(1)} + \sum_{i} (z \cdot \widetilde{\Sigma_i})P_i(z \cdot \widetilde{\Sigma_1}, \dots , z \cdot \widetilde{\Sigma_k})
\]
wjhere $P$ is some polynomial in $k$ variables.  We can upper bound the coefficients of the polynomial in terms of $\Delta, k$ and we also know that 
\[
(z \cdot \widetilde{\Sigma_i})^2 \preceq_{SOS} O_k(1)\Delta^{O(1)}
\]
due to the constraints in our system.  Thus, we can bound the pseudoexpectation
\[
-\widetilde{\E}\left[\sum_{i} (z \cdot \widetilde{\Sigma_i})P_i(z \cdot \widetilde{\Sigma_1}, \dots , z \cdot \widetilde{\Sigma_k}) \right] \leq O_k(1)\Delta^{O_k(1)} \sqrt{\widetilde{\E}\left[ \sum_{i} (z \cdot \widetilde{\Sigma_i})^2 \right]}
\]
via Cauchy Schwarz.  Putting everything together the same way as in Lemma \ref{lemma:trace-bound}, we deduce
\[
\widetilde{\E}\left[ \sum_{i} (z \cdot \widetilde{\Sigma_i})^2 \right] \geq \left( \frac{w_{\min}(z \cdot \Sigma_k)^{O_k(1)} -  O_k(1)\eps' \Delta^{O_k(1)}}{O_k(1)\Delta^{O_k(1)} }\right)^2
\]
and now we are done.
\end{proof}

Putting Lemmas \ref{lemma:trace-bound} and \ref{lemma:span} together, we now prove that $V$ is essentially contained within the span of the union of the top principal components of $\widetilde{\E}[\widetilde{\Sigma_i} \widetilde{\Sigma_i}^T]$ over all $i$.

\begin{lemma}\label{lemma:solve-covariance}
For each $i$, let $M_i$ be the $D \times D$ matrix given by
\[
M_i = \widetilde{\E}[\widetilde{\Sigma_i}\widetilde{\Sigma_i^T}].
\]
Assume that for a sufficiently small function $f$ depending only on $k$,
\begin{align*}
\Delta \leq \eps'^{-f(k)} \\
w_{\min}, c \geq \eps'^{f(k)}
\end{align*}
Let $V_i$ be the subspace spanned by the top $k$ singular vectors of $M_i$.  Then for all $i$, the projection of  the true covariance matrix $\Sigma_i$ onto the orthogonal complement of $\spn (V_1, \dots , V_k)$ has length at most $\eps'^{\Omega_k(1)}$.
\end{lemma}
\begin{proof}
Assume for the sake of contradiction that the desired statement is false for $\Sigma_i$.  Let $z$ be the projection of $\Sigma_i$ onto the orthogonal complement of $\spn (V_1, \dots , V_k)$.  By Lemma \ref{lemma:span}, 
\begin{equation}\label{eq:norm-lowerbound}
\sum_j z^TM_jz \geq \left( \frac{w_{\min}(z \cdot \Sigma_i)^{O_k(1)} -  O_k(1)\eps' \Delta^{O_k(1)}}{O_k(1)\Delta^{O_k(1)} }\right)^2
\end{equation}
so there is some $j$ for which
\[
z^TM_j z \geq \frac{1}{k}\left( \frac{w_{\min}(z \cdot \Sigma_i)^{O_k(1)} -  O_k(1)\eps' \Delta^{O_k(1)}}{O_k(1)\Delta^{O_k(1)} }\right)^2
\]
On the other hand, Lemma \ref{lemma:trace-bound} implies that the sum of the singular values of $M_j$ outside the top $k$ is at most 
\[
\eps'^{2^{-k}}O_k(1)  \left(\frac{\Delta}{w_{\min}c} \right)^{O_k(1)}
\]
Since $z$ is orthogonal to the span of the top-$k$ singular vectors of $M_j$, we get 
\begin{equation}\label{eq:norm-upperbound}
z^T M_j z \leq \eps'^{2^{-k}}O_k(1)  \left(\frac{\Delta}{w_{\min}c} \right)^{O_k(1)} \norm{z}_2^2
\end{equation}
Note $z \cdot \Sigma_i = \norm{z}_2^2$ since $z$ is a projection of $\Sigma_i$ onto a subspace.  Now combining (\ref{eq:norm-lowerbound}) and (\ref{eq:norm-upperbound}) we get a contradiction unless 
\[
\norm{z}_2 \leq \eps'^{\Omega_k(1)}
\]
\end{proof}

\subsubsection{Finishing Up: Finding the Covariances and then the Means}\label{sec:solve-means}
Now we can brute-force search over the subspace spanned by the union of the top $k$ singular vectors of $M_1, \dots , M_k$.  Note that the SOS system $\mcl{S}$ is clearly feasible as it is solved when the $u_i, v_i$ form orthonormal bases for the true subspaces and the $\widetilde{w_i}, A_i, B_i$ are within $\eps'^{O_k(1)}$ of the true values (i.e. the values needed to express the true means and covariances in the orthonormal basis given by the $u_i,v_i$).

Thus, brute forcing over an $\eps'^{O_k(1)}$-net for the $\widetilde{w_i}, A_i, B_i$, we will find a feasible solution.  By Lemma \ref{lemma:span} and Lemma \ref{lemma:solve-covariance}, once we find any feasible solution, we will be able to obtain a set of $(1/\eps')^{O_k(1)}$ estimates at least one of which, say $\overline{\Sigma_1}, \dots , \overline{\Sigma_k}$, satisfies
\[
\norm{\Sigma_i - \overline{\Sigma_i}}_2^2 \leq \eps'^{\Omega_k(1)}
\]
for all $i$.  With these estimates we will now solve for the means.  Note we can assume that our covariance estimates are exactly correct because we can pretend that the true mixture is actually $N(\mu_1, \overline{\Sigma_1}), \dots , N(\mu_k, \overline{\Sigma_k})$ and our estimates for the Hermite polynomials of this mixture will be off by at most $O_k(1)\eps'^{\Omega_k(1)}$.  Thus, making this assumption will only affect the dependence on $\eps'$ that we get at the end.  From now on we can write $\Sigma_i$ to denote the true covariances and treat these as known quantities.
\\\\
Now we set up the same system as in Section \ref{sec:sos-setup} except we no longer have the variables $v_1, \dots , v_k$ and no longer have the $\widetilde{\Sigma_i}$.  These will instead be replaced by real values from $\Sigma_i$.  Formally:

\begin{definition}[SOS program for  learning means]\label{def:sos-means-only}
We will have the following variables
\begin{itemize}
    \item $u_1 = (u_{11}, \dots , u_{1d}) , \dots, u_k = (u_{k1}, \dots, u_{kd})$
\end{itemize}
In the above $u_1, \dots, u_k \in \R^d$.  We guess coefficients $a_{ij}$ where $i,j \in [k]$ expressing the means in this orthonormal basis.  We ensure that the guesses satisfy the property that for every pair of vectors $A_i = (a_{i1}, \dots , a_{ik}), A_j = (a_{j1}, \dots , a_{jk})$ either $A_i = A_j$ or 
\[
\norm{A_i - A_j}_2 \geq \frac{c}{2} 
\]
We ensure that
\[
\norm{A_i}_2 \leq 2\Delta
\]
We also guess the mixing weights $\widetilde{w_1}, \dots , \widetilde{w_k}$ and ensure that our guesses are all at least $w_{\min}/2$.

Now we set up the constraints.  Let $C$ be a sufficiently large integer depending only on $k$.  Define $\widetilde{\mu_i} = a_{i1}u_1 + \dots + a_{ik}u_k$.  These are linear expressions in the variables that we are solving for.  Now consider the hypothetical mixture with mixing weights $\widetilde{w_i}$, means $\widetilde{\mu_i}$, and covariances $I + \Sigma_i$.  The Hermite polynomials for this hypothetical mixture $\widetilde{h_i}(X)$ can be written as formal polynomials in $X = (X_1, \dots , X_d)$ with coefficients that are polynomials in $u$.  Note that we can explicitly write down these Hermite polynomials.  The set of constraints for our SOS system is as follows:
\begin{itemize}
    \item $\norm{u_i}_2^{2} = 1$ for all $1 \leq i \leq k$
    \item $u_i \cdot u_j = 0$ for all $i \neq j$
    \item For all $p = 1,2, \dots , C$
    \[
    \norm{v(\widetilde{h_p}(X) - \overline{h_p}(X))}^2 \leq 100\epsilon'
    \]
\end{itemize}

\end{definition}

Now we can repeat the same arguments from Section \ref{sec:solve-covariances} to prove that once we find a feasible solution, we can recover the span of the $\mu_i$.  The important generating functions are 
\begin{align*}
F(y) = \sum_{i = 1}^k w_i e^{\mu_i(X)y + \frac{1}{2}\Sigma_i(X)y^2} \\
\widetilde{F}(y) = \sum_{i = 1}^k \widetilde{w_i} e^{\widetilde{\mu_i}(X)y + \frac{1}{2} \Sigma_i(X)y^2}
\end{align*}
Define the differential operators as before except with $\widetilde{\Sigma_i}$ replaced with $\Sigma_i$.  Let $\mcl{D}_i$ denote the differential operator $(\partial - (\mu_i(X) + \Sigma_i(X)y))$ and $\widetilde{\mcl{D}_i}$ denote the differential operator $(\partial - (\widetilde{\mu_i}(X) + \Sigma_i(X)y))$.  All derivatives are taken with respect to $y$.  The two key differential operators to consider are 
\begin{align*}
\widetilde{\mcl{D}_{k}}^{2^{2k-1} - 2^{k-1} - 1}\widetilde{\mcl{D}_{k-1}}^{2^{2k-2}} \dots \widetilde{\mcl{D}_{1}}^{2^k} \mcl{D}_k^{2^{k-1}} \dots \mcl{D}_1^{1} \\ 
\mcl{D}_{k}^{2^{2k-1} - 2^{k-1} - 1}\mcl{D}_{k-1}^{2^{2k-2}} \dots\mcl{D}_{1}^{2^k} \widetilde{\mcl{D}_k}^{2^{k-1}} \dots \widetilde{\mcl{D}_1}^{1} 
\end{align*}
Note the change to $2^{2k-1} - 2^{k-1} - 1$ from $2^{2k-1} - 1$ in the exponent of the first term.  This is because when operating on $ P(y)e^{\mu_k(X)y + \frac{1}{2} \Sigma_k(X)y^2}$ for some polynomial $P$, the operator $\mcl{D}_k$ reduces the degree of $P$ by $1$ while the operator $\widetilde{D_k}$ does not change the degree of $P$ (whereas before this operator increased the degree of the formal polynomial $P$).  Similar to Claim \ref{claim:elimination1} and Claim \ref{claim:elimination2} in Section \ref{sec:solve-covariances}, we have
\begin{claim}\label{claim:mean-elimination1}
Write
\[
\widetilde{\mcl{D}_{k}}^{2^{2k-1} - 2^{k-1} - 1}\widetilde{\mcl{D}_{k-1}}^{2^{2k-2}} \dots \widetilde{\mcl{D}_{1}}^{2^k} \mcl{D}_k^{2^{k-1}} \dots \mcl{D}_1^{1}(\widetilde{F})
\]
as a formal power series in $y$.  Its evaluation at $y = 0$ is 
\[
C_k\widetilde{w_k}(\widetilde{\mu_k}(X) - \mu_k(X))^{2^{k-1}}\prod_{i=1}^{k-1} (\Sigma_k(X) - \Sigma_i(X) )^{2^{i-1}}\prod_{i=1}^{k-1}(\Sigma_k(X) - \Sigma_i(X) )^{2^{k+i-1}}
\]
where $C_k$ is a constant depending only on $k$.
\end{claim}

\begin{claim}\label{claim:mean-elimination2}
Write
\[
\mcl{D}_{k}^{2^{2k-1} - 2^{k-1} - 1}\mcl{D}_{k-1}^{2^{2k-2}} \dots \mcl{D}_{1}^{2^k} \widetilde{\mcl{D}_k}^{2^{k-1}} \dots \widetilde{\mcl{D}_1}^{1}(F)
\]
as a formal power series in $y$.  Its evaluation at $y = 0$ is 
\[
C_k w_k(\widetilde{\mu_k}(X) - \mu_k(X))^{2^{k-1}}\prod_{i=1}^{k-1} (\Sigma_k(X) - \Sigma_i(X) )^{2^{i-1}}\prod_{i=1}^{k-1}(\Sigma_k(X) - \Sigma_i(X) )^{2^{k+i-1}}
\]
where $C_k$ is a constant depending only on $k$.

\end{claim}

Now repeating the arguments in Lemmas \ref{lemma:key-bound1}, \ref{lemma:key-bound2}, \ref{lemma:trace-bound},\ref{lemma:span}, \ref{lemma:solve-covariance}, we can prove that for any feasible solution, the subspace spanned by the top $k$ singular vectors of each of $\widetilde{\E}[\widetilde{\mu_1}\widetilde{\mu_1}^T], \dots , \widetilde{\E}[\widetilde{\mu_k}\widetilde{\mu_k}^T]$ approximately contains all of $\mu_1, \dots , \mu_k$.  We can now brute force search over this subspace (and since we are already brute-force searching over the mixing weights), we will output some set of candidate components that are close to the true components.

\subsubsection{All Pairs of Parameters are Equal or Separated}\label{sec:equal-or-sep}
In the case where some pairs of parameters may be equal (but pairs $(\mu_i, \Sigma_i)$ and $(\mu_j, \Sigma_j)$ cannot be too close), we can repeat essentially the same arguments from the previous section but with minor adjustments in the number of times we are applying each differential operator.

We can assume that our guesses for the coefficients $A_i, B_i$ satisfy the correct equality pattern in the sense that $A_i = A_j$ if and only if $\mu_i = \mu_j$ and otherwise $\norm{A_i - A_j} \geq c/2$ and similar for the parameters $B_i$.  This is because there are only $O_k(1)$ different equality patterns.

Now without loss of generality let $\{\Sigma_{1}, \dots , \Sigma_{j} \}$ ($j < k$) be the set of covariance matrices that are equal to $\Sigma_k$.  The key differential operators to consider are 
\begin{align*}
\widetilde{\mcl{D}_{k}}^{2^{2k-1} - 1 - 2^k - \dots - 2^{k + j} }\widetilde{\mcl{D}_{k-1}}^{2^{2k-2}} \dots \widetilde{\mcl{D}_{1}}^{2^k} \mcl{D}_k^{2^{k-1}} \dots \mcl{D}_1^{1} \\ 
\mcl{D}_{k}^{2^{2k-1} - 1 - 2^{0} - \dots - 2^j }\mcl{D}_{k-1}^{2^{2k-2}} \dots\mcl{D}_{1}^{2^k} \widetilde{\mcl{D}_k}^{2^{k-1}} \dots \widetilde{\mcl{D}_1}^{1} 
\end{align*}

Similar to Claim \ref{claim:elimination1} and Claim \ref{claim:elimination2}, we get
\begin{claim}\label{claim:general-elimination1}
Let $\{\Sigma_{1}, \dots , \Sigma_{j} \}$ ($j < k$) be the set of covariance matrices that are equal to $\Sigma_k$.  Note this also implies $\{ \widetilde{\Sigma_{1}}, \dots , \widetilde{\Sigma_{j}} \}$ are precisely the subset of $\{\widetilde{\Sigma_i} \}$ that are equal to $\widetilde{\Sigma_k}$. Write
\[
\widetilde{\mcl{D}_{k}}^{2^{2k-1}  - 1 - 2^k - \dots - 2^{k + j - 1}}\widetilde{\mcl{D}_{k-1}}^{2^{2k-2}} \dots \widetilde{\mcl{D}_{1}}^{2^k} \mcl{D}_k^{2^{k-1}} \dots \mcl{D}_1^{1}(\widetilde{F})
\]
as a formal power series in $y$.  Its evaluation at $y = 0$ is 
\[
C_k\widetilde{w_k}\prod_{i=1}^k (\widetilde{\Sigma_k}(X) - \Sigma_i(X) )^{2^{i-1}}\prod_{i=1}^{j}(\widetilde{\mu_k}(X) - \widetilde{\mu_i}(X))^{2^{k+i-1}}\prod_{i=j + 1}^{k-1}(\widetilde{\Sigma_k}(X) - \widetilde{\Sigma_i}(X) )^{2^{k+i-1}}
\]
where $C_k$ is a constant depending only on $k$.
\end{claim}

\begin{claim}\label{claim:general-elimination2}
Let $\{\Sigma_{1}, \dots , \Sigma_{j} \}$ ($j < k$) be the set of covariance matrices that are equal to $\Sigma_k$.  Note this also implies $\{ \widetilde{\Sigma_{1}}, \dots , \widetilde{\Sigma_{j}} \}$ are precisely the subset of $\{\widetilde{\Sigma_i} \}$ that are equal to $\widetilde{\Sigma_k}$. Write
\[
\mcl{D}_{k}^{2^{2k-1}  - 1 - 2^k - \dots - 2^{k + j - 1}}\mcl{D}_{k-1}^{2^{2k-2}} \dots \mcl{D}_{1}^{2^k} \widetilde{\mcl{D}_k}^{2^{k-1}} \dots \widetilde{\mcl{D}_1}^{1}(F)
\]
as a formal power series in $y$.  Its evaluation at $y = 0$ is 
\[
C_k\widetilde{w_k}\prod_{i=1}^k (\Sigma_k(X) - \widetilde{\Sigma_i}(X) )^{2^{i-1}}\prod_{i=1}^{j}(\mu_k(X) - \mu_i(X))^{2^{k+i-1}}\prod_{i=j + 1}^{k-1}(\Sigma_k(X) - \Sigma_i(X) )^{2^{k+i-1}}
\]
where $C_k$ is a constant depending only on $k$.
\end{claim}

Now we can repeat the arguments in Lemmas \ref{lemma:key-bound1}, \ref{lemma:key-bound2}, \ref{lemma:trace-bound},\ref{lemma:span}, \ref{lemma:solve-covariance}.  The key point is that the constraints in our SOS program give explicit values for 
\begin{align*}
\norm{v(\widetilde{\mu_i}(X) - \widetilde{\mu_j}(X))}^2 \\
\norm{v(\widetilde{\Sigma_i}(X) - \widetilde{\Sigma_j}(X))}^2
\end{align*}
in terms of $A_i, B_i$ (which are explicit real numbers). We can then repeat the arguments in Section \ref{sec:solve-means} (with appropriate modifications to the number of times we apply each differential operator) to find the means.

\section{Robust Moment Estimation}\label{sec:moment-est}
In Section \ref{sec:close-case}, we showed how to learn the parameters of a mixture of Gaussians $\mcl{M}$ with components that are not too far apart when we are given estimates for the Hermite polynomials.  In this section, we show how to estimate the Hermite polynomials from an $\eps$-corrupted sample.  Putting the results together, we will get a robust learning algorithm in the case when the components are not too far apart.  

While the closeness of components in Section \ref{sec:close-case} is defined in terms of parameter distance, we will need to reason about TV-distance between components in order to integrate our results into our full learning algorithm.  We begin with a definition.
\begin{definition}\label{def:well-conditioned}
We say a mixture of Gaussians $w_1G_1 + \dots  + w_kG_k$ is $\delta$-well-conditioned if
\begin{enumerate}
    \item Let $\mcl{G}$ be the graph on $[k]$ obtained by connecting two nodes $i,j$ if $d_{\TV}(G_i, G_j) \leq 1 - \delta$.  Then $\mcl{G}$ is connected
    \item $d_{\TV}(G_i, G_j) \geq \delta $ for all $i \neq j$
    \item $w_{\min} \geq \delta$
\end{enumerate}
\end{definition}

The main theorem that we will prove in this section is as follows.
\begin{theorem}\label{thm:full-close-case}
There is a function $f(k) > 0$ depending only on $k$ such that given an $\eps$-corrupted sample from a $\delta$-well-conditioned mixture of Gaussians 
\[
\mcl{M} = w_1N(\mu_1, \Sigma_1) + \dots + w_kN(\mu_k, \Sigma_k)
\]
where $ \delta \geq \eps^{f(k)}$, there is a polynomial time algorithm that outputs a set of $(1/\eps)^{O_k(1)}$ candidate mixtures $\{  \widetilde{w_1N}(\widetilde{\mu_1}, \widetilde{\Sigma_1}) + \dots + \widetilde{w_k}N(\widetilde{\mu_k}, \widetilde{\Sigma_k} \}$ and with high probability, at least one of them satisfies that for all $i$:
\[
\abs{w_i - \widetilde{w_i}} + d_{\TV}(N(\mu_i, \Sigma_i), N(\widetilde{\mu_i}, \widetilde{\Sigma_i})) \leq \poly(\eps)
\]
\end{theorem}

\subsection{Distance between Gaussians}
As mentioned earlier, we will first introduce a few tools for relating parameter distance and TV distance between Gaussians.

The following is a standard fact.
\begin{claim}\label{claim:parameter-dist}
For two Gaussians $N(\mu_1, \Sigma_1), N(\mu_2, \Sigma_2)$
\[
d_{\TV}(N(\mu_1, \Sigma_1), N(\mu_2, \Sigma_2)) = O\left( \left( (\mu_1 - \mu_2)^T\Sigma_1^{-1}(\mu_1 - \mu_2)\right)^{1/2} + \norm{\Sigma_1^{-1/2}\Sigma_2\Sigma_1^{-1/2} - I}_F\right)
\]
\end{claim}
\begin{proof}
See e.g. Fact 2.1 in \cite{kane2020robust}.
\end{proof}

Next, we will prove a bound in the opposite direction, that when Gaussians are not too far apart in TV distance, then their parameters also cannot be too far apart.
\begin{lemma}\label{lem:parameter-closeness}
Let $\mcl{M}$ be a mixture of $k$ Gaussians that is $\delta$-well conditioned.  Let $\Sigma$ be the covariance matrix of the mixture.  Then
\begin{enumerate}
    \item $\Sigma_i \leq  \poly(\delta)^{-1}\Sigma$ for all components of the mixture
    \item $\Sigma_i \geq \poly(\delta) \Sigma$ for all components of the mixture
    \item For any two components $i,j$, we have
    $\norm{\Sigma^{-1/2}(\mu_i - \mu_j)} \leq \poly(\delta)^{-1}$
    \item For any two components $i,j$, we have
    $ \norm{\Sigma^{-1/2}(\Sigma_i - \Sigma_j)\Sigma^{-1/2}}_2 \leq \poly(\delta)^{-1}$
    
\end{enumerate}
where the coefficients and degrees of the polynomials may depend only on $k$.  
\end{lemma}
\begin{proof}
The statements are invariant under linear transformations so without loss of generality let $\Sigma = I$. Assume for the sake of contradiction that the first condition is failed.  Then there is some direction $v$ such that say 
\[
v^T\Sigma_1 v \geq \delta^{-10k}
\]
There must be some $i \in [k]$ such that $v^T \Sigma_i v \leq 1$ since otherwise the variance  of the mixture in direction $v$ would be bigger than $1$.  Now we claim that $i$ and $1$ cannot be connected in $\mcl{G}$, the graph defined in Definition \ref{def:well-conditioned}.  To see this, if they were connected, then there must be two vertices $j_1, j_2$ that are consecutive along the path between $1$ and $i$ such that
\[
\frac{v^T\Sigma_{j_1}v}{v^T\Sigma_{j_2}v} \geq \delta^{-10}
\]
But then $d_{\TV}(G_{j_1}, G_{j_2}) \geq 1 - \delta$.  To see this, let $\sqrt{v^T\Sigma_{j_2}v} = c$.  We can project both Gaussians onto the direction $v$ and note that the Gaussian $G_{j_1}$ is spread over width $\delta^{-5}c$ whereas the Gaussian $G_{j_2}$ is essentially contained in a strip of width $O(\log 1/\delta)c$.
\\\\
Now we may assume that the first condition is satisfied.  Now we consider when the third condition is failed. Assume that
\[
\norm{(\mu_i - \mu_j)} \geq k\delta^{-20k}
\]
Now let $v$ be the unit vector in direction $\mu_i - \mu_j$.  Projecting the Gaussians $G_i, G_j$ onto direction $v$ and considering the path between them, we must find $j_1,j_2$ that are connected such that 
\[
\norm{(\mu_{j_1} - \mu_{j_2})} \geq \delta^{-20k}
\]
Now, using the fact that the first condition must be satisfied (i.e. $v^T\Sigma_{j_1} v, v^T\Sigma_{j_2} v \leq \delta^{-10k}$) we get that $d_{\TV}(G_{j_1}, G_{j_2}) \geq 1 - \delta$, a contradiction.
\\\\
Now we may assume that the first and third conditions are satisfied.  Assume now that the second condition is not satisfied.  Without loss of generality, there is some vector $v$ such that 
\[
v^T\Sigma_1 v \leq (\delta/k)^{10^2k}
\]
If there is some component $i$ such that 
\[
v^T\Sigma_i v \geq (\delta/k)^{50k}
\]
then comparing the Gaussians along the path between $i$ and $1$ in the graph $\mcl{G}$, we get a contradiction.  Thus, we now have 
\[
v^T\Sigma_i v \leq (\delta/k)^{50}
\]
for all components.  Note that the covariance of the entire mixture is the identity.  Thus, there must be two components with
\[
\abs{v \cdot \mu_i - v \cdot \mu_j} \geq \frac{1}{2k}.
\]
Taking the path between $i$ and $j$, we must be able to find two consecutive vertices $j_1, j_2$ such that 
\[
\abs{v \cdot \mu_{j_1} - v \cdot \mu_{j_2}} \geq \frac{1}{2k^2}.
\]
However, we then get $d_{\TV}(G_{j_1}, G_{j_2}) > 1 - \delta$, a contradiction.
\\\\
Now we consider when the first three conditions are all satisfied.  Using the first two conditions, we have bounds on the smallest and largest singular value of $\Sigma_i^{1/2}\Sigma_j^{-1/2}$ for all $i,j$.  Thus, 
\[
\norm{\Sigma_i - \Sigma_j}_2 \leq \poly(\delta)^{-1}\norm{I - \Sigma_i^{-1/2}\Sigma_j\Sigma_i^{-1/2}}_2
\]
for all $i,j$.  However if for some $i,j$ that are connected in $\mcl{G}$, we have 
 \[
\norm{(\Sigma_i - \Sigma_j)}_2 \geq (k/\delta)^{10^4}
\]
then we would have
\[
\norm{I - \Sigma_i^{-1/2}\Sigma_j\Sigma_i^{-1/2}}_2 \geq (k/\delta)^{10^3}
\]
and this would contradict the assumption that $d_{\TV}(G_i, G_j) \leq 1 - \delta$ (this follows from the same argument as in Lemma 3.2 of \cite{kane2020robust}).  Now using triangle inequality along each path, we deduce that for all $i,j$
\[
\norm{(\Sigma_i - \Sigma_j)}_2 \leq (k/\delta)^{10^5}
\]
completing the proof.
\end{proof}

As a corollary to the previous lemma, in a $\delta$-well conditioned mixture, all component means and covariances are close to the mean and covariance of the overall mixture.
\begin{corollary}\label{corollary-indiv-closeness}
Let $\mcl{M}$ be a mixture of $k$ Gaussians that is $\delta$-well conditioned.  Let $\mu,\Sigma$ be the mean and covariance matrix of the mixture.  Then we have for all $i$
\begin{itemize}
    \item   $\norm{\Sigma^{-1/2}(\mu - \mu_i)}_2 \leq \poly(\delta)^{-1}$
    \item    $\norm{\Sigma^{-1/2}(\Sigma - \Sigma_i)\Sigma^{-1/2}}_2 \leq \poly(\delta)^{-1}$
\end{itemize}
\end{corollary}
\begin{proof}
The statement is invariant under linear transformation so we may assume $\Sigma = I$ and $\mu = 0$.  Then noting 
\[
\mu_i = \mu + w_1(\mu_i - \mu_1) + \dots + w_k( \mu_i - \mu_k)
\]
and using Lemma \ref{lem:parameter-closeness}, we have proved the first part.  Now for the second part, note $\Sigma = \sum_{i=1}^k w_i(\Sigma_i + \mu_i\mu_i^T)$
and hence we have
\[
\Sigma = \Sigma_i +  w_1(\Sigma_1 - \Sigma_i) + \dots + w_k(\Sigma_k - \Sigma_i) + \sum_{i=1}^k w_i \mu_i \mu_i^T 
\]
and using Lemma \ref{lem:parameter-closeness} and the first part, we are done.

\end{proof}

\subsection{Hermite Polynomial Estimation}

Now we show how to estimate the Hermite polynomials of a $\delta$-well-conditioned mixture $\mcl{M}$ if we are given an $\eps$-corrupted sample (where $\delta \geq \eps^{f(k)}$ for some sufficiently small function $f(k) > 0$ depending only on $k$).  Our algorithm will closely mirror the algorithm in \cite{kane2020robust}.  

The first step will be to show that we can robustly estimate the mean and covariance of the mixture $\mcl{M}$ and then we will use these estimates to compute a linear transformation to place the mixture in isotropic position.
\begin{lemma}\label{lem:isotropic-position}
There is a sufficiently small function $f(k)$ depending only on $k$ such that given a $\eps$-corrupted sample from a $\delta$-well-conditioned mixture of Gaussians $\mcl{M}$ with true mean and covariance $\mu, \Sigma$ respectively, where $ \delta \geq \eps^{f(k)}$, then with high probability we can output estimates $\wh{\mu}$ and $\wh{\Sigma}$ such that 
\begin{enumerate}
    \item 
    $\norm{\Sigma^{-1/2}(\wh{\Sigma} - \Sigma)\Sigma^{-1/2}}_2 \leq  \eps^{\Omega_k(1)}$
    \item
    $ \norm{\Sigma^{-1/2}(\wh{\mu} - \mu)}_2 \leq \eps^{\Omega_k(1)}$
\end{enumerate}
\end{lemma}
\begin{proof}
This can be proven using a similar argument to Proposition 4.1 in \cite{kane2020robust}. First we will estimate the covariance of the mixture.  Note that the statement is invariant under linear transformation (and the robust estimation algorithtm that we will use, Theorem 2.4 in \cite{kane2020robust}, is also invariant under linear transformation), so it suffices to consider when $\Sigma = I$.  Let the components of the mixture be $G_1, \dots , G_k$.  Note that by pairing up our samples, we have access to a $2\eps$-corrupted sample from the distribution $\mcl{M} - \mcl{M'}$ (i.e. the difference of two independent samples from $\mcl{M}$).  For each such sample say $Y \sim \mcl{M} - \mcl{M'}$, $\Sigma =  0.5 \E[YY^T]$.  We will now show that $Z = YY^T$ where $Z$ is flattened into a vector, has bounded covariance.  Note that we can view $Y$ as being sampled from a mixture of $O(k^2)$ Gaussians $ G_i - G_j$ (where we may have $i = j$).  We now prove that 
\begin{itemize}
\item For $Y \sim G_i - G_j$ and $Z = YY^T$, $\E[Z \otimes Z] - \E[Z] \otimes \E[Z] \leq \poly(\delta)^{-1}I$ 
\item For $Y \sim G_i - G_j, Y' \sim G_{i'} - G_{j'}$ and $Z = YY^T, Z' = Y'Y'^T$, $\norm{\E[Z - Z']}_2^2 = \poly(\delta)^{-1}$
\end{itemize}
Using Lemma \ref{lem:parameter-closeness} and Corollary \ref{corollary-indiv-closeness}, we have $\poly(\delta)^{-1}$ bounds on $\norm{\mu_i}_2$, $\norm{\Sigma_i}_{\textsf{op}}$ and $\norm{\Sigma_i - \Sigma_j}_2$ for all $i,j$.  We can now follow the same argument as Proposition 4.1 in \cite{kane2020robust} to bound the above two quantities.  With these bounds, by Theorem 2.4 in \cite{kane2020robust}, we can robustly estimate the covariance.  Once we have an estimate for the covariance $\wh{\Sigma}$, we can apply the linear transformation $\wh{\Sigma}^{-1/2}$ and robustly estimate the mean (which now has covariance close to identity).    
\end{proof}

Using the above, we can place our mixture in isotropic position.  This mirrors Proposition 4.2 in \cite{kane2020robust}.

\begin{corollary}\label{corollary:isotropic-position}
There is a sufficiently small function $f(k)$ depending only on $k$ such that given a $\eps$-corrupted sample from a $\delta$-well-conditioned mixture of Gaussians $\mcl{M} = w_1G_1 + \dots + w_kG_k$ with mean and covariance $\mu, \Sigma$ where $ \delta \geq \eps^{f(k)}$, there is a polynomial time algorithm that with high probability outputs an invertible linear transformation $L$ so that
\begin{enumerate}
    \item $\norm{L(\mu)}_2 \leq \poly(\eps)$
    \item $\norm{I - L(\Sigma)}_2 \leq \poly(\eps)$
\end{enumerate}
\end{corollary}
\begin{proof}
We can first obtain estimates $\wh{\mu}$ and $\wh{\Sigma}$ using Lemma \ref{lem:isotropic-position}.  We can then apply the linear transformation 
\[
L(x) = \wh{\Sigma}^{-1/2}(x - \wh{\mu})
\]
It follows from direct computation that this transformation satisfies the desired properties.
\end{proof}

Once our mixture is placed in isotropic position, we will estimate the Hermite polynomials and then we will be able to use Theorem \ref{thm:main-close-case}.  The following lemma can be easily derived from the results in \cite{kane2020robust} (see Lemmas 2.7,2.8 and 5.2  there).  

\begin{lemma}\label{lemma:estimate-hermite}
Let $\mcl{M}$ be a mixture of Gaussians $w_1 N(\mu_1, I + \Sigma_1) + \dots + w_k N(\mu_k, I + \Sigma_k)$.  Then
\[
\norm{\E_{z \sim \mcl{M}}\left(v_X(H_{m}(X.z)) \otimes v_X(H_{m}(X.z)) \right)]}_2 = O_m(1 + \max{\norm{\Sigma_i}}_2 + \max{\norm{\mu_i}})^{2m}
\]
where $H_{m}(X,z)$ is defined as in definition \ref{def:Hermite-two-var} and $v_X(H_{m}(X.z))$ denotes vectorizing as a polynomial in $X$ so that the entries of the vector are polynomials in $z$.
\end{lemma}

Kane \cite{kane2020robust} works with Hermite polynomial tensors, which are tensorized versions of the Hermite polynomials we are using.  It is clear that these two notions are equivalent up to $O_k(1)$ factors as long as $m$ is $O_k(1)$ (writing them as formal polynomials instead of tensors simply collapses symmetric entries of the tensor but this collapses at most $O_m(1)$ entries together at once).

We can now combine everything in this section with Theorem \ref{thm:main-close-case} to complete the proof of Theorem \ref{thm:full-close-case}.

\begin{proof}[Proof of Theorem \ref{thm:full-close-case}]
We can split the samples into $O(1)$ parts that are each $O(1)\eps$ corrupted samples.  First, we use Corollary \ref{corollary:isotropic-position} to compute a transformation $L$ that places the mixture in nearly isotropic position.  Now Lemma \ref{lem:parameter-closeness} and Corollary \ref{corollary-indiv-closeness} gives us bounds on how far each of the means is from $0$ and how far each of the covariances is from $I$.  We can apply Lemma \ref{lemma:estimate-hermite} and standard results from robust estimation of bounded covariance distributions (see e.g. Theorem 2.2 in \cite{kane2020robust}) to obtain estimates $\overline{h_{m, L(\mcl{M})}}(X)$ for the Hermite polynomials of the mixture $L(\mcl{M})$ such that 
\[
\norm{v\left( \overline{h_{m, L(\mcl{M})}}(X) - h_{m, L(\mcl{M})}(X)\right) }_2 \leq \poly(\eps)
\]
where $m$ is bounded as a function of $k$.  Now we must verify that the remaining hypotheses of Theorem \ref{thm:main-close-case} are satisfied with $\eps' = \poly(\eps)$ for the transformed mixture $L(\mcl{M})$.  
\begin{itemize}
    \item Corollary \ref{corollary-indiv-closeness} gives the required upper bound on $\norm{L(\mu_i)}$ and $\norm{L( \Sigma_i) - I }$
    \item The first two conditions of Lemma \ref{lem:parameter-closeness}, combined with Claim \ref{claim:parameter-dist}, imply the condition that no pair of components has essentially the same mean and covariance
    \item Finally, the mixing weights are unchanged by the linear transformation so the third condition is easily verified (since the original mixture is $\delta$-well-conditioned)
\end{itemize}

Thus, by Theorem \ref{thm:main-close-case} we can obtain a list of $(1/\eps)^{O_k(1)}$ candidate mixtures at least one of which satisfies
\[
\norm{w_i - \widetilde{w_i}} + \norm{L(\mu_i) - \widetilde{\mu_i}}_2 + \norm{L(\Sigma_i) - \widetilde{\Sigma_i}}_2 \leq \poly(\eps)
\]
for all $i$.  By Claim \ref{claim:parameter-dist}, we know that the components we compute are close in TV to the true components.  Now applying the inverse transformation $L^{-1}$ to all of the components, we are done.

\end{proof}

\section{Rough Clustering}\label{sec:clustering}


As mentioned earlier in the proof overview, the first step in our full algorithm will be to cluster the points.  We present our clustering algorithm in this section.  This section closely mirrors the work in \cite{diakonikolas2020robustly}. We first define a measure of closeness between Gaussians that we will use throughout the paper.  
\begin{definition}\label{def:closeness}
We say that two Gaussians $N(\mu, \Sigma)$ and $N(\mu', \Sigma')$ are $C$-close if all of the following conditions hold
\begin{enumerate}
    \item  (mean condition) For all unit vectors $v \in \R^d$, we have $ (v \cdot \mu - v \cdot \mu')^2 \leq C v^T(\Sigma + \Sigma')v $
        \item  (variance condition) For all unit vectors $v \in \R^d$, we have $\max(v^T\Sigma v, v^T\Sigma' v) \leq  C\min(v^T\Sigma v, v^T\Sigma'v) $
        \item  (covariance condition) Finally, we have $\norm{I - \Sigma'^{-1/2}\Sigma\Sigma'^{-1/2}}_2^2  \leq C $
\end{enumerate}
\end{definition}

The main theorem that we aim to prove in this section is the following, which implies that if the true mixture can be well-clustered into submixtures, then we can recover this clustering with constant-accuracy.
\begin{theorem}\label{thm:rough-clustering}
Let $k, D, \gamma $ be parameters.  Assume we are given $\eps$-corrupted samples from a mixture of Gaussians $w_1G_1 + \dots  + w_kG_k$ where the mixing weights $w_i$ are all rational numbers with denominator bounded by a constant $A$.  Let $A_1, \dots , A_l$ be a partition of the components such that 
\begin{enumerate}
    \item For any $j_1, j_2$ in the same piece of the partition $G_{j_1}, G_{j_2}$ are $D$-close
    \item For any $j_1, j_2$ in different pieces of the partition, $G_{j_1}, G_{j_2}$ are not $D'$-close
\end{enumerate}
where $D'>  F(k,A, D, \gamma) $ for some sufficiently large function $F$.  Assume that $t >F(k,A, D, \gamma) $ and $\eta, \eps, \delta < f(k,A, D, \gamma)$ for some sufficiently small function $f$.  Then with probability at least $1 - \gamma$, if $X_1, \dots , X_n$ is an $\eps$-corrupted sample from the mixture $w_1G_1 + \dots + w_kG_k$ with $n \geq \poly(1/\eps, 1/\eta,  1/\delta, d)^{O(k,A)}$, then one of the clusterings returned by {\sc Rough Clustering} (see Algorithm \ref{alg:rough-cluster}) gives a $\gamma$-corrupted sample of each of the submixtures given by $A_1, \dots , A_l$.
\end{theorem}
\begin{remark}
Note that the last statement is well defined because the assumption about the partition essentially implies that all pairs of components in different submixtures are separated so $\gamma$-corrupted sample simply means correctly recovering a $1-\gamma$-fraction of the original points that were drawn from the corresponding submixture.
\end{remark}

In this section, it will suffice to consider when the mixing weights are equal as we can subdivide one component into many identical ones so from now on we assume $w_1 = \dots = w_k = 1/k$ and all dependencies on $A$ become dependencies on $k$. 


We begin with a few preliminaries.  The following claim is a simple consequence of the definition.

\begin{claim}\label{claim:dist-transitivity}
Let $G_1, G_2, G_3$ be Gaussians such that $G_1$ and $G_2$ are $C$-close and $G_2$ and $G_3$ are $C$-close.  Then $G_1$ and $G_3$ are $\poly(C)$-close.  
\end{claim}
\begin{proof}
The second condition follows immediately from the fact that $G_1$ and $G_2$ are $C$-close and $G_2$ and $G_3$ are $C$-close.  Now we know that for all vectors $v$, $v^T\Sigma_1v,v^T\Sigma_2v, v^T\Sigma_3v $ are all within a $\poly(C)$ factor of each other.  This means that the singular values of $\Sigma_i^{1/2}\Sigma_j^{-1/2}$ are all bounded above and below by $\poly(C)$.  From this and the triangle inequality, we get the first and third conditions.  
\end{proof}

The next claim follows immediately from Lemma 3.6 in \cite{diakonikolas2020robustly}.
\begin{claim}\label{claim:trivial-dist-bound}
There is a decreasing function $f$ such that $f(C) > 0$ for all $C > 0$ such that if two Gaussians $G_1, G_2$ are $C$-close then
\[
d_{\TV}(G_1, G_2) \leq 1 - f(C)
\]
\end{claim}

We will now show that either all pairs in the mixture are not too far apart, or there exists a nontrivial partition of the mixture into two parts that are separated in either mean, variance in some direction, or covariance.  This parallels Corollary 3.7 in \cite{diakonikolas2020robustly}.  However, we require a slightly different statement because their paper specializes to the case where all pairs of components are separated.  We use $\mu, \Sigma$ to denote the mean and covariance of the overall mixture.
\begin{claim}\label{claim:exists-partition}
Let $C > 100$ be a constant.  Let $C_k$ be a sufficiently large constant depending only on $C$ and $k$.  Assume that there are $i,j \in [k]$ such that $N(\mu_i, \Sigma_i)$ and $N(\mu_j, \Sigma_j)$ are not $C_k$-close.  Then there exists a partition of $[k]$ into two disjoint sets $S,T$ such that for any $a \in S, b \in T$, $N(\mu_a, \Sigma_a)$ is not $k^C$-close to $N(\mu_b, \Sigma_b)$ and at least one of the following holds:
\begin{enumerate}
    \item There is a direction $v$ such that for all $a \in S, b \in T$,
    \[
    ((\mu_a - \mu_b) \cdot v) \geq \max\left(  k^C(v^T(\Sigma_a + \Sigma_b)v ), \frac{v^T\Sigma v}{k^2} \right)
    \]
    \item There is a direction $v$ such that for all $a \in S, b \in T$,

    \[\frac{v^T\Sigma_a v}{v^T\Sigma_bv} \geq k^{C} \mbox{ and }\frac{v^T\Sigma_a v}{v^T \Sigma v} \geq \frac{1}{k^{2Ck}}
    \]
    \item We have
    \[
    \norm{I - \Sigma_a^{-1/2}\Sigma_b\Sigma_a^{-1/2}}^2 \geq k^C \max\left(\norm{\Sigma_a^{1/2}A_{ab}\Sigma_a^{1/2}},  \norm{\Sigma_b^{1/2}A_{ab}\Sigma_b^{1/2}}, \norm{\Sigma^{1/2}A_{ab}\Sigma^{1/2}}\right)
    \]
    where $A_{ab} = \Sigma_a^{-1/2}\left(I - \Sigma_a^{-1/2}\Sigma_b\Sigma_a^{-1/2}\right)\Sigma_a^{-1/2} $
\end{enumerate}
\end{claim}
\begin{proof}
We break into a few cases:

\paragraph{Case 1:}

Suppose that there is a $v$ such that for some $a,b$
\[
((\mu_a - \mu_b) \cdot v)^2 \geq 10k^{2} \cdot k^C \max_i(v^T\Sigma_iv)
\]
then we claim we are done.  To see this, first observe that
\[
v^T\Sigma v = \frac{1}{k^2} \sum_{i \neq j}((\mu_i - \mu_j) \cdot v)^2 + \frac{1}{k} \sum v^T\Sigma_i v
\]
so then choosing $a,b$ such that $((\mu_a - \mu_b) \cdot v)^2$ is maximal, we have $((\mu_a - \mu_b) \cdot v)^2 \geq 0.1 v^T\Sigma v$. 
Now we can partition the components based on the value of $\mu_i \cdot v$.  We can ensure that the gap between the clusters has size at least $\frac{(\mu_a - \mu_b)\cdot v}{k}$.  This will imply for all $a \in S, b \in T$
\[
((\mu_a - \mu_b) \cdot v)^2 \geq k^C v^T(\Sigma_a + \Sigma_b)v)
\]
i.e. the corresponding components are not $k^C$-close.  Since we can choose $C_k$ sufficiently large, the first condition is also satisfied and we are done in this case.
\\

\paragraph{Case 2:}
Alternatively suppose there is a $v$ such that 
\[
\frac{\max_i(v^t\Sigma_i v)}{\min_i(v^T\Sigma_i v)} \geq k^{4Ck}
\]
In this case, we can partition the components based on the value of $v^T\Sigma_i v$.  Without loss of generality we have
\[
v^T\Sigma_1v \geq \dots \geq v^T\Sigma_k v
\]
Note that since we are not in the first case $v^T\Sigma_1 v \geq \frac{v^T \Sigma v}{20k^{2+C}}$.
Next, because
$\frac{v^T\Sigma_k v}{v^T \Sigma v} \leq \frac{1}{k^{2Ck}}$
there must be some consecutive $i, i+1$ such that 
\[
\left( \frac{v^T\Sigma_i v}{v^T \Sigma v}\right) \geq k^C  \left( \frac{v^T\Sigma_{i+1} v}{v^T \Sigma v}\right)
 \mbox{ and } \left( \frac{v^T\Sigma_i v}{v^T \Sigma v}\right)  \geq \frac{1}{k^{2Ck}}
\]
partitioning into $S = \{1,2, \dots , i \}$ and $T = \{i+1, \dots , k \}$, we immediately verify that the desired conditions (second condition) are satisfied.
\\

\paragraph{Case 3:}

Finally, it remains to consider the situation where neither the condition in Case $1$ nor the condition in Case $2$ holds.  Note that by assumption, there is some pair $a,b \in [k]$ for which $N(\mu_a \Sigma_a), N(\mu_b, \Sigma_b)$ are not $C_k$-close.  Since we can choose
\[
C_k > (kC)^{10kC}
\]
this pair cannot fail the variance condition in any direction (second condition of Definition \ref{def:closeness}).  This pair also cannot fail the mean condition in any direction (first condition of Definition \ref{def:closeness}) because then we would have
\[
((\mu_a - \mu_b) \cdot v)^2 \geq C_k v^T\Sigma_a v \geq \frac{C_k}{k^{4Ck}} \max_i(v^T\Sigma_i v)
\]
and we would be in the first case.  Thus, we must actually have
\[
\norm{I - \Sigma_a^{-1/2}\Sigma_b\Sigma_a^{-1/2}}_2^2  \geq C_k
\]
Next, we claim that for all $i,j$, $\Sigma_i^{1/2}\Sigma_j^{-1/2}$ has smallest and largest singular value in the interval
\[
\mathcal{I} \triangleq \left[ \frac{1}{k^{4Ck}} ,  k^{4Ck} \right]
\]
If this were not true, without loss of generality we can find a unit vector $v$ such that $\norm{\Sigma_i^{1/2}\Sigma_j^{-1/2}v}_2 \geq k^{4Ck}$.
But this implies 
\[
\frac{(\Sigma_j^{-1/2}v)^T \Sigma_i (\Sigma_j^{-1/2}v)}{(\Sigma_j^{-1/2}v)^T \Sigma_j (\Sigma_j^{-1/2}v)} \geq  k^{8Ck}
\]
meaning we are actually in case $2$.  Similarly, we can show that $\Sigma_i^{1/2}\Sigma^{-1/2}$ has smallest and largest singular value in the interval $\mathcal{I}$ or else we would be in Case $1$.  
\\\\
To complete the proof, let $a_0,b_0$ be indices corresponding to a pair of components that are not $C_k$-close and construct the following graph.  Two nodes $i,j$ are connected if and only if 
\[
\norm{\Sigma_{a_0}^{-1/2}\Sigma_j\Sigma_{a_0}^{-1/2} - \Sigma_{a_0}^{-1/2}\Sigma_i\Sigma_{a_0}^{-1/2}}_2^2  \leq \frac{C_k}{k^2}
\]
This graph must not be connected since otherwise there would be a path of length at most $k$ between $a_0$ and $b_0$ and summing the above inequalities along this path, this would contradict the fact that 
\[
\norm{I - \Sigma_{a_0}^{-1/2}\Sigma_{b_0}\Sigma_{a_0}^{-1/2}}_2^2  \geq C_k.
\]
We claim that it suffices to take $S$ and $T$ to be two connected components of the graph.  Indeed, for any $a \in S, b \in T$, we have
\[
\norm{\Sigma_{a_0}^{-1/2}\Sigma_a\Sigma_{a_0}^{-1/2} - \Sigma_{a_0}^{-1/2}\Sigma_b\Sigma_{a_0}^{-1/2}}_2^2 \geq \frac{C_k}{k^2}
\]
Now observe
\[
I - \Sigma_a^{-1/2}\Sigma_b\Sigma_a^{-1/2} = (\Sigma_a^{-1/2}\Sigma_{a_0}^{1/2}) \left(\Sigma_{a_0}^{-1/2}\Sigma_a\Sigma_{a_0}^{-1/2} - \Sigma_{a_0}^{-1/2}\Sigma_b\Sigma_{a_0}^{-1/2}\right) (\Sigma_{a_0}^{1/2}\Sigma_{a}^{-1/2})
\]
and combining with the singular value bounds we showed for $\Sigma_i^{1/2}\Sigma_j^{-1/2}$ and $\Sigma_i^{1/2}\Sigma^{-1/2}$, we have
\[
\norm{I - \Sigma_a^{-1/2}\Sigma_b\Sigma_a^{-1/2}}_2^2 \geq \max\left(k^C, k^C \norm{A_{ab}}\right)
\]
for any $a,b$ on different sides of the partition.  The other quantities in the third condition can be bounded similarly as long as $C_k$ is chosen to be sufficiently large.
\end{proof}

\subsection{SOS Program}
To solve the clustering problem, we set up the same polynomial constraints as in Diakonikolas et al. \cite{diakonikolas2020robustly}. Recall that Definition~\ref{def:pseudoexpectation} gives a recipe for turning this into an SDP relaxation. 

\begin{definition}[Clustering Program $\mcl{A}$, restated from \cite{diakonikolas2020robustly}]\label{def:sos-program}
Let $X_1, \dots , X_n \in \R^d$ represent the samples.  Let $w_1, \dots , w_n, z_1, \dots , z_n, X_1', \dots , X_n'$ and $\Sigma, \Sigma^{1/2}, \Sigma^{-1/2} \in \R^{d \times d}$ (we think of the $\Sigma$ as $d \times d$ matrices whose entries are variables) be indeterminates that we will solve for in the system.  We think of the $w$ variables as weights on the points and the $z$ variables as representing whether points are outliers.  We will enforce that the subset of points weighted by $w$ has moments that are approximately Gaussian.  The full system of polynomial constraints is given below:
\begin{enumerate}
    \item We have parameters $t \in \N$ that is even and $\delta, \eps > 0$.
    \item Let $\mcl{A}_{\textsf{corruptions}} = \{z_i^2 = z_i \}_{i \in [n]}, \{z_i(X_i - X'_i) = 0 \}_{i \in [n]}, \{\sum_{i \in [n]}z_i = (1-\eps)n/k \}$
    \item Let $\mcl{A}_{\textsf{subset}} = \{w_i^2 = w_i \}_{i \in [n]},  \{\sum_{i \in [n]}w_i  = n/k \}$
    \item Let $\mu(w) = \frac{k}{n} \sum_{i \in [n]}w_iX_i'$
    \item Let $\Sigma(w) = \frac{k}{n} \sum_{i \in [n]}w_i(X_i' - \mu(w))(X_i' - \mu(w))^T$ 
    \item Let $\mcl{A}_{\textsf{matrices}} = \{ (\Sigma^{1/2})^2 = \Sigma(w) \} ,\{ (\Sigma^{-1/2}\Sigma^{1/2})^2 = \Sigma^{-1/2}\Sigma^{1/2} \} , \{ \Sigma^{-1/2}\Sigma^{1/2}w_i(X_i' - \mu(w)) = w_i(X_i'- \mu(w)) \}_{i \in [n]}$ 
    \item Let $\mcl{A}_{\textsf{moments}} $ be the following set of polynomial inequalities for all $s \leq t$
    \[
    \norm{\frac{k}{n}\sum_{i \in [n]}w_i[\Sigma^{-1/2}(X_i' - \mu(w))]^{\otimes s} - M_s}^2 \leq \delta d^{-2t}
    \]
    where $M_s = \E_{g \in N(0,I)}[g^{\otimes s}]$ is the moment tensor of a standard Gaussian.
\end{enumerate}
\end{definition}


We will work with the same set of deterministic conditions on the samples as in Diakonikolas et al. \cite{diakonikolas2020robustly}. These conditions hold with high probability for the uncorrupted points. 
\begin{definition}[Deterministic conditions, restated from \cite{diakonikolas2020robustly}]\label{def:deterministic-cond}
Fix Gaussians $G_1 , \dots , G_k$ on $\R^d$.  For $\delta, \psi > 0$ and $t \in \N$.  The $(\delta, \psi,t)$-deterministic conditions with respect to $G_1, \dots , G_k$ on a set of samples $X_1, \dots , X_n \in \R^d$ are 
\begin{enumerate}
    \item There is a partition of $\{X_1, \dots , X_n \}$ into $k$ pieces $S_1, \dots , S_k$ each of size $n/k$ such that for all $i \in [k]$ and $s \leq t$
    \[
    \norm{\frac{k}{n} \sum_{j \in S_i}[\overline{\Sigma}_i^{-1/2}(X_j - \overline{\mu}_i)]^{\otimes s} - M_s}_F^2 \leq d^{-2t}\delta
    \]
    where $\overline{\Sigma}_i$ and $\overline{\mu}_i$ denote the empirical mean and covariance of the uniform distribution over elements of $S_i$ and $M_s = \E_{g \in N(0,I)}[g^{\otimes s}]$ is the moment tensor of a standard Gaussian.
    \item For $a \in [k], v \in \R^d, A \in \R^{d \times d}$ we define
    \begin{enumerate}
    \item $E_a(v) = \{X_i \in S_a| ((X_i - \mu_a) \cdot v)^2 \leq O(1)\log(1/\psi)v^T\Sigma_a v \}$
    \item $F_a(v) = \{ (X_i, X_j) \in S_a^2 | ((X_i - X_j) \cdot v)^2 \geq \Omega(1) \cdot \psi v^T\Sigma_a v \}$
    \item $G_a(A) = \{ (X_i, X_j) \in S_a^2 | (X_i - X_j)^TA(X_i - X_j) =  2 \langle\Sigma_a , A \rangle \pm O(1) \log(1/\psi) \cdot \norm{\Sigma_a A}_F \}$.
    
    \end{enumerate}
    
    Then for every $v \in \R^d, A \in \R^{d \times d}$ we have 
    \begin{itemize}
        \item $|E_a(v)| \geq (1-\psi)(n/k)$
        \item $|F_a(v)|, |G_a(A)| \geq (1- \psi)(n/k)^2$
    \end{itemize}
    
\end{enumerate}
\end{definition}

\begin{claim}[Restated from \cite{diakonikolas2020robustly}]\label{claim:concentration}
For all even $t$, if 
\[
n \geq \log(1/\gamma)^{Ct}d^{10kt}/\delta^2
\]
for some sufficiently large constant $C$ and $\psi \geq \delta$, then $X_1, \dots , X_n$ drawn i.i.d from $\frac{1}{k}\sum_{i=1}^k G_i$ satisfy Definition \ref{def:deterministic-cond} with probability at least $1 - \gamma$.
\end{claim}

We will use the following key lemmas from \cite{diakonikolas2020robustly}.  The setup is exactly the same.  Let $X_1, \dots , X_n \in \R^d$ satisfy the $(\delta, \psi, t)$-deterministic conditions (Definition \ref{def:deterministic-cond}) with respect to Gaussians $G_1, \dots , G_k$.  Let $S_1, \dots , S_k$ be the partition guaranteed in the definition.  Let $Y_1, \dots , Y_n$ be an $\eps$-corruption of $X_1, \dots , X_n$ and let $\mcl{A}$ be the clustering program (Definition \ref{def:sos-program}) for $Y_1, \dots , Y_n$.  For indeterminates $w_1, \dots , w_n$, define
\[
\alpha_i(w) = \sum_{j \in S_i}w_j.
\]
Below we will assume $\psi, \tau$ are smaller than some universal constants $\psi_0, \tau_0 > 0$. 
\\\\
Recall in Claim \ref{claim:exists-partition} that there are essentially three different ways that two Gaussians can be separated in TV distance.  We call these mean separation, variance separation, and covariance separation.  The lemmas below roughly assert that if two Gaussians are separated in one of these ways, then a valid solution to the clustering program $\mcl{A}$ cannot assign significant weight to both of them.

\begin{lemma}[Mean Separation, restated from \cite{diakonikolas2020robustly}]\label{lemma:mean-separation}
For every $\tau > 0$, there is $s = \widetilde{O}(1/\tau^2)$ such that if $\eps, \delta \leq s^{-O(s)}k^{-20}$ then for all $a,b \in [k]$, all $v \in \R^d$ and all sufficiently small $\rho > 0$, if
\[
\la \mu_a - \mu_b , v \ra^2 \geq \rho \E_{X,X' \sim \frac{1}{k}\sum G_i}\la X - X', v \ra^2,
\]
then
\[
\mcl{A} \vdash_{O(s)} \left(\frac{\alpha_a(w)\alpha_b(w)}{n^2} \right)^s \leq (s\log 1/\psi )^{O(s)} \cdot \left( \frac{\la v, \Sigma_av \ra + \la v, \Sigma_b v \ra}{\la \mu_a - \mu_b, v \ra^2} \right)^{\Omega(s)} + \rho^{-O(s)}(\tau^{\Omega(s)} + \eps^{\Omega(s)}k^{O(s)}s^{O(s^2)} + \psi^{\Omega(s)} )
\]
\end{lemma}

\begin{lemma}[Variance Separation, restated from \cite{diakonikolas2020robustly}]\label{lemma:variance-separation}
For every $\tau > 0$, there is $s = \widetilde{O}(1/\tau^2)$ such that if $\eps, \delta \leq s^{-O(s)}k^{-20}$ then for all $a,b \in [k]$, all $v \in \R^d$ and all sufficiently small $\rho > 0$, if
\[
\la v, \Sigma_b v\ra  \geq \rho \E_{X,X' \sim \frac{1}{k}\sum G_i}\la X - X', v \ra^2,
\]
then
\[
\mcl{A} \vdash_{O(s)} \left(\frac{\alpha_a(w)\alpha_b(w)}{n^2} \right)^s \leq \psi^{-O(s)} \cdot \left(  s^{O(s)}\left( \frac{\la v, \Sigma_a v \ra}{\la v, \Sigma_b v \ra}\right)^{\Omega(s)} + \rho^{-O(s)}(\tau^{\Omega(s)} + \eps^{\Omega(s)}k^{O(s)}s^{O(s^2)}   )\right) + \rho^{-O(s)}\psi^{\Omega(s)}
\]
\end{lemma}

\begin{lemma}[Covariance Separation, restated from \cite{diakonikolas2020robustly}]\label{lemma:covariance-separation}
Let $\Sigma$ be the covariance of the mixture $\frac{1}{k}\sum G_i$.  If $\eps, \delta < k^{-O(1)}$, then for all $a,b \in [k]$ and $A \in \R^{d \times d}$,
\begin{align*}
\mcl{A} \vdash_{O(1)} \left(\frac{\alpha_a(w)\alpha_b(w)}{n^2} \right)^{16} \leq O(\log 1/\psi)^8 \cdot \frac{\norm{\Sigma^{1/2}A\Sigma^{1/2}}_F^8 + \norm{\Sigma_a^{1/2}A\Sigma_a^{1/2}}_F^8 + \norm{\Sigma_b^{1/2}A\Sigma_b^{1/2}}_F^8 }{\la \Sigma_a - \Sigma_b A \ra^8} \\ + O(\psi^4) + O(\eps^2k^{20})
\end{align*}

\end{lemma}

\subsection{Clustering Algorithm}
We use essentially the same clustering algorithm as \cite{diakonikolas2020robustly}.

\begin{algorithm}[h]
\caption{{\sc Rough Clustering} }
\begin{algorithmic} 
\State \textbf{Input}: $\eps$-corrupted samples $X_1, \dots , X_n$ and parameters $t, \delta, \eps, k, \eta$
\State Initialize a list of subsets $L = \{ \}$
\For {$\textsf{count} = 0,1, \dots , 100k \log 1/\eta$}
\State Let $\mcl{A}$ be the clustering program (Definition \ref{def:sos-program}) for $X_1, \dots , X_n$
\State Compute the pseudoexpectation $\widetilde{E}$ that satisfies the constraints $\mcl{A}$ (Definition~\ref{def:sos-program}) and maximizes 
\[
\widetilde{E}\left[\sum_{i \not\in \cup_{R \in L}R} w_i\right]
\]
\State Choose a random $i \sim [n]$ with probability $p_i = \frac{\widetilde{\E}[w_i]}{\sum \widetilde{\E}[w_i]}$
\State Create set $R$ by adding each element $j \in [n]$ independently with probability $\frac{\widetilde{\E}[w_iw_j]}{\widetilde{\E}[w_i]} $
\State Add $R$ to the list $L$
\EndFor
\State Let $L = \{R_1, \dots , R_m \}$
\For {all subsets $S \subset L$}
\State Recurse on $\cup_{i \in S}R_i$ for each of $k \rightarrow 1,2, \dots , k-1$ and $t, \delta, \eps, \eta$ unchanged
\EndFor
\State Return $\{X_1, \dots , X_n \}$ (as one cluster) and all unions of some combination of the clusters returned in each computation branch
\end{algorithmic}
\label{alg:rough-cluster}
\end{algorithm}

\begin{proof}[Proof of Theorem \ref{thm:rough-clustering}]
We can use Claim \ref{claim:concentration} to ensure that with $1 - \gamma/2$ probability, the deterministic conditions in Definition \ref{def:deterministic-cond} are satisfied for all submixtures and the various values of $\delta, \psi, t$ that we will need.

First, if all pairs of components are $D'$ close, then returning the entire sample as one cluster suffices.  Now, we may assume that there is some pair that is not $D'$-close.  We apply Claim \ref{claim:exists-partition} and let $U,V$ be the partition of the components given by the claim.  Let $C$ be a sufficiently large function of $k,D, \gamma$ that we will set later.  We can do this as long as we ensure that $D'$ is a sufficiently large function of $k,C$.  We ensure that $k^C > D$.    Note that each of the pieces $A_1, \dots , A_l$, must be entirely in $U$ or in $V$ because of our assumption about closeness between the components.  We claim that 
\begin{equation}\label{eq:partition}
\widetilde{\E}\left[\left(\sum_{i \in \cup_{j \in U} S_j} w_i \right)\left(\sum_{i \in \cup_{j \in V} S_j} w_i \right) \right]\leq \gamma' n^2
\end{equation}
where we can make $\gamma'$ sufficiently small in terms of $\gamma,D, k$ by choosing $D'$ and the functions $f, F$ suitably.
\\\\
Below we will let $a,b$ be indices such that $a \in U$ and $b \in V$.  If the first clause of  Claim \ref{claim:exists-partition} is satisfied, then we can take $\rho = poly(1/k)$ and for $\tau$ sufficiently small in terms of $\gamma, k, D, C$, we have
\[
\widetilde{\E}\left( \frac{\alpha_{a}(w)\alpha_b(w)}{n^2}\right)^s \leq k^{-C/2s}
\]
Summing over all $a \in U, b \in V$,  this gives (\ref{eq:partition}).
\\\\
If the second clause of  Claim \ref{claim:exists-partition} is satisfied, then we can take 
\[
\rho = \min_{a \in U}\frac{v^T\Sigma_a v}{v^T \Sigma v} \geq \frac{1}{k^{2Ck}}
\]
We choose $\tau $ sufficiently small in terms of $\gamma, k, C, D$ and combining with the fact that $(v^t\Sigma_a v) \geq k^{C}(v^T\Sigma_bv)$ 
for all $a \in U, b \in V$ we get
\[
\widetilde{\E}\left( \frac{\alpha_{a}(w)\alpha_b(w)}{n^2}\right)^s \leq k^{-\Omega(C)s}
\]
Finally, when the third clause of Claim \ref{claim:exists-partition} is satisfied follows similarly after setting $A = A_{ab}$.  In all cases, we now have 
(\ref{eq:partition}).  The next step will be to analyze our random sampling to select the subset $R$.  First note
\[
\widetilde{\E}[|R|] = \frac{\sum_{i,j}\widetilde{\E}[w_iw_j]}{\sum_{i}\widetilde{\E}[\sum_{i} w_i]} = \frac{n}{k}
\]
Next we analyze the intersections with the two sides of the partition $U,V$.  We will slightly abuse notation and use $i \in U$ when $i \in \cup_{j \in U}S_j$ and it is clear from context that we are indexing the samples.  Conditioned on the first index that is randomly chosen satisfying $i \in U$ then
\[
\E[|R \cap V|] = \frac{\sum_{i_1 \in U, i_2 \in V}\widetilde{\E}[w_{i_1}w_{i_2}]}{\sum_{i \in U}\widetilde{\E}[w_i]} \leq \frac{\gamma'n^2}{\sum_{i \in U}\widetilde{\E}[w_i]}
\]
repeating the same argument for when $i \in V$, we have $\E[\min(|R \cap U| , |R \cap V|)] \leq \gamma'kn$. 
Finally, we lower bound the expected number of new elements that $R$ adds to the list $L$.  This quantity is 
\[
\frac{\widetilde{\E}[\sum_{i \in [n], j \notin L}w_iw_j]}{n/k} = \sum_{j \notin L}\widetilde{\E}[w_j]
\]
where by $j \notin L$ we mean $j$ is not in the union of all previous subsets in the list $L$.  Note that indicator functions of the components $S_1, \dots , S_k$ are all valid pseudoexpectations and since we are picking the pseudoexpectation that maximizes $\sum_{j \notin L}\widetilde{\E}[w_j]$, the expected number of new elements added to $L$ is at least
\[
\frac{n - |\cup_{R \in L}R|}{k}
\]
Now we analyze the recombination step once we finalize $L = \{R_1, \dots , R_m \}$.  For any sufficiently small function $h(k, \gamma, D)$, we claim that by choosing $D'$ and the functions $f, F$ appropriately, we can ensure with $1 - h(k, D, \gamma)$ probability, there is some recombination that gives a $1 - h(k, D, \gamma)$-corrupted sample of the submixture corresponding to $U$.  To see this, it suffices to set $\eta < h(k, D, \gamma)$ and then look at the first $m' = 100 k \log 1/h(k,D, \gamma)$ subsets in $L$.  Their union has expected size $(1 - h(k,D,  \gamma)^{100})n$. 
Next, among $R_1, \dots , R_{m'} $,
\[
\E\left[\sum_{i=1}^{m'} \min(|R_i \cap U| , |R_i \cap V|) \right] \leq \gamma' km'n
\]
If we ensure that $\gamma'$ is sufficiently small in terms of $\gamma,D, k$, then using Markov's inequality, with $1 - h(k,D, y)$ probability, there is some recombination that gives a $1 - h(k, D, \gamma)$-corrupted sample of the submixture corresponding to $U$.  We can make the same argument for $V$.  Now we can recurse and repeat the argument because each of these submixtures only contains at most $k-1$ true components.
\end{proof}

\subsection{Improved Clustering Result from \cite{bakshi2020robustlya}} \label{sec:improved-clustering}

In \cite{bakshi2020robustlya}, the authors obtain a rough clustering result similar to Theorem \ref{thm:rough-clustering} but are able to remove the bounded fractionality assumption.  Their result is restated using our notation below.
\begin{theorem}[\cite{bakshi2020robustlya}]\label{thm:improved-clustering}
Let $k, D, \gamma $ be parameters.  Assume we are given $\eps$-corrupted samples from a mixture of Gaussians $w_1G_1 + \dots  + w_kG_k$ where the mixing weights $w_i$ are all at least $1/A$ for some constant $A$.  Let $A_1, \dots , A_l$ be a partition of the components such that 
\begin{enumerate}
    \item For any $j_1, j_2$ in the same piece of the partition $G_{j_1}, G_{j_2}$ are $D$-close
    \item For any $j_1, j_2$ in different pieces of the partition, $G_{j_1}, G_{j_2}$ are not $D'$-close
\end{enumerate}
where $D'>  F(k,A, D, \gamma) $ for some sufficiently large function $F$.  Assume that $t >F(k,A, D, \gamma) $ and $\eta, \eps, \delta < f(k,A, D, \gamma)$ for some sufficiently small function $f$.  Then with probability at least $1 - \gamma$, if $X_1, \dots , X_n$ is an $\eps$-corrupted sample from the mixture $w_1G_1 + \dots + w_kG_k$ with $n \geq \poly(1/\eps, 1/\eta,  1/\delta, d)^{O(k,A)}$, then there is an algorithm that runs in $\poly(n)$ time and returns $O_k(1)$ candidate clusterings, at least one of which gives a $\gamma$-corrupted sample of each of the submixtures given by $A_1, \dots , A_l$.
\end{theorem}

Observe that Theorem \ref{thm:improved-clustering} is the same as Theorem \ref{thm:rough-clustering} but with the bounded fractionality assumption removed.  Theorem \ref{thm:rough-clustering} is the only source of the bounded fractionality assumption in our paper.  In the subsequent sections, replacing all uses of Theorem \ref{thm:rough-clustering} with Theorem \ref{thm:improved-clustering} allows us to remove the bounded fractionality assumption from our main result.

\section{Putting Everything Together}\label{sec:full-alg}

We can now combine our clustering results and our results for learning mixtures of Gaussians that are not too separated to get a learning algorithm in the fully general case.  Our main theorem is stated below.

\begin{theorem}\label{thm:full-alg}
Let $k, A, b > 0$ be constants.  There is a sufficiently large function $G$ and a sufficiently small function $g$ depending only on $k,A,b$ (with $G(k,A,b),g(k,A,b) > 0$) such that given an $\eps$-corrupted sample $X_1, \dots , X_n$ from a mixture of Gaussians $\mcl{M} = w_1G_1 + \dots + w_kG_k \in \R^d$ where the $G_i$ have variance at least $\poly(\eps/d)$ and at most $\poly(d/\eps)$ in all directions and
\begin{itemize}
    \item The $w_i$ are all rational with denominator at most $A$
    \item $d_{\TV}(G_i, G_j) \geq b$
\end{itemize}
and $n \geq (d/\eps)^{G(k,A,b)}$, then there is an algorithm that runs in time $\poly(n)$ and with $0.99$ probability outputs a mixture 
\[
\widetilde{\mcl{M}} = \widetilde{w_1}\widetilde{G_1} + \dots + \widetilde{w_k}\widetilde{G_k}
\]
such that $d_{\TV}(\widetilde{\mcl{M}}, \mcl{M}) \leq \eps^{g(k,A,b)}$. 
\end{theorem}

\subsection{Distance Between Gaussians}
We will need to prove a few preliminary results.  The main lemma we prove in this section is the following, which gives a stronger bound than the triangle inequality for TV distance between Gaussians.

\begin{lemma}\label{lem:distance-between-gaussians}
Let $\lambda$ be a constant.  Let $A,B,C$ be Gaussian distributions.  Assume that $d_{\TV}(A,B) \leq 1 - \lambda$.  If $d_{\TV}(A,C) \geq 1 - \eps$ and $\eps$ is sufficiently small then 
\[
d_{\TV}(B,C) \geq 1 - \poly( \eps)
\]
(where the RHS may depend on $\lambda$).
\end{lemma}
Note that this result is not true for arbitrary distributions $A,B,C$.  We actually need to exploit the fact that $A,B,C$ are Gaussian.

Our proof will parallel results in Section 8 of \cite{diakonikolas2020robustly}.  First, a definition:

\begin{definition}
For two distributions $P,Q$ let
\[
h(P,Q) = -\log(1 - d_{\TV}(P,Q)).
\]
\end{definition}

The key ingredient is the following result from \cite{diakonikolas2020robustly}:
\begin{lemma}[Restated from \cite{diakonikolas2020robustly}]\label{lem:prob-ratios}
Let $A$ and $B$ be two Gaussians with $h(A,B) = O(1)$.  If $D \in \{A, B \}$ then
\[
P_{x \sim D}\left[ \eps \leq \frac{A(x)}{B(x)} \leq \frac{1}{\eps} \right] \geq 1 - \poly(\eps)
\]
\end{lemma}

\begin{proof}[Proof of Lemma \ref{lem:distance-between-gaussians}]
Note that 
\[
P_{x \sim A} \left[ \eps^{0.5} \leq \frac{A(x)}{C(x)} \leq \frac{1}{\eps^{0.5}} \right] \leq \eps^{0.5}
\]
If this weren't the case, then $A$ and $C$ would have more than $\eps$ overlap, contradicting our assumption.  Next, by Lemma \ref{lem:prob-ratios},
\begin{equation}\label{eq:prob-bound}
P_{x \sim A}\left[ \eps^{0.1} \leq \frac{A(x)}{B(x)} \leq \frac{1}{\eps^{0.1}} \right] \geq 1 - \poly(\eps)
\end{equation}
Combining the above two inequalities, we deduce 
\begin{equation}\label{eq:prob-bound2}
P_{x \sim A}\left[\eps^{0.4} \leq \frac{C(x)}{B(x)} \leq \frac{1}{\eps^{0.4}} \right] \leq \poly(\eps)
\end{equation}
Let $0 < c < 0.1$ be a constant such that the RHS of (\ref{eq:prob-bound2}) is at most $\eps^c$.  By Lemma \ref{lem:prob-ratios}
\[
P_{x \sim B}\left[ \eps^{c/2} \leq \frac{A(x)}{B(x)} \leq \frac{1}{\eps^{c/2}} \right] \geq 1 - \poly(\eps)
\]
and combining with (\ref{eq:prob-bound2}), we deduce
\[
P_{x \sim B}\left[ \eps^{0.4} \leq \frac{C(x)}{B(x)} \leq \frac{1}{\eps^{0.4}} \right] \leq  \poly(\eps)
\]
which implies $d_{\TV}(B,C) \geq 1 - \poly( \eps)$.
\end{proof}

\subsection{Full Algorithm}
We are now ready to prove Theorem \ref{thm:full-alg}.  We begin by describing the algorithm.  Our full algorithm consists of several phases.
\begin{enumerate}
    \item Cluster with constant accuracy into constant-separated submixtures
    \item Learn parameters of submixtures to constant accuracy
    \item Recluster all points and form new $\poly(\eps)$-separated submixtures
    \item Learn parameters of submixtures to $\poly(\eps)$ accuracy
\end{enumerate}

The algorithm {\sc Learn Parameters (well-conditioned)} for learning the parameters of a well-conditioned mixture of Gaussians (see Theorem \ref{thm:full-close-case}) is summarized in Algorithm \ref{alg:learn-params}.
\begin{algorithm}[h]
\caption{{\sc Learn Parameters (well-conditioned)} }
\begin{algorithmic} 
\State \textbf{Input:} $\eps$-corrupted sample $X_1, \dots , X_n$ from $\delta$-well-conditioned mixture of Gaussians $\mcl{M} = w_1G_1 + \dots + w_kG_k$
\State Estimate Hermite polynomials of $\mcl{M}$
\State Solve for parameters using SOS (see Section \ref{sec:close-case})
\end{algorithmic}
\label{alg:learn-params}
\end{algorithm}

Our full algorithm is described in the next block Algorithm \ref{alg:full}.

\begin{algorithm}[h]
\caption{{\sc Full Algorithm} }
\begin{algorithmic} 
\State \textbf{Input:} $\eps$-corrupted sample $X_1, \dots , X_n$ from mixture of Gaussians $\mcl{M} = w_1G_1 + \dots + w_kG_k$
\State Run {\sc Rough Clustering} Algorithm to split sample into subsamples for submixtures where all pairs are $D$-close for constant $D$
\For {each candidate clustering}
\State Run {\sc Learn Parameters (well-conditioned)} for each submixture 
\State Output candidate components
\EndFor
\For {each set of candidate components $\widetilde{G_1}, \dots \widetilde{G_k}$}
\State Assign samples to components according to maximum likelihood to form sets of samples $\{\widetilde{S_1}, \dots , \widetilde{S_k} \}$
\For {all partitions of $[k]$ into sets $R_1, \dots , R_l$  }
\State Run {\sc Learn Parameters (well-conditioned)} on each of $\cup_{i \in R_j}S_i$ for all $j \in [l]$
\State Output candidate components
\EndFor
\EndFor
\State Hypothesis test over all candidate components to find a mixture $\widetilde{\mcl{M}}$ that is $\poly(\eps)$-close to $\mcl{M}$
\end{algorithmic}
\label{alg:full}
\end{algorithm}


\subsection{Analysis of {\sc Full Algorithm}}

The first step will be to show that among the first set of candidate components that we output, there are some that are within constant distance (say $< c(k)$ for some sufficiently small function $c$) of the true components.  
\begin{lemma}\label{lem:const-accuracy}
Let $k, A, b > 0$ be constants and $\theta$ be a desired accuracy.  There is a sufficiently large function $G$ and a sufficiently small function $g$ depending only on $k, A , b, \theta$ such that given an $\eps$-corrupted sample $X_1, \dots , X_n$ from a mixture of Gaussians $\mcl{M} = w_1G_1 + \dots + w_kG_k \in \R^d$ where
\begin{itemize}
    \item The $w_i$ are all rational with denominator at most $A$
    \item $d_{\TV}(G_i, G_j) \geq b$
\end{itemize}
and 
\begin{itemize}
    \item $ \eps < g(k,A,b, \theta)$
    \item $n \geq (d/\eps)^{G(k,A,b, \theta)}$
\end{itemize}
then there is an algorithm that runs in time $\poly(n)$ and with $0.999$ probability outputs a set of $(1/\theta)^{G(k,A,b, \theta)}$ candidate mixtures at least one of which satisfies
\begin{align*}
\max \left( d_{\TV} (\widetilde{G_1}, G_1) , \dots , d_{\TV} (\widetilde{G_k}, G_k) \right) \leq \theta \\
\max \left( |\widetilde{w_1} - w_1|, \dots , |\widetilde{w_k} - w_k| \right) \leq \theta
\end{align*}

\end{lemma}
\begin{proof}
We will use Theorem \ref{thm:rough-clustering} to argue that the clustering algorithm finds some set of candidate clusters that can then be used to learn the parameters via Theorem \ref{thm:full-close-case}.  The main thing we need to prove is that we can find the $D,D'$ satisfying the hypotheses of Theorem \ref{thm:rough-clustering}.  In the argument below, all functions may depend on $k,A,b, \theta$ but we may omit writing some of these variables in order to highlight the important dependences.

Note that Claim \ref{claim:trivial-dist-bound} combined with Theorem \ref{thm:full-close-case} imply that if we have a $\gamma$-corrupted sample of a submixture of $\mcl{M}$ where all pairs are $D$-close and $\gamma < f(D, \theta)$ for some sufficiently small function $f$ then we can learn the components of the submixture to the desired accuracy.  Now if the separation condition of Theorem \ref{thm:rough-clustering} were satisfied with $ \gamma = f(D, \theta)$ and $D' > F(k,D , \gamma)$ then we would be done. 
\\\\\
We now show that there is some constant $D$ depending only on $k,A,b, \theta$ for which this is true.  Assume that the condition does not hold for some value of $D_0$.  Then construct a graph $G_{D_0}$ on nodes $1,2, \dots , k$ where two nodes are connected if and only if they are $D$-close.  Take the connected components in this graph.  Note that by Claim \ref{claim:dist-transitivity}, all pairs in the same connected component are $\poly(D_0)$-close.  Thus, there must be an edge between two components such that $G_i$ and $G_j$ are $D_1$-close for 
\[
D_0 < D_1 < F(k, \poly(D_0), f(\poly(D_0), \theta))
\]
Now the graph $G_{D_1}$ has one less connected component than $G_{D_0}$.  Starting from say $D_0 = 2$, we can iterate this argument and deduce that the entire graph will be connected for some constant value of $D$ depending only on $k,A,b, \theta$.  Now by Claim \ref{claim:dist-transitivity} it suffices to treat the entire mixture as one mixture and we can apply Claim \ref{claim:trivial-dist-bound} and Theorem \ref{thm:full-close-case} to complete the proof.
\end{proof}

Our next step is to show that if our algorithm starts with component estimates that are accurate within some constant and guesses a good set of clusters, then the resulting subsamples (after assigning according to maximum likelihood) are equivalent to $\poly(\eps)$-corrupted samples from the corresponding submixtures.  First, we prove a preliminary claim which implies that a good set of clusters exists. 

\begin{claim}\label{claim:exists-good-partition}
Let $\mcl{M} = w_1G_1 + \dots + w_kG_k $ be a mixture of Gaussians.  For any constant $c > 0$ and parameter $\eps$, there exists a function  $f(c,k)$ such that there exists a partition (possibly trivial) of $[k]$ into sets $R_1, \dots , R_l$ such that 
\begin{itemize}
    \item If we draw edges between all $i,j$ such that $d_{\TV}(G_i, G_j) \leq 1 - \eps^{c \kappa}$
    then each piece of the partition is connected
    \item For any $i,j$ in different pieces of the partition $  d_{\TV}(G_i, G_j) \geq 1 - \eps^{\kappa}$
\end{itemize}
and $f(c,k) < \kappa < 1$.
\end{claim}
\begin{proof}
For a real number $f$, let $\mcl{G}_f$ be the graph on $[k]$ obtained by connecting two nodes $i,j$ if and only if $d_{\TV}(G_i, G_j) \leq 1 - f$.  Consider $\mcl{G}_{\eps^{c^{k}}}$.  Consider the partition formed by taking all connected components in this graph.  If this partition does not satisfy the desired condition, then there are some two $G_i,G_j$ in different components such that 
\[
d_{\TV}(G_i, G_j) \leq 1 - \eps^{c^{k-1}}
\]
Thus, the graph $\mcl{G}_{\eps^{c^{k-1}}}$ has strictly fewer connected components than $\mcl{G}_{\eps^{c^{k}}}$.  We can now repeat this argument on $\mcl{G}_{\eps^{c^{k-1}}}$.  However, the number of connected components in  $\mcl{G}_{\eps^{c^{k}}}$ is at most $k$ so we conclude that there must be some $c^k < \kappa < 1$ for which the desired condition is satisfied.
\end{proof}

We will also need the following results about the VC-dimension of hypotheses obtained by comparing the density functions of two mixtures of Gaussians.  The reason we need these VC dimension bounds is that we will need to argue that given \textit{any} constant-accuracy estimates, we can obtain a clustering that is $\poly(\eps)$ accurate.  While naively this would require union bounding over infinitely many possibilities for the initial estimates, the VC dimension bound allows to get around this and obtain uniform convergence over all possible initial estimates.  

Technically for our clustering result, we only need the VC dimension bound for single Gaussians (instead of mixtures of $k$ Gaussians).  However, we will need the VC dimension bound for mixtures of Gaussians later when we do hypothesis testing so we state the full result below.  First we need a definition.
\begin{definition}\label{def:yatracos}
Let $\mcl{F}$ be a family of distributions on some domain $\mcl{X}$.  Let $\mcl{H}_{\mcl{F},a}$ be the set of functions of the form $f_{\mcl{M}_1, \mcl{M}_2,\dots , \mcl{M}_a}$ where $\mcl{M}_1, \mcl{M}_2, \dots , \mcl{M}_a \in \mcl{F}$ and 
\begin{align*}
f_{\mcl{M}_1, \mcl{M}_2,\dots , \mcl{M}_a}(x) = 
\begin{cases}
1 \text{ if } \mcl{M}_1(x) \geq \mcl{M}_2(x), \dots , \mcl{M}_a(x) \\
0 \text{ otherwise}
\end{cases}   
\end{align*}
where $\mcl{M}_i(x)$ denotes the pdf of the corresponding distribution at $x$.
\end{definition}

\begin{lemma}[Theorem 8.14 in \cite{anthony2009neural}]\label{lem:VCdim}
Let $\mcl{F}_{k}$ be the family of distributions that are a mixture of at most $k$ Gaussians in $\R^d$.  Then the VC dimension of $\mcl{H}_{\mcl{F}_k,a}$ is $\poly(d,a, k)$.
\end{lemma}

It is a standard result in learning theory that for a hypothesis class with bounded VC dimension, taking a polynomial number of samples suffices to get a good approximation for all hypotheses in the class.
\begin{lemma}[\cite{vapnik2015uniform}]\label{lem:uniformconvergence}
Let $\mcl{H}$ be a hypothesis class of functions from some domain $\mcl{X}$ to $\{0,1 \}$ with VC dimension $V$.  Let $\mcl{D}$ be a distribution on $\mcl{X}$.  Let $\eps, \delta >0$ be parameters.  Let $S$ be a set of $n = \poly(V, 1/\eps,  \log 1/\delta)$ i.i.d samples from $\mcl{D}$.  Then with $1 - \delta$ probability, for all $f \in \mcl{H}$
\[
\left \lvert \E_{x \sim S}[f(x)] - \E_{x \sim \mcl{D}}[f(x)]\right \rvert \leq \eps \,.
\]
\end{lemma}

Now we can prove our lemma about obtaining a $\poly(\eps)$-accurate clustering into submixtures when given constant-accuracy estimates for the components.

\begin{lemma}\label{lem:find-good-clusters}
Let $\mcl{M} = w_1G_1 + \dots + w_kG_k \in \R^d$ be a mixture of Gaussians where
\begin{itemize}
    \item The $w_i$ are all rational with denominator at most $A$
    \item $d_{\TV}(G_i, G_j) \geq b$
\end{itemize}
There exists a sufficiently small function $g(k,A,b)> 0$ depending only on $k,A,b$ such that the following holds.   Let $X_1, \dots , X_n$ be an $\eps$-corrupted sample from the mixture $\mcl{M}$ where $\eps < g(k,A,b)$ and $n = \poly(d/\eps)$ for some sufficiently large polynomial.  Let $S_1, \dots , S_k \subset \{X_1, \dots, X_n \}$ denote the sets of samples from each of the components $G_1, \dots , G_k$ respectively.  Let $R_1, \dots , R_l$ be a partition such that for $i_1 \in R_{j_1}, i_2 \in R_{j_2}$ with $j_1 \neq j_2$,
\[
d_{\TV}(G_{i_1}, G_{i_2}) \geq 1 - \eps'
\]
where $  \eps \leq \eps' \leq g(k,A,b)$. Let $\widetilde{G_1}, \dots , \widetilde{G_k}$ be any Gaussians such that $d_{\TV}(G_i, \widetilde{G_i}) \leq g(k,A,b)$ for all $i$.    Let $\widetilde{S_1}, \dots , \widetilde{S_k} \subset \{X_1, \dots , X_n \}$ be the subsets of samples obtained by assigning each sample to the component $\widetilde{G_i}$ that gives it the maximum likelihood.  Then with probability  at least $0.999$,
\[
\left| \left(\cup_{i \in R_j}S_i\right)  \cap \left(\cup_{i \in R_j} \widetilde{S_i}\right) \right| \geq (1-\poly(\eps')) \left|\left(\cup_{i \in R_j}S_i\right) \right|
\]
for all $j$.
\end{lemma}
\begin{proof}
First, we will upper bound the expected number of uncorrupted points that are mis-classified for each $j \in [l]$ when the Gaussians $\widetilde{G_1}, \dots , \widetilde{G_k}$ are fixed.  This quantity can be upper bounded by
\[
\sum_{j_1 \neq j_2} \sum_{\substack{i_1 \in R_{j_1} \\ i_2 \in R_{j_2}}} \int 1_{\widetilde{G_{i_1}}(x) > \widetilde{G_{i_2}}(x)} d G_{i_2}(x)
\]

Clearly we can ensure $d_{\TV}(G_i, \widetilde{G_i}) \leq 1/2$.  Thus, by Lemma \ref{lem:distance-between-gaussians} and the assumption about $R_1, \dots , R_l$, $d_{\TV}(\widetilde{G_{i_1}}, G_{i_2}) \geq 1 - \poly(\eps')$ for all $G_{i_2}$ where $i_2$ is not in the same piece of the partition as $i_1$.  Let $c$ be such that 
\[
d_{\TV}(\widetilde{G_{i_1}}, G_{i_2}) \geq 1 - \eps'^c
\]
By Lemma \ref{lem:prob-ratios}, 
\[
\Pr_{x \in G_{i_2}}\left[  \eps'^{c/2} \leq \frac{\widetilde{G_{i_2}}(x)}{G_{i_2}(x)} \leq \eps'^{c/2} \right] \geq 1 - \poly(\eps')
\]
and combining the above two inequalities, we deduce
\[
\int 1_{\widetilde{G_{i_1}}(x) > \widetilde{G_{i_2}}(x)} d G_{i_2}(x) \leq \poly(\eps')
\]
Since we are only summing over $O_k(1)$ pairs of components, as long as $\eps'$ is sufficiently small compared to $k, A, b$, the expected fraction of misclassified points is $\poly(\eps')$.  
\\\\
Next, note that the clustering depends only on the comparisons between the values of the pdfs of the Gaussians $\widetilde{G_1}, \dots , \widetilde{G_k}$ at each of the samples $X_1, \dots , X_n$.  Since $n = \poly(d/\eps)$ for some sufficiently large polynomial, applying Lemma \ref{lem:VCdim} and Lemma \ref{lem:uniformconvergence} completes the proof (note that the fraction of corrupted points is at most $\eps$ so it does not matter how they are clustered).
\end{proof}

Combining Lemma \ref{lem:const-accuracy}, Claim \ref{claim:exists-good-partition}, Lemma \ref{lem:find-good-clusters} and Theorem \ref{thm:full-close-case}, we can show that at least one of the sets of candidate parameters that our algorithm outputs is close to the true parameters.
\begin{lemma}\label{lem:list-learning}
Let $k, A, b > 0$ be constants.  There is a sufficiently large function $G$ and a sufficiently small function $g$ depending only on $k,A,b$ such that given an $\eps$-corrupted sample $X_1, \dots , X_n$ from a mixture of Gaussians $\mcl{M} = w_1G_1 + \dots + w_kG_k \in \R^d$ where
\begin{itemize}
    \item The $w_i$ are all rational with denominator at most $A$
    \item $d_{\TV}(G_i, G_j) \geq b$
\end{itemize}
and $n \geq (d/\eps)^{G(k,A,b)}$, with $0.999$ probability, among the set of candidates output by {\sc Full Algorithm}, there is some $\{\widetilde{w_1}, \widetilde{G_1}, \dots , \widetilde{w_k}, \widetilde{G_k} \}$ such that for all $i$ we have
\[
|w_i - \widetilde{w_i}| + d_{\TV}(G_i, \widetilde{G_i}) \leq \poly(\eps)
\]
\end{lemma}
\begin{proof}
This follows from combining Lemma \ref{lem:const-accuracy}, Claim \ref{claim:exists-good-partition}, Lemma \ref{lem:find-good-clusters} and finally applying Theorem \ref{thm:full-close-case}.  Note we can choose $c$ in Claim \ref{claim:exists-good-partition} so that when combined with Lemma \ref{lem:find-good-clusters}, the resulting accuracy that we get on each submixture is high enough that we can then apply Theorem \ref{thm:full-close-case} (we can treat the subsample corresponding to each submixture as a $\poly(\eps')$-corrupted sample).  We apply Lemma \ref{lem:find-good-clusters} with $\eps' = \eps^{\kappa}$ where the $\kappa$ is obtained from Claim \ref{claim:exists-good-partition}.
\end{proof}

We have shown that our algorithm recovers a list of candidate mixtures, at least one of which is close to the true mixture.  The last result that we need is a hypothesis testing routine.  This is similar to the hypothesis testing result in \cite{kane2020robust}.  However, there is a subtle difference that the samples we use to hypothesis test may not be independent of the hypotheses.  This is because the adversary sees all of the data points and may corrupt the data in a way to affect the list of hypotheses that we output.  Thus, we must prove that given an $\eps$-corrupted sample and \textit{any} list of hypotheses with the promise that at least one of the hypotheses is close to the true distribution, we must output a hypothesis that is close to the true distribution.  
\begin{lemma}\label{lem:hypothesis-test}
Let $\mcl{F}$ be a family of distributions on some domain $\mcl{X}$ with explicitly computable density functions that can be efficiently sampled from.  Let $V$ be the VC dimension of $\mcl{H}_{\mcl{F},2}$ (recall Definition \ref{def:yatracos}). Let $\mcl{D}$ be an unknown distribution in $\mcl{F}$.  Let $m$ be a parameter.  Let $X_1, \dots , X_n$ be an $\eps$-corrupted sample from $\mcl{D}$ with $n \geq  \poly(m,\eps, V)$ for some sufficiently large polynomial.  Let $H_1, \dots , H_m$ be distributions in $\mcl{F}$ given to us by an adversary with the promise that
\[
\min(d_{\TV}(\mcl{D}, H_i) ) \leq \eps \,.
\]
Then there exists an algorithm that runs in time $\poly(n, \eps)$ and outputs an $i$ with $1 \leq i \leq m$ such that with $0.999$ probability 
\[
d_{TV}(\mcl{D}, H_i) \leq O(\eps) \,.
\]
\end{lemma}
\begin{proof}
The proof will be very similar to the proof in \cite{kane2020robust} except we will use the VC dimension bound and Lemma \ref{lem:uniformconvergence} to obtain a bound over all possible hypothesis distributions given to us by the adversary.

For each $i,j$, define $A_{i,j}$ to be the subset of $\mcl{X}$ where $H_i(x) \geq H_j(x)$ (where we abuse notation and use $H_i, H_j$ to denote their respective probability density functions).  Note $d_{\TV}(H_i, H_j) = |H_i(A_{i,j}) - H_j(A_{i,j}) |$.  By Lemma \ref{lem:uniformconvergence}, we can ensure that with high probability, for all $i,j$, the empirical estimates of $A_{i,j}$ are close to their true values, i.e.
\[
| \mcl{D}(A_{i,j}) - X(A_{i,j})| \leq 2\eps \,.
\]

Now, since the distributions $H_1, \dots , H_m$ can be efficiently sampled from, we can obtain estimates $\wh{H_l}(A_{i,j})$ that are within $\eps$ of $H_l(A_{i,j})$ for all $i,j, l$.  Now, it suffices to return any $l$ such that for all $i,j$,
\[
|  X(A_{i,j}) - \wh{H_l}(A_{i,j}) | \leq 4\eps \,.
\]
Note that any $l$ such that $d_{\TV}(\mcl{D}, H_l) \leq \eps$ must satisfy the above by the triangle inequality.  Next, we argue that any such $l$ must be sufficient.  To see this, let $l'$ be such that $d_{\TV}(\mcl{D}, H_{l'}) \leq \eps$.  Then
\begin{align*}
 d_{\TV}(\mcl{D}, H_{l}) \leq \eps + d_{\TV}(H_l, H_{l'}) =   \eps + |H_l(A_{l,l'}) - H_{l'}(A_{l,l'}) | \leq 2\eps + |H_l(A_{l,l'}) - \mcl{D}(A_{l,l'}) | \\ \leq 2\eps + |  X(A_{l,l'}) - \mcl{D}(A_{l,l'}) | + |  X(A_{l,l'}) - \wh{H_l}(A_{l,l'}) | + | \wh{H_l}(A_{l,l'}) - H_l(A_{l,l'})| = O (\eps) \,.
\end{align*}
\end{proof}

We can now complete the proof of our main theorem.
\begin{proof}[Proof of Theorem \ref{thm:full-alg}]
Combining Lemma \ref{lem:list-learning}, Lemma \ref{lem:hypothesis-test} and Lemma \ref{lem:VCdim}, we immediately get the desired bound. 
\end{proof}

\section{Identifiability}\label{sec:identifiability}

Theorem \ref{thm:full-alg} implies that we can learn a mixture that is close to the true mixture in TV distance.  In order to prove that we recover the individual components, it suffices to prove identifiability.  In this section we prove the following.

\begin{theorem}\label{thm:identifiability}
Let $\mcl{M} = w_1G_1 + \dots + w_{k_1}G_{k_1}$ and $\mcl{M'} = w_1'G_1' + \dots + w_{k_2}'G_{k_2}'$ be mixtures of Gaussians such that $TV( \mcl{M}, \mcl{M'}) \leq \eps$ and the $G_i, G_i'$ have variance at least $\poly(\eps/d)$ and at most $\poly(d/\eps)$ in all directions.  Further assume,
\begin{itemize}
    \item $d_{\TV}(G_i, G_j) \geq b, d_{\TV}(G_i', G_j') \geq b$ for all $i \neq j$
    \item $w_i, w_i' \geq w_{\min} $
\end{itemize}
where $w_{\min} \geq f(k)$, $b \geq \eps^{f(k)}$ where $ k = \max(k_1,k_2)$ and $f(k)> 0$ is sufficiently small function depending only on $k$.  Then $k_1 = k_2$ and there exists a permutation $\pi$ such that 
\[
|w_i - w_{\pi(i)}'| + d_{\TV}(G_i, G_{\pi(i)}') \leq \poly(\eps)
\]
\end{theorem}

While technically, we do not need to prove identifiability in an algorithmic manner, our proof will mirror our main algorithm.  We will first prove identifiability in the case where the two mixtures are $\delta$-well conditioned for $\delta = \poly(\eps)$.

\begin{lemma}\label{lem:close-identifiability}
Let $\mcl{M} = w_1G_1 + \dots + w_{k_1}G_{k_1}$ and $\mcl{M'} = w_1'G_1' + \dots + w_{k_2}'G_{k_2}'$ be two $\delta$-well conditioned mixtures of Gaussians such that $d_{\TV}( \mcl{M}, \mcl{M'}) \leq \eps$ and $ \delta \geq \eps^{f(k)}$ where $ k = \max(k_1,k_2)$ and $f(k)> 0$ is sufficiently small function depending only on $k$.  Then $k_1 = k_2$ and there exists a permutation $\pi$ such that 
\[
|w_i - w_{\pi(i)}'| + d_{\TV}(G_i, G_{\pi(i)}') \leq \poly(\eps)
\]
\end{lemma}
\begin{proof}
Let $\mu, \Sigma,  \mu', \Sigma'$ be the means and covariances of the mixtures $\mcl{M}$ and $\mcl{M'}$.  Let $\mu_i, \Sigma_i, \mu_i', \Sigma_i'$ be the means and covariances of the respective components.  Without loss of generality we may assume $\mu = 0$, $\Sigma = I$.  The results in Section \ref{sec:moment-est}, namely Corollary \ref{corollary:isotropic-position}, imply that 
\begin{align*}
\norm{I - \Sigma'} = \poly(\eps) \\
\norm{\mu'} = \poly(\eps)
\end{align*}
This is because we can simulate an $\eps$-corrupted sample from $\mcl{M'}$ by just sampling from $\mcl{M}$ (since $d_{\TV}(\mcl{M}, \mcl{M'}) \leq \eps$) and then robustly estimate the mean and covariance of this sample.  Thus, by Corollary \ref{corollary-indiv-closeness}, we have for all $i$,
\begin{align*}
    \norm{\mu_i}, \norm{\mu_i'} &\leq \poly(\delta)^{-1} \\
    \norm{\Sigma_i - I}, \norm{\Sigma_i' - I} &\leq \poly(\delta)^{-1}
\end{align*}
Now, we can use Lemma \ref{lemma:estimate-hermite} to estimate the Hermite polynomials of the mixtures $\mcl{M}, \mcl{M'}$.  Since we can robustly estimate the means of bounded-covariance distributions (see Theorem 2.2 in \cite{kane2020robust}, Lemma \ref{lemma:estimate-hermite}), we must have
\[
\norm{v\left(h_{m, \mcl{M}}(X) - h_{m, \mcl{M'}}(X)\right)}_2  \leq \poly(\eps)
\]
Also note that since each of the mixtures is $\delta$-well conditioned, using Claim \ref{claim:parameter-dist} and Lemma \ref{lem:parameter-closeness} implies that
\[
\norm{\mu_i - \mu_j}_2 + \norm{\Sigma_i - \Sigma_j}_2 \geq \poly(\delta)
\]
and similar for the components of the mixture $\mcl{M}'$.  Repeating the argument in Section \ref{sec:reduction}, it suffices to prove the lemma in the case when all pairs of parameters are separated or equal i.e. among the sets $\{ \mu_i \} \cup \{ \mu_i' \}$ and $\{ \Sigma_i \} \cup \{\Sigma_i' \}$, each pair of parameters is either equal or separated by at least $\poly(\delta)$.  If we prove this. we can then deduce the statement of the lemma in the general case with worse, but still polynomial dependencies on $\eps$.
\\\\
Now we consider the generating functions 
\begin{align*}
F &= \sum_{i=1}^{k_1} w_ie^{\mu_i(X)y + \frac{1}{2}\Sigma_i(X)y^2} = \sum_{m=0}^{\infty}\frac{1}{m!}h_{m, \mcl{M}}(X)y^n \\
F'&= \sum_{i=1}^{k_2} w_i'e^{\mu_i'(X)y + \frac{1}{2}\Sigma_i'(X)y^2} = \sum_{m=0}^{\infty}\frac{1}{m!}h_{m, \mcl{M}'}(X)y^n
\end{align*}
where similar to in Section \ref{sec:close-case}, $\mu_i(X) = \mu_i \cdot X, \Sigma_i(X) = X^T\Sigma_iX$.  Consider the pair $(\mu_{k_2}',\Sigma_{k_2}')$.  We claim that there must be some $i$ such that 
\[
(\mu_i, \Sigma_i) = (\mu_{k_2}',\Sigma_{k_2}')
\]
Assume for the sake of contradiction that this is not the case.  Let $S_1$ be the subset of $[k_1]$ such that $\Sigma_i = \Sigma_{k_2'}$ and let $S_2$ be the subset of $[k_2 - 1]$ such that $\Sigma_j' = \Sigma_{k_2}'$.  Define the differential operators 
\begin{align*}
    \mcl{D}_i = \partial - \mu_i(X) - \Sigma_i(X)y\\
    \mcl{D}_i' = \partial - \mu_i'(X) - \Sigma_i'(X)y
\end{align*}
where as before, partial derivatives are taken with respect to $y$.  Now consider the differential operator
\[
\mcl{D} = \left(\mcl{D}'_{k_2 - 1} \right)^{2^{k_1 + k_2 - 2}}\left(\mcl{D}'_1 \right)^{2^{k_1}}\mcl{D}_{k_1}^{2^{k_1 - 1}}\mcl{D}_1^{1}
\]
Note by Claim \ref{claim:degree-reduction}, $\mcl{D}(F) = 0$.  Using Claim \ref{claim:leading-coeff}, 
\[
\mcl{D}(F') = P(y,X)e^{\mu_{k_2}'(X)y + \frac{1}{2}\Sigma_{k_2}'(X)y^2}
\]
where $P$ is a polynomial of degree
\[
\deg(P) = 2^{k_1 + k_2 - 1} - 1 - \sum_{i \in S_1}2^{i - 1} - \sum_{i \in S_2} 2^{k_1 + i - 2}
\]
and has leading coefficient 
\[
C_0 = w_{k_2}' \prod_{i \in [k_1]\backslash S_1}(\Sigma_{k_2}' - \Sigma_i)^{2^{i-1}} \prod_{i \in S_1}(\mu_{k_2}' - \mu_i)^{2^{i-1}}\prod_{i \in [k_2 - 1]\backslash S_2}(\Sigma_{k_2}' - \Sigma_i')^{2^{k_1 + i-2}} \prod_{i \in S_2}(\mu_{k_2}' - \mu_i')^{2^{k_1 + i-2}}.
\]
If there is no $i$ such that $(\mu_i, \Sigma_i) = (\mu_{k_2}',\Sigma_{k_2}')$ then
\[
C_0 \geq \delta^{O_k(1)}
\]
We can now compare 
\begin{align*}
\left(\mcl{D}_{k_2}'\right)^{\deg(P)}\mcl{D}(F) \\
\left(\mcl{D}_{k_2}'\right)^{\deg(P)}\mcl{D}(F')
\end{align*}
evaluated at $y = 0$.  The first quantitiy is $0$ because $\mcl{D}(F)$ is identically $0$ as a formal power series.  The second expression is equal to $\Omega_k(1)C_0$.  However, the coefficients of the formal power series $F,F'$ are the Hermite polynomials $h_{m, \mcl{M}}(X)$ and $h_{m, \mcl{M}'}(X)$.  We assumed that 
\[
\norm{v\left(h_{m, \mcl{M}}(X) - h_{m, \mcl{M'}}(X)\right)}_2  \leq \poly(\eps)
\]
so this is a contradiction as long as $\eps$ is smaller than $\delta^{F(k)}$ for some sufficiently large function $F$ depending only on $k$.  Thus, there must be some component of the mixture $\mcl{M}$ that matches each component of $\mcl{M'}$.  We can then repeat the argument in reverse to conclude that $\mcl{M}$ and $\mcl{M'}$ have the same components.  Finally, assume that we have two mixtures $\mcl{M} = w_1G_1 + \dots + w_kG_k$ and $\mcl{M'} = w_1'G_1 + \dots + w_k'G_k$ on the same set of components.  WLOG 
\[
w_1 - w_1' < \dots < w_l - w_l' < 0 < w_{l+1} - w_{l+1}' < \dots < w_k - w_k'
\]
Then we can consider
\begin{align*}
(w_1' - w_1)G_1 + \dots + (w_l' - w_l)G_l \text{ and }\\
(w_{l+1} - w_{l+1}')G_{l+1} + \dots + (w_k - w_k')G_k
\end{align*}
each treated as a mixture.  If 
\[
\sum_{i=1}^k \abs{w_i - w_i'} > \eps^{\zeta}
\]
for some sufficiently small $\zeta$ depending only on $k$, we can then normalize each of the above into a mixture (i.e. make the mixing weights sum to $1$) and repeat the same argument, using the fact that pairs of components cannot be too close, to obtain a contradiction.  Thus, actually the components and mixing weights of the two mixtures must be $\poly(\eps)$-close and this completes the proof.
\end{proof}

To complete the proof of Theorem \ref{thm:identifiability}, we will prove a sort of cluster identifiability that mirrors our algorithm and then combine with Lemma \ref{lem:close-identifiability}.

\begin{proof}[Proof of Theorem \ref{thm:identifiability}]
Let $c$ be a sufficiently small constant that we will set later.  We apply Claim \ref{claim:exists-good-partition} on the mixture $\mcl{M}$ with parameter $c$  to find a partition $R_1, \dots , R_l$.    Let $\kappa$ be the parameter obtained from the statement of  Claim \ref{claim:exists-good-partition} i.e. $\kappa$ depends on $k$ and $c$.  First, we claim that each of the components $G_1', \dots , G_{k_2}'$ must be essentially contained within one of the clusters.  To see this, for each $j \in [k_2]$ there must be some $i$ such that 
\[
d_{\TV}(G_i, G_j') \leq 1 - \frac{w_{\min}}{2k} \leq 1 - \Omega_k(1)
\]
without loss of generality $i \in R_1$.  Then by Lemma \ref{lem:distance-between-gaussians}, for all $a \notin R_1$, 
\[
d_{\TV}(G_a, G_j') \geq 1 - \poly(\eps^{\kappa})
\]
where the polynomial may depend on $k$ \emph{but does not depend on $c$}.  The above implies that we can match each of the components $G_1', \dots , G_{k_2}'$ uniquely to one of the clusters $R_1, \dots , R_l$ where it has constant overlap with $\cup_{i \in R_j}G_i$.  Let $S_1$ be the subset of $[k_2]$ corresponding to the components among $G_1', \dots , G_{k_2}'$ that are matched to $R_1$.  Consider the mixtures
\begin{align*}
\mcl{M}_1 = \frac{\sum_{i \in R_1} w_iG_i}{\sum_{i \in R_1} w_i }\\
\mcl{M}_1' = \frac{\sum_{i \in S_1} w_i'G_i'}{\sum_{i \in S_1}w_i'}
\end{align*}
The above (combined with our assumed lower bound on the minimum mixing weight) implies that
\[
d_{\TV}(\mcl{M}_1, \mcl{M}_1') \leq  \poly(\eps^{\kappa})
\]
where again the polynomial may depend on $k$ but not $c$.  Now if we choose $c$ sufficiently small, we can apply Lemma \ref{lem:close-identifiability} to deduce that the components and mixing weights of $\mcl{M}_1, \mcl{M}_1' $ must be close.  We can then repeat this argument for all of the clusters $R_1, \dots , R_l$ to complete the proof.
\end{proof}

Combing Theorem \ref{thm:full-alg} and Theorem \ref{thm:identifiability}. we have
\begin{theorem}\label{thm:parameter-learning}
Let $k, A, b > 0$ be constants.  There is a sufficiently large function $G$ and a sufficiently small function $g$ depending only on $k,A,b$ (with $G(k,A,b),g(k,A,b) > 0$) such that given an $\eps$-corrupted sample $X_1, \dots , X_n$ from a mixture of Gaussians $\mcl{M} = w_1G_1 + \dots + w_kG_k \in \R^d$ where the $G_i$ have variance at least $\poly(\eps/d)$ and at most $\poly(d/\eps)$ in all directions and 
\begin{itemize}
    \item The $w_i$ are all rational with denominator at most $A$
    \item $d_{\TV}(G_i, G_j) \geq b$
\end{itemize}
and $n \geq (d/\eps)^{G(k,A,b)}$, then there is an algorithm that runs in time $\poly(n)$ and with $0.99$ probability outputs a set of components $\widetilde{G_1}, \dots , \widetilde{G_k}$ and mixing weights $\widetilde{w_1}, \dots ,\widetilde{w_k}$ such that there exists a permutation $\pi$ on $[k]$ with
\[
|w_i - \widetilde{w}_{\pi(i)}| + d_{\TV}(G_i, \widetilde{G}_{\pi(i)}) \leq  \eps^{g(k,A,b)}
\]
for all $i$.
\end{theorem}

\subsection{Improving the Separation Assumption}\label{sec:improved-separation}

With simple modifications to the analysis, we obtain the following improvement of Theorem \ref{thm:full-alg} in \cite{liu2021learning}.  
\begin{theorem}[\cite{liu2021learning}]\label{thm:improved-mixture-learning-cite}
Let $k, A > 0$ be constants.  There is a sufficiently large function $G$ and a sufficiently small function $g$ depending only on $k,A$ such that given an $\eps$-corrupted sample $X_1, \dots , X_n$ from a mixture of Gaussians $\mcl{M} = w_1G_1 + \dots + w_kG_k \in \R^d$ where $\eps < g(k,A)$, the $w_i$ are all rational with denominator at most $A$, and $n \geq (d/\eps)^{G(k,A)}$, there is an algorithm that runs in time $\poly(n)$ and with $0.999$ probability, outputs a set of $(1/\eps)^{O_{k,A}(1)}$ candidate mixtures such that for at least one of these candidates,  $\{\widetilde{w_1}, \wt{G_1}, \dots , \widetilde{w_k}, \wt{G_k} \}$, we have
\[
|w_i - \widetilde{w_i}| + d_{\TV}(G_i, \wt{G_i}) \leq \eps^{g(k,A)}
\]
for all $i \in [k]$.
\end{theorem}
To go from Theorem \ref{thm:full-alg} to Theorem \ref{thm:improved-mixture-learning-cite}, the main idea to remove the constant separation assumption is just that we can find a scale $\delta$ such that all pairs of components either have $d_{\TV}(G_i, G_j) \leq \delta$ or $ d_{\TV}(G_i, G_j) \geq \delta'$ for $\delta' \gg \delta$.  This is possible because the number of components $k$ is a constant.  We can then merge components whose TV distance is less than $\delta$, treating them as the same component and the remaining components will be sufficiently separated.  See \cite{liu2021learning} for more details.

The results of \cite{liu2021learning} were obtained using the previous clustering subroutine of \cite{diakonikolas2020robustly}.  If we instead plug in the updated clustering results of \cite{bakshi2020robustlya} (see Theorem \ref{thm:improved-clustering} vs Theorem \ref{thm:rough-clustering}), we can remove the bounded fractionality assumption on the mixing weights.  Also, we can ensure that the algorithm outputs a unique mixture instead of a list by running the hypothesis testing routine in Lemma \ref{lem:hypothesis-test}.

\begin{theorem}\label{thm:improved-mixture-learning}
Let $k, A > 0$ be constants.  There is a sufficiently large function $G$ and a sufficiently small function $g$ depending only on $k,A$ such that given an $\eps$-corrupted sample $X_1, \dots , X_n$ from a mixture of Gaussians $\mcl{M} = w_1G_1 + \dots + w_kG_k \in \R^d$ where $\eps < g(k,A)$, the $w_i$ are all at least $1/A$, and $n \geq (d/\eps)^{G(k,A)}$, there is an algorithm that runs in time $\poly(n)$ and with $0.999$ probability, outputs a mixture $\wt{M} = \widetilde{w_1}\wt{G_1} + \dots + \widetilde{w_k}, \wt{G_k} $, such that
\[
d_{\TV}(\mcl{M}, \wt{\mcl{M}}) \leq \eps^{g(k,A)} \,.
\]
\end{theorem}

Finally, combining the above with identifiability (Theorem \ref{thm:identifiability}), we immediately get an improved version of our main theorem for parameter learning.
\begin{theorem}\label{thm:main-improved}
Let $k, A > 0$ be constants.  There is a sufficiently large function $G$ and a sufficiently small function $g$ depending only on $k,A$ (with $G(k,A),g(k,A) > 0$) such that given an $\eps$-corrupted sample $X_1, \dots , X_n$ from a mixture of Gaussians $\mcl{M} = w_1G_1 + \dots + w_kG_k \in \R^d$ where the $G_i$ have variance at least $\poly(\eps/d)$ and at most $\poly(d/\eps)$ in all directions and 
\begin{itemize}
    \item The $w_i$ are all at least $1/A$
    \item $d_{\TV}(G_i, G_j) \geq \eps^{g(k,A)}$
\end{itemize}
and $n \geq (d/\eps)^{G(k,A)}$, then there is an algorithm that runs in time $\poly(n)$ and with $0.99$ probability outputs a set of components $\widetilde{G_1}, \dots , \widetilde{G_k}$ and mixing weights $\widetilde{w_1}, \dots ,\widetilde{w_k}$ such that there exists a permutation $\pi$ on $[k]$ with
\[
|w_i - \widetilde{w}_{\pi(i)}| + d_{\TV}(G_i, \widetilde{G}_{\pi(i)}) \leq  \eps^{g(k,A)}
\]
for all $i$.
\end{theorem}

\bibliographystyle{alpha}
\bibliography{bibliography}
\end{document}